\documentclass[12pt]{article}
\usepackage[utf8]{inputenc}
\usepackage[english]{babel}
\usepackage{amsmath}
\usepackage{mathrsfs}
\usepackage{amssymb}
\usepackage{amsthm}
\newtheorem{lemma}{Lemma}
\usepackage{bm}
\usepackage{amsfonts}
\usepackage{xspace}
\usepackage{setspace}
\usepackage[normalem]{ulem}
\usepackage{graphicx}
\usepackage{multirow}
\usepackage{booktabs}
\usepackage{appendix}
\usepackage{enumerate}
\usepackage{url} 
\usepackage{authblk}
\usepackage{xcolor}
\usepackage{subcaption}
\usepackage{xr}
\usepackage{adjustbox}
\usepackage{tabularx}
\usepackage{tikz}
\usepackage[colorlinks=true,linkcolor=red,citecolor=blue]{hyperref}
\usepackage{enumitem}
\usepackage{natbib}
\usepackage{iftex}
\usepackage{mathtools}
\usepackage{etoolbox}
\usepackage{algorithm}
\usepackage{algpseudocode}
\usetikzlibrary{shapes, shapes.geometric, arrows.meta, positioning, calc, decorations.pathreplacing, backgrounds, fit}

\definecolor{myblue}{RGB}{70, 130, 180}
\definecolor{myorange}{RGB}{255, 140, 0}
\definecolor{lightbg}{RGB}{248, 250, 255}
\newcommand{\blind}{1}

\makeatletter
\newcommand*{\addFileDependency}[1]{
\typeout{(#1)}
\@addtofilelist{#1}
\IfFileExists{#1}{}{\typeout{No file #1.}}
}\makeatother

\newcommand{\mP}{\mathbb{P}}
\newcommand{\mG}{\mathbb{G}}

\newcommand{\bbeta}{\boldsymbol{\beta}}

\newcommand{\bGamma}{\boldsymbol{\Gamma}}

\newcommand{\bP}{\mathbf{P}}
\newcommand{\bX}{\mathbf{X}}
\newcommand{\bb}{\mathbf{b}}
\newcommand{\bC}{\mathbf{C}}

\newcommand{\mD}{\mathcal{D}}

\newcommand{\bx}{\mathbf{x}}

\newcommand{\bH}{\mathbf{H}}

\newcommand{\bh}{\mathbf{h}}

\newcommand{\mO}{\mathcal{O}}

\newcommand{\calX}{\mathcal{X}}

\newcommand{\lb}{\left\{}
\newcommand{\rb}{\right\}}
\newcommand{\lessp}{\lesssim_{\mP}}

\newcommand{\whp}{\widehat{p}}

\newtheorem{assumption}{Assumption}%
\newtheorem{remark}{Remark}%
\newtheorem{theorem}{Theorem}%

\newtheorem{corollary}{Corollary}

\makeatletter
\begin{document}

\def\spacingset#1{\renewcommand{\baselinestretch}%
{#1}\small\normalsize} \spacingset{1}


\if1\blind
{
  \title{\bf Doubly cross-fit debiased machine learning of heterogeneous treatment effects under principal stratification}
\author{
  Jiaqi Tong$^{1}$
  and Fan Li$^{1,*}$
  \vspace{0.2cm} 
  \\
  $^{1}$Department of Biostatistics, Yale School of Public Health, New\\ Haven, CT, USA \\
  $^{*}$\emph{email}: fan.f.li@yale.edu
}

  \maketitle
} \fi

\if0\blind
{
  \bigskip
  \bigskip
  \bigskip
  \begin{center}
    {\LARGE\bf Learning heterogeneous treatment effects under principal stratification}
\end{center}
  \medskip
} \fi

\bigskip
\begin{abstract}
Principal stratification provides a foundational framework for causal inference with intermediate outcomes by defining causal effects within subpopulations, yet existing work has largely focused on average effects across strata rather than treatment effect heterogeneity within strata. Such within-stratum heterogeneity informs individualized treatment decisions but the associated methods are sparse. We address this gap by studying the identification and estimation of the conditional principal causal effects under principal ignorability combined with an odds ratio sensitivity parameterization, which relaxes the monotonicity assumption. To efficiently learn these estimands, we propose a novel doubly cross-fit doubly robust machine learner that resolves the nested nuisance structure inherent to principal stratification. Leveraging sequential orthogonal debiased machine learning with regularized least-squares sieves, we derive $\mathcal{L}^2$ and uniform limit theory, establish oracle efficiency, and construct uniform confidence bands for the proposed estimator. We use simulations to demonstrate the finite-sample performance of our estimator, and provide an empirical analysis of a randomized trial in acute lung injury, revealing informative patterns of treatment effect heterogeneity within the always-survivor subpopulation.
\end{abstract}

\noindent%
{\it Keywords:} Debiased machine learning, double cross-fitting, efficient influence function, oracle efficiency, principal stratification, uniform confidence bands
\vfill

\newpage
\spacingset{1.8} 

\section{Introduction}

\subsection{Background and related literature}

Principal stratification is a general framework for causal inference with intermediate outcomes that lie on the pathway between assignment and outcome \citep{frangakis2002principal}; it has received widespread attention in statistical science for addressing noncompliance \citep{angrist1994}, truncation by death \citep{zhang2003estimation}, {principal surrogate evaluation \citep{gilbert2008evaluating}}, and related problems. A canonical example is truncation by death \citep{zhang2003estimation}, in which case the quality of life outcome becomes ambiguously defined for individuals who die before assessment. Comparisons of observed outcomes across groups conflate treatment effects with differences in survival, thereby compromising causal interpretation. Principal stratification addresses this challenge by targeting principal causal effects (PCEs), well-defined average causal effect estimands within subpopulations (principal strata) characterized by the joint potential values of the intermediate outcome under counterfactual conditions.

A substantial literature has emerged regarding the nonparametric point identification of PCEs \citep[e.g.,][]{DingandLu2016,JiangJRSSB2022,tong2025semiparametric}. Comparisons among PCEs can reveal treatment effect variation arising from \emph{endogenous} intermediate outcomes \citep{page2015principal}. In contrast, relatively little attention has been devoted to characterizing treatment effect variation within principal strata, that is, how treatment effects may potentially vary as a function of \emph{exogenous} baseline covariates among individuals belonging to a given principal stratum (such as always-survivors). This distinction is important: heterogeneity across principal strata reflects differences in causal effects between latent subpopulations, whereas heterogeneity within principal strata captures clinically and scientifically meaningful effect modification that remains after conditioning on joint potential intermediate outcomes. Developing statistical methods to learn such within-stratum heterogeneity remains a largely open problem.

Several related efforts have examined treatment effect heterogeneity within a certain principal stratum. Under the instrumental variable (IV) framework, \cite{syrgkanis2019machine} studied estimation of the conditional local average causal effect (CLATE) among compliers using a two-stage orthogonal learning approach, while \cite{johnson2022detecting} proposed a matching estimator and developed data-adaptive tests for discovering effect modifiers among compliers. Furthermore, \cite{takatsu2025doubly} proposed a Wald-type ratio estimator for CLATE, where the numerator and denominator are estimated separately using the data-adaptive doubly robust (DR) learner \citep{kennedy2023towards}. Alternatively, \cite{bargagli2022heterogeneous} proposed a Bayesian Causal Forest approach and \cite{spanbauer2024flexible} employed Bayesian additive regression trees (BARTs) to flexibly estimate CLATE estimands but did not establish the asymptotic properties of these estimators. Notably, all of these methods depend on the exclusion restriction and monotonicity assumptions, which are inherent to the IV definition and restrict attention to the complier principal stratum. This limits their applicability to other types of intermediate variables. In the context of truncation by death, \citet{chen2024bayesian} developed a mixture of BARTs approach to estimate the conditional survivor average causal effect (CSACE), but this approach hinges on monotonicity and parametric distributional assumptions. {Recently, under monotonicity and principal ignorability, \cite{zhang2026estimation} proposed several flexible causal machine learners for the conditional PCE among always-takers, compliers, and never-takers. However, conceptually, this study focuses primarily on estimating the conditional PCE given the full set of exogenous baseline covariates, which may be moderate to high dimensional and can therefore complicate estimation and interpretation. Inferentially, this work relies on computationally intensive bootstrap procedures for pointwise inference and does not provide a procedure for simultaneous inference.} More broadly, general data-adaptive methods that are robust to nuisance estimation errors, allow treatment effect heterogeneity within all principal strata, accommodate diverse intermediate variables, and provide strong large-sample theoretical guarantees, such as oracle efficiency and uniform inference, remain underdeveloped.

\subsection{Our contributions}
The central contribution of this article is a flexible debiased machine learning framework for estimating a new class of estimands, referred to as \emph{conditional principal causal effects}, which characterize treatment effect variation within principal strata across a low-dimensional, interpretable subset of covariates. First, we propose a novel doubly cross-fit doubly robust (DCDR) causal machine learner that targets conditional principal causal effects identified under an odds ratio sensitivity parameter and principal ignorability assumption; this allows for estimating heterogeneous treatment effects within all four principal strata and nests popular assumptions such as monotonicity and counterfactual intermediate independence as two special cases \citep{tong2025semiparametric}.
To address unique challenges when conditioning on a subset of covariates within principal strata---in which case certain nuisance functions are themselves functions of the remaining nuisance functions, we 
derive an orthogonal learner to avoid linear nuisance error propagation. Our orthogonal learner is robust to nuisance functions because it is the solution to a single estimating equation, obtained by taking the directional derivative of a Neyman-orthogonal loss function with respect to the target estimand. Our approach is thus insensitive to the perturbation of nuisance functions around their true values \citep{foster2023orthogonal}. This is in contrast to \citet{takatsu2025doubly} (developed under the IV framework) who focused on a ratio estimator by separately learning the numerator and denominator. For operationalization, we combine the double cross-fitting scheme in \cite{newey2018cross} with the proposed orthogonal learner to properly handle nested nuisance functions and to ensure faster convergence rates. 
Second, we construct computationally efficient uniform confidence bands and provide a formal treatment of the asymptotic theory of the proposed orthogonal DCDR learner, including its $\mathcal{L}^2$ and uniform limit theory, when the final-stage pseudo-outcome regression is based on the regularized least-squares sieve method \citep{belloni2015some}. We show that under suitable regularity conditions, our proposed estimator is oracle efficient and the uniform confidence bands are asymptotically honest. Compared to \cite{semenova2021debiased} (developed without an intermediate variable), we leverage regularization as a useful approach to select hyperparameters and improve finite-sample inference. {Our work is related to the concurrent work of \cite{zhang2026estimation}, but differs substantially in several dimensions, including target estimand, identification assumptions and formulas, statistical challenges, estimation and inference procedures, and asymptotic analysis. A detailed comparison of our work with \cite{zhang2026estimation} is provided in Supplementary Material Table \ref{supp;tab:comparison-zhang}.}


The remainder of this article is organized as follows. Section \ref{sec:notation-assumptions-identification} introduces estimands, assumptions, and nonparametric identification formulas. Section \ref{sec:dcf-tr-learning} proposes our DCDR learning approach for estimation and inference. Section \ref{sec:large-sample-theory-inference-bands} offers a formal treatment of the asymptotic theory for the proposed causal effect estimator. Section \ref{sec:sim_exp_main} reports the results of simulation experiments. Section \ref{s:Data_Example} exemplifies the proposed approach with an empirical application, and Section \ref{sec:conclusion} concludes.

\section{Notation, assumptions, and identification}\label{sec:notation-assumptions-identification}
We consider a setting with a binary treatment $Z \in \{0,1\}$, where $Z=1$ denotes treatment and $Z=0$ denotes control, a vector of baseline covariates $\bC$, a binary intermediate outcome $D \in \{0,1\}$, and a final unconstrained outcome $Y$. We observe $n$ independent and identically distributed copy of the data vector $\mO \coloneqq (Y, D, Z, \bC)$, specified through the underlying full data vector $\mD \coloneqq (Y(0), Y(1), D(0), D(1), Z, \bC)$, where $Y(z)$ and $D(z)$ denote the potential values of the final outcome and intermediate outcome under assignment $Z=z$. Under the Stable Unit Treatment Value Assumption, we have $Y = Y(Z)$ and $D = D(Z)$. In principal stratification analysis, the pair of joint potential intermediate outcomes, $G=(D(0), D(1))$, is unaffected by assignment and thus regarded as a pre-treatment covariate \citep{frangakis2002principal}. We define the principal strata as $G=d_0d_1$ when $D(0)=d_0$ and $D(1)=d_1$ with $d_0,d_1\in\{0,1\}$. In the context of noncompliance ($D$ is the treatment receipt), the strata $G=\{01,11,10,00\}$ correspond to compliers, always-takers, defiers, and never-takers \citep{angrist1994}; in the context of truncation by death ($D$ is the binary survival status prior to outcome measurement), these strata correspond to the subpopulation of the protected, always-survivors, harmed, and never-survivors \citep{rubin2006causal}. 

The canonical target of inference under principal stratification is the principal causal effect (PCE), defined as $\tau_{d_0d_1}=E\{Y(1)-Y(0) | G= d_0d_1\}$. While nonparametric point identification of the PCE has been previously investigated, for example, by \citet{DingandLu2016}, \citet{JiangJRSSB2022}, and \citet{tong2025semiparametric}, little attention has been given to exploring treatment effect heterogeneity within each stratum. In this article, we define a finer class of estimands, referred to as conditional principal causal effect (CPCE), by 
\begin{align*}
    \tau_{d_0d_1}(\bX)\coloneqq E\{Y(1)-Y(0)|G=d_0d_1,\bX\},~~d_0,d_1\in\{0,1\}, 
\end{align*}
where $\bX \in \mathbb{R}^p$ is a potentially low-dimensional function of $\bC$, chosen to represent the set of potential effect modifiers reflecting scientific interest. 
Estimand $\tau_{d_0d_1}(\bX)$ denotes the average treatment effect conditional on $\bX$ within the principal stratum $d_0d_1$, and the PCE is a weighted summary of CPCE within the principal stratum, i.e., $\tau_{d_0d_1}=\int \tau_{d_0d_1}(\bX) dF(\bX|G=d_0d_1)$. When $D$ denotes the treatment receipt, $\tau_{01}(\bX)$ is the CLATE studied, for example, in \citet{abrevaya2015estimating}, \citet{syrgkanis2019machine}, \citet{johnson2022detecting} and \citet{bargagli2022heterogeneous} under an IV framework. When $D$ denotes the survival status, $\tau_{11}(\bX)$ is the CSACE studied in \citet{chen2024bayesian} under a Gaussian mixture modeling framework. We invoke the following assumptions to achieve point identification.

\begin{assumption}[\emph{Treatment Ignorability}]\label{assump:treatment-ignorability}
    $\{Y(0),Y(1),D(0),D(1)\}\perp Z|\bC.$
\end{assumption}
\begin{assumption}[\emph{Mean Principal Ignorability}]\label{assump:PI_weak} (2.a) $E\{Y(1)|G=11,\bC\}=E\{Y(1)|G=01,\bC\allowbreak\}$; (2.b) $E\{Y(1)|G=10,\bC\}=E\{Y(1)|G=00,\bC\}$; (2.c) $E\{Y(0)|G=11,\bC\}=E\{Y(0)\allowbreak|G=10,\bC\}$; (2.d) $E\{Y(0)|G=01,\bC\}=E\{Y(0)|G=00,\bC\}$. 
\end{assumption}

Assumption \ref{assump:treatment-ignorability} states that the treatment is as if randomized given all pre-treatment covariates. Assumption \ref{assump:PI_weak} is a typical condition required for point identifying PCE \citep{hayden2005estimator,DingandLu2016} and states that the mean of $Y(z)$ is independent of the cross-world intermediate potential outcome $D(1-z)$, conditional on $\{D(z) = d_z, \bC\}$. Under Assumptions \ref{assump:treatment-ignorability}-\ref{assump:PI_weak}, the finest CPCE given full baseline covariates is identifiable as
\begin{align}\label{eq:iden-cpce-allc=x}
    \tau_{d_0d_1}(\bC)=m_{1d_1}(\bC)-m_{0d_0}(\bC),
\end{align}
where $m_{zd}(\bC) := E\{Y | Z = z, D = d, \bC\}$ is the conditional mean outcome given $(Z, D, \bC) = (z, d, \bC)$. Interestingly, under mean principal ignorability, it is sufficient to adjust for the observed intermediate outcome for estimating the finest CPCE. {Recently, for estimating $\tau_{d_0d_1}(\bC)$, \cite{zhang2026estimation} proposed a plug-in learner and two multiply robust learners: one based on an adaptation of the DR learner of \cite{kennedy2023towards} conditional on the observed intermediate outcome, and another motivated by the triply robust semiparametric one-step estimator of \cite{JiangJRSSB2022}.
}

Although $\tau_{d_0d_1}(\bC)$ represents the most granular estimand, investigators often prioritize a smaller, interpretable set of effect modifiers informed by subject-matter knowledge; in such cases, causal inference for CPCE estimands requires more nuanced considerations. Specifically, we assume that $\bX$ is of lower dimension, while the full set of measured covariates $\bC$ may be high-dimensional to ensure the plausibility of the ignorability assumptions. When $\bX$ is a proper subset of $\bC$ such that $\bX\subsetneq \bC$, identifying the CPCE additionally requires the identification of the principal score (PS), defined as the conditional probability of belonging to a given principal stratum, $e_{d_0d_1}(\bC)=\Pr(G=d_0d_1 | \bC)$, for $d_0,d_1\in\{0,1\}$. We focus on this more challenging problem hereafter. Since $G$ is not fully observable, we make the following sensitivity assumption to connect the PS with the observed data distribution. 

\begin{assumption}[\emph{Conditional odds ratio}]\label{assump:cor}
There exists a known function $\theta(\bullet):\mathcal{C}\to[0,\infty]$, such that $\{e_{11}(\bC)e_{00}(\bC)\}/\{e_{10}(\bC)e_{01}(\bC)\}=\theta(\bC)$, where $\mathcal{C}$ is the covariate support.
\end{assumption}
Assumption \ref{assump:cor} introduces an odds ratio sensitivity parameter, which is not identifiable from the observed data. This odds ratio parameterization offers two major advantages. First, it unifies monotonicity \citep{DingandLu2016} and counterfactual intermediate independence \citep{hayden2005estimator} as special limiting cases. That is, (i) monotonicity $D(1)\geq D(0)$ holds if and only if $\theta(\bC)\to\infty$ and $p_1(\bC)>p_0(\bC)$ with probability one, where $p_z(\bC)=\Pr(D=1|Z=z,\bC)$ has also been referred to as (an estimable) principal score \citep{JiangJRSSB2022}; (ii) and counterfactual intermediate independence holds if and only if $\theta(\bC)=1$. 
Second, unlike other parameterizations (e.g., the ratio of strata proportions in \cite{DingandLu2016}), the odds ratio sensitivity parameter is familiar to practitioners and is margin-free; that is, its range is not constrained by the unknown marginal distributions of $D(1)$ or $D(0)$ (cf. Remark 2 in \cite{tong2025semiparametric}).

Under Assumptions \ref{assump:treatment-ignorability}-\ref{assump:cor}, \citet{tong2025semiparametric} proved that the PS is identifiable, provided that $\Pr(\theta(\bC)-1\neq0)=1$, and is given by
\begin{align*}
    e_{d_0d_1}(\bC)=&\frac{(-1)^{d_0+d_1}}{2\{\theta(\bC)-1\}}\left[1+\{\theta(\bC)-1\}\kappa_{d_0d_1}(\bC)-\sqrt{\delta(\bC)}\right],
\end{align*}
where $\kappa_{d_0d_1}(\bC)=(2d_1-1)p_0(\bC)+(2d_0-1)p_1(\bC)+2(1-d_0)(1-d_1)$ and $\delta(\bC)=\left[1+\{\theta(\bC)-1\}\{p_0(\bC)+p_1(\bC)\}\right]^2-4\theta(\bC)\{\theta(\bC)-1\}p_0(\bC)p_1(\bC)$. Moreover, to accommodate the monotonicity assumption ($D(1)\geq D(0)$), one computes $\lim_{\theta\to\infty} e_{d_0d_1}(\bC)$ given $p_1(\bC)>p_0(\bC)$, which yields $e_{01}(\bC)=p_1(\bC)-p_0(\bC)$, $e_{00}(\bC)=1-p_1(\bC)$, and $e_{11}(\bC)=p_0(\bC)$; to accommodate counterfactual intermediate independence ($D(0)\perp D(1)|\bC$), one computes $\lim_{\theta\to1} e_{d_0d_1}(\bC)$, which yields $e_{11}(\bC)=p_0(\bC)p_1(\bC)$, $e_{01}(\bC)=p_1(\bC)\{1-p_0(\bC)\}$, $e_{00}(\bC)=\{1-p_0(\bC)\}\{1-p_1(\bC)\}$, and $e_{10}(\bC)=p_0(\bC)\{1-p_1(\bC)\}$.
Furthermore, under Assumptions \ref{assump:treatment-ignorability}-\ref{assump:cor}, we prove in Supplementary Material Section \ref{sec:derivation-identification-CPCE} that CPCE $\tau_{d_0d_1}(\bX)$ is also nonparametrically point identifiable based on the following g-computation formula:
\begin{align}\label{eq:identification-CPCE-subset}
    \tau_{d_0d_1}(\bX)=&\frac{E\{e_{d_0d_1}(\bC)(m_{1d_1}(\bC)-m_{0d_0}(\bC))|\bX\}}{E\{e_{d_0d_1}(\bC)|\bX\}}:=\frac{\tau^N_{d_0d_1}(\bX)}{\tau^D_{d_0d_1}(\bX)},
\end{align}
where $\bX\in\mathcal{X}:=\{\bX\subseteq\mathbb{R}^p:\Pr(\tau^D_{d_0d_1}(\bX)>0)=1\}$. Finally, setting $\bX = \bC$ as the full covariate vector in Equation \eqref{eq:identification-CPCE-subset} yields $\tau_{d_0d_1}(\bC) = m_{1d_1}(\bC) - m_{0d_0}(\bC)$, which recovers Equation \eqref{eq:iden-cpce-allc=x} as a special case. 

\section{Doubly robust learner with double cross-fitting}\label{sec:dcf-tr-learning}
\subsection{Motivation and overview}
To proceed, we define the propensity score as $\pi(\bC)=\Pr(Z=1|\bC)$ and the collection of nuisance functions as $\bGamma=(p_0,p_1,m_{0d_0},m_{1d_1},\pi,\tau^N_{d_0d_1},\tau^D_{d_0d_1})$. Further, define $\bGamma_1:=(p_0, p_1, \pi, m_{0d_0},m_{1d_1})$ as the \emph{base} nuisance functions. We refer to these as `base' nuisance functions because they are regression functions of observed data that can be estimated by existing data-adaptive machine learners. However, generalizing orthogonal learning to estimate the CPCE is complicated due to the \emph{intermediate} nuisance functions $\bGamma_2:=(\tau_{d_0d_1}^N,\tau_{d_0d_1}^D)$, which are themselves regression functions of $\bGamma_1$ and cannot be directly obtained without knowledge of the base nuisance functions. This leads to the following two questions. First, can one construct an optimal estimator for the intermediate nuisance $\bGamma_2$? Second, given reasonable estimates for $\bGamma_2$, can one construct an optimal estimator for learning CPCE that does not heavily rely on the quality of learners for $\bGamma_2$? Under the IV framework, \citet{takatsu2025doubly} proposed to use two separate DR learners for the counterparts of $\tau_{d_0d_1}^N$ and $\tau_{d_0d_1}^D$, and take the ratio to learn the final conditional estimands. This addresses the first challenge but not the second, because separately learning $\tau_{d_0d_1}^N$ and $\tau_{d_0d_1}^D$ using the same data may not guarantee the optimal bias-variance trade-off for the CPCE and can lead to suboptimal convergence rates. For example, when $\tau_{d_0d_1}^N$ and $\tau_{d_0d_1}^D$ are significantly more complex than the CPCE estimand $\tau_{d_0d_1}$, regularization for learning $\tau_{d_0d_1}^N$ or $\tau_{d_0d_1}^D$ may not be optimal for learning the simpler $\tau_{d_0d_1}$. Under the principal stratification framework without monotonicity, we address both challenges. By adapting the general idea of \citet{newey2018cross}, we propose a three-stage, doubly cross-fit pseudo-outcome regression estimator based on sequential orthogonal learning to ensure high-quality estimation of these intermediate nuisance functions while achieving the optimal bias-variance trade-off for the CPCE. This procedure applies an additional orthogonal learner (Stage 2) to estimate $\bGamma_2$, which necessitates an additional sample split within the training data. A schematic of our approach is presented in Figure \ref{fig:do_learner_schematic}.

\begin{figure}[ht!]
\centering
\begin{tikzpicture}[
    scale=0.7, 
    transform shape,
    font=\sffamily,
    >=Stealth,
    node distance=0.4cm and 0.5cm,
    thick,
    foldbox/.style={
        draw=myblue,
        very thick,
        rounded corners=3pt,
        minimum width=5.2cm, 
        align=center,
        fill=white,
        font=\footnotesize
    },
    outerfold/.style={
        foldbox,
        minimum height=0.4125cm,
        inner sep=2pt
    },
    innerfold/.style={
        foldbox,
        minimum height=0.6cm,
        font=\small
    },
    holdoutbox/.style={
        foldbox,
        draw=myorange,
        fill=myorange!10,
        inner sep=2pt
    },
    processbox/.style={
        draw=black!60,
        dashed,
        thick,
        fill=white,
        align=center,
        inner sep=6pt,
        rounded corners=2pt,
        text width=6.2cm 
    }
]

    \node[processbox] (Step1) {
        \textbf{Stage 1}\\
        Use $\mathcal{F}_{s}^c$ (Blue Outer Folds)\\
        to estimate base nuisance\\
        functions $\widehat{\boldsymbol{\Gamma}}_{1s}$.
    };

    \node[processbox, right=4.5cm of Step1] (Step2) {
        \textbf{Stage 2 Step 1}\\
        1. Train on $\mathcal{F}_{sr}^c$ (Blue Inner Folds)\\
        2. Predict on $\mathcal{F}_{sr}$ (Orange Inner Fold)\\
         $~~~~\Rightarrow\widehat{Den}, \widehat{Num}$
    };


    \node[outerfold, below=0.6cm of Step1] (OF1) {Outer Fold 1};
    \node[outerfold, below=0.15cm of OF1] (OF2) {Outer Fold 2};
    \node[outerfold, below=0.15cm of OF2] (OF3) {Outer Fold 3};
    \node[outerfold, below=0.15cm of OF3] (OF4) {Outer Fold 4};
    
    \node[holdoutbox, minimum height=0.4125cm, below=0.6cm of OF4] (OF5) {Outer Fold 5};

    \node[innerfold] (IF1) at (Step2 |- OF1.north) [anchor=north] {Inner Fold 1};
    \node[innerfold, below=0.15cm of IF1] (IF2) {Inner Fold 2};
    \node[holdoutbox, dashed, minimum height=0.6cm, below=0.15cm of IF2 , font=\small] (IF3) {Inner Fold 3};

    \path (OF1.north east) -- (OF4.south east) coordinate[midway, xshift=4pt] (OuterBraceRight);
    \path (IF1.north west) -- (IF3.south west) coordinate[midway, xshift=-4pt] (InnerBraceLeft);
    
    \draw[->, ultra thick, myblue] (OuterBraceRight) -- (InnerBraceLeft) node[midway, above, align=center, font=\bfseries, text=black] {Partition $\mathcal{F}_s^c$ into \\ $R=3$  Inner Folds};

    \draw[->, dashed] (Step1.south) -- (OF1.north);
    \draw[->, dashed] (Step2.south) -- (IF1.north);


    \node[processbox, below=0.5cm of OF5] (Step4) {
        \textbf{Stage 3 Step 1}\\
        Compute $\widehat{Y}_i^*$ for $i \in \mathcal{F}_s$\\
        Input 1: $\widehat{\boldsymbol{\Gamma}}_{1s}$ (from Stage 1)\\
        Input 2: $\widehat{\boldsymbol{\Gamma}}_{2s}$ (from Stage 2)
    };

    \node[processbox] (Step3) at (Step2 |- Step4) {
        \textbf{Stage 2 Step 2}\\
        Regress $\{\widehat{Den}_i, \widehat{Num}_i\}$ on $\mathbf{X}_i$
        for $i\in\mathcal{F}_s^c$
         $\Rightarrow\widehat{\boldsymbol{\Gamma}}_{2s}$.
    };

    \draw[->, myblue, thick] (IF3.south) -- (Step3.north);
    \draw[->, myorange, ultra thick] (OF5.south) -- (Step4.north);

    \draw[->, black, thick] (Step3.west) -- (Step4.east)
        node[midway, above] {$\widehat{\bGamma}_{2s}$};

    \draw[->, black, thick] (Step1.west) -- ++(-0.8,0) |- (Step4.west)
        node[pos=0.5, left] {$\widehat{\bGamma}_{1s}$};

    
    \coordinate (LeftEdge) at ($(Step1.west) + (-1.4cm, 0)$);

    \node[processbox, below=0.6cm of Step4, anchor=north, xshift=5.25cm, draw=black, very thick, text width=14cm] (Step5) {
        \textbf{Stage 3 Step 2}\\
        Regress pooled $\{\widehat{Y}_i^*\}_{i=1}^n$ on $\{\mathbf{X}_i\}_{i=1}^n$
        via penalized least-squares sieves.
    };
    \draw[->, ultra thick] (Step4.south) -- (Step4.south |- Step5.north);

    \begin{scope}[on background layer]
        \node[draw=blue!10, fill=blue!5, rounded corners=10pt, fit=(Step1) (Step2) (Step5) (OF1) (LeftEdge), inner sep=0.3cm] (Container) {};
    \end{scope}

\end{tikzpicture}
\caption{Schematic of the doubly cross-fit pseudo-outcome regression procedure via sequential orthogonal learning, assuming $S=5,R=3$. Here, pseudo-outcomes for $\tau_{d_0d_1}^D$ and $\tau_{d_0d_1}^N$ are denoted as $Den(\bGamma_1) := \varphi_{d_0d_1}^D$ and $Num(\bGamma_1) := \varphi_{d_0d_1}^N$; $\mathcal{F}_s$ and $\mathcal{F}_{sr}$ denote outer and inner testing samples, respectively, with their complements, $\mathcal{F}_s^c$ and $\mathcal{F}_{sr}^c$, serving as training samples, where $\mathcal{F}_{sr}\cup \mathcal{F}_{sr}^c=\mathcal{F}_s^c$.}
\label{fig:do_learner_schematic}
\end{figure}
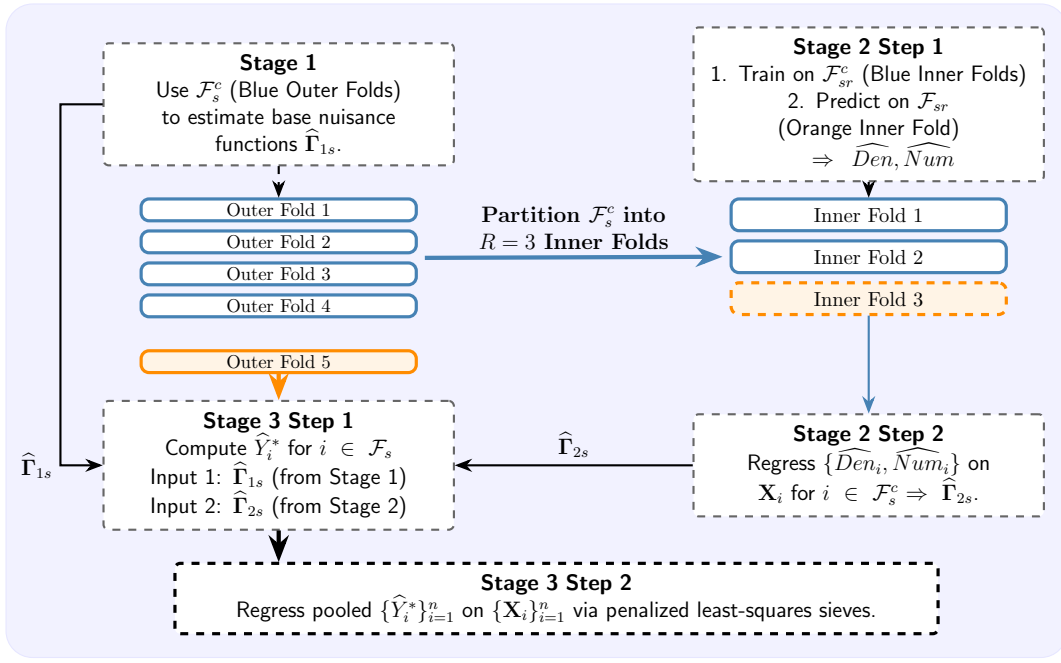

\subsection{Some technical notation}
We define a few technical notation for the remainder of the article. Let $\lVert \bullet \rVert_{\mathrm{op}}$ denote the matrix operator norm and $\lVert f \rVert_{\mP,q} = \left( \int |f|^q d\mP \right)^{1/q}$ the $\mathcal{L}^q(\mP)$ norm. In particular, $\lVert f \rVert_{\mP,\infty}=\sup_{\bx} |f(\bx)|$ is the uniform norm. Let $\omega_k(\bullet)$ be the $k$th eigenvalue of a generic matrix $\bullet\in\mathbb{R}^{K\times K}$ such that $\omega_1\leq\ldots\leq\omega_K$. We write $a_n\lesssim b_n$ if $a_n\leq cb_n$ for some constant $c$ independent of $n$, and $a_n\asymp b_n$ if $a_n\lesssim b_n$ and $b_n\lesssim a_n$. We use $X =_d Y$ to denote that two random variables $X$ and $Y$ have the same distribution, and write $X =_d Y + o_{\mP}(a_n)$ if $(X - Y)/a_n$ converges to zero in probability. We also use $a_n \lesssim_{\mP} b_n$ to denote that the stochastic sequence $a_n$ converges no faster than $b_n$, i.e., $a_n = O_{\mP}(b_n)$. The notations $a_n \lesssim_{\mP} b_n$ and $a_n = O_{\mP}(b_n)$ are used interchangeably for convenience. 

\subsection{Estimating CPCE using penalized linear sieves}\label{SS:CPCE-inference}
To construct the pseudo-outcome for $\tau_{d_0d_1}(\bX)$, we first derive the efficient influence function (EIF) for the smoothed parameter $\widetilde{\tau}_{d_0d_1}=E\{\tau_{d_0d_1}(\bX)\}=E\{{\tau^N_{d_0d_1}(\bX)}/{\tau^D_{d_0d_1}(\bX)}\}$, given by Theorem \ref{thm:EIF for smoothed cpce}. To facilitate the exposition of the EIF, we define the following function of the observed data vector $\mathcal{O}$, for $z\in\{0,1\}$: $\psi_{F(Y,D,\bC),z}=\pi(\bC)^{-z}\{1-\pi(\bC)\}^{-(1-z)}\mathcal{I}(Z=z)\Big[F(Y,D,\bC)-E\{F(Y,D,\bC)|Z=z,\bC\}\Big]+E\{F(Y,D,\bC)|Z=z,\bC\}$,
where $\mathcal{I}(\bullet)$ is the indicator function. For example, $\psi_{D,z}$ and $\psi_{YD,z}$ can be seen as the uncentered EIFs for estimating $E\{D(z)\}$ and $E\{Y(z)D(z)\}$, respectively \citep{JiangJRSSB2022,tong2025semiparametric}.

\begin{theorem}\label{thm:EIF for smoothed cpce}
Under Assumptions \ref{assump:treatment-ignorability}-\ref{assump:cor}, 
the uncentered EIFs for the smoothed CPCE denominator, $\widetilde{\tau}_{d_0d_1}^D:=E\{\tau_{d_0d_1}^D(\bX)\}$, and numerator, $\widetilde{\tau}_{d_0d_1}^N:=E\{\tau_{d_0d_1}^N(\bX)\}$, are respectively given by $\varphi^D_{d_0d_1}(D,Z,\bC)=e_{d_0d_1}(\bC)+\phi^e_{d_0d_1}(D,Z,\bC)$ and $\varphi^N_{d_0d_1}(\mO)=\varphi^D_{d_0d_1}(D,Z,\bC)\{m_{1d_1}(\bC)-m_{0d_0}(\bC)\}+e_{d_0d_1}(\bC)\phi^m_{d_0d_1}(\mO)$,
where $\phi^m_{d_0d_1}(\mO)=\widetilde{\phi}^m_{d_1}(\mO)-\widetilde{\phi}^m_{d_0}(\mO)$, $\widetilde{\phi}^m_{d_z}(\mO)=\Pr(D=d_z|Z=z,\bC)^{-1}\{\psi_{YD^{d_z}(1-D)^{1-d_z},z}-m_{zd_z}(\bC)\psi_{D^{d_z}(1-D)^{1-d_z},z}\}$, and $\phi_{d_0d_1}^e(D,Z,\bC)$ is given by, provided that $\Pr(\theta(\bC)-1\neq0)=1$, 
{
\begin{align*}
\frac{(-1)^{d_0+d_1}}{2\sqrt{\delta(\bC)}}\left\{\sum_{z=0}^1(\psi_{D,z}-p_z(\bC))\left\{(2d_{1-z}-1)\sqrt{\delta(\bC)}-\left[1+(\theta(\bC)-1)p_z(\bC)-(\theta(\bC)+1)p_{1-z}(\bC)\right]\right\}\right\}.
\end{align*}
}
Then the EIF for the smoothed CPCE $\widetilde{\tau}_{d_0d_1}$ is given by
\begin{align*}
    \varphi_{d_0d_1}(\mO) &= -\widetilde{\tau}_{d_0d_1}+\tau_{d_0d_1}(\bX)+\frac{\varphi^N_{d_0d_1}(\mO)-\tau_{d_0d_1}(\bX)\varphi^D_{d_0d_1}(D,Z,\bC)}{\tau_{d_0d_1}^D(\bX)}.
\end{align*}

\end{theorem}
To accommodate monotonicity and counterfactual intermediate independence in Theorem \ref{thm:EIF for smoothed cpce} as two special limiting cases, only the expression for $\phi_{d_0d_1}^e$ needs modification, while all other expressions remain the same. For the former, one assumes $p_1(\bC)>p_0(\bC)$ and obtains $\lim_{\theta\to\infty}\phi_{d_0d_1}^e=(2d_0-d_1)\{\psi_{D,0}-p_0(\bC)\}+(2d_1-d_0-1)\{\psi_{D,1}-p_1(\bC)\}$. 
For the latter, one obtains 
$\lim_{\theta\to1}\phi_{d_0d_1}^e=\sum_{z=0,1}(-1)^{d_z+1}\{\psi_{D,z}-p_z(\bC)\}p_{1-z}(\bC)^{d_{1-z}}\{1-p_{1-z}(\bC)\}^{1-d_{1-z}}$.
More generally, the EIF in Theorem \ref{thm:EIF for smoothed cpce} motivates the form of the pseudo-outcome at the final stage: 
\begin{align}
    Y^\ast(\mO;d_0d_1,\bGamma):=\tau_{d_0d_1}(\bX)+\frac{\varphi^N_{d_0d_1}(\mO)-\tau_{d_0d_1}(\bX)\varphi^D_{d_0d_1}(D,Z,\bC)}{\tau_{d_0d_1}^D(\bX)}.\label{eq:EIF-psedo}
\end{align}
By the law of total expectation, we can verify that $E\{Y^\ast(\mO;d_0d_1,\bGamma)|\bX\}=\tau_{d_0d_1}(\bX)$,
which motivates the final-stage regression of the pseudo-outcome $Y^\ast(\mO; \bGamma)$ (omitting $d_0d_1$ whenever there is no ambiguity) on the subset of covariates, $\bX$. However, the pseudo-outcome $Y^\ast(\mO;\bGamma)$ is unobserved, as it depends on the nuisance functions $\bGamma$ that must be estimated from the observed data. 

To mitigate the ``own-observation'' bias (arising from using the same data for prediction and nuisance estimation) and nonlinearity bias (arising from using the same data for base and intermediate nuisance estimation), we generalize the double cross-fitting scheme proposed by \citet{newey2018cross}. 
We randomly partition the full sample into $S$ disjoint outer folds of approximately equal size. Let $\mathcal{F}_s \subseteq [n]$ denote the set of indices in the $s$-th fold, where $[n] = \{1, \ldots, n\}$. 
In Stage one, the base nuisance estimates $\widehat{\bGamma}_{1s}$ based on the training sample $\mathcal{F}^c_{s}=[n]\setminus \mathcal{F}_{s}$ are obtained using a flexible, data-adaptive algorithm (e.g., Super Learner \citep{van2007super}). In Stage two, we estimate the intermediate nuisance functions $\widehat{\bGamma}_{2s}$ by further partitioning the training sample $\mathcal{F}^c_s$ into $R$ disjoint inner folds, denoted as $\mathcal{F}_{sr} \subseteq \mathcal{F}^c_s$ for $r \in [R]$. We define the intermediate pseudo-outcomes as $Den(\bGamma_1):=\varphi_{d_0d_1}^D$ and $Num(\bGamma_1):=\varphi_{d_0d_1}^N$. The estimate $\widehat{\bGamma}_{2s}$ is obtained by regressing the cross-fitted pseudo-outcomes $\{Den_i(\widehat{\bGamma}_{1sr}), Num_i(\widehat{\bGamma}_{1sr})\}_{i\in\mathcal{F}_s^c}$ on $\{\bX_i\}_{i\in\mathcal{F}_s^c}$ using nonparametric regression methods, where $\widehat{\bGamma}_{1sr}$ is estimated using $\mathcal{F}_{sr}^c=\mathcal{F}_{s}\backslash\mathcal{F}_{sr}$ and distinct from $\widehat{\bGamma}_{1s}$. 
In Stage three, the estimated pseudo-outcome $\{\widehat{Y}_i^\ast=Y_i^\ast(\widehat{\bGamma}_s)\}_{i=1}^n$ is regressed on $\{\bX_i\}_{i=1}^n$ via a least-squares series estimator with generalized ridge penalization (explained below). 
That is, DCDR learner employs sequential orthogonal learning: an inner layer for $\bGamma_2$ and an outer layer for CPCE estimands. For clarity, the double cross-fitting scheme described above is outlined in Algorithm \ref{alg:double_cross_fitting}. In practice, we recommend setting $S=5$ and $R=3$ when the sample size is moderate; this configuration is utilized throughout the simulations and data applications in this manuscript.

\begin{algorithm}[t!]
\caption{Doubly cross-fit doubly robust learning}
\label{alg:double_cross_fitting}
\begin{algorithmic}[1]
\Require Observed data $\{\mathcal{O}_i\}_{i=1}^n$, $S$ outer folds, and $R$ inner folds 
\Ensure Estimated pseudo-outcome $\widehat{Y}_i^*$ for all $i\in[n]$

\State Randomly partition full sample into $S$ disjoint outer folds $\mathcal{F}_s,s\in[S]$.

\For{$s = 1$ to $S$}
    \State \textbf{Stage 1:} Define training sample $\mathcal{F}_s^c = [n] \setminus \mathcal{F}_s$. Estimate $\widehat{\boldsymbol{\Gamma}}_{1s}$ using $\mathcal{F}_s^c$.
    
    \State \textbf{Stage 2 Step 1:} Randomly partition $\mathcal{F}_s^c$ into $R$ disjoint inner folds $\mathcal{F}_{sr},r\in[R]$.
    \For{$r = 1$ to $R$}
        \State Estimate $\widehat{\boldsymbol{\Gamma}}_{1sr}$ using $\mathcal{F}_{sr}^c = \mathcal{F}_s^c \setminus \mathcal{F}_{sr}$.
        \State Compute $\widehat{Den}=\varphi_{d_0d_1}^D(\widehat{\boldsymbol{\Gamma}}_{1sr}), \widehat{Num}=\varphi_{d_0d_1}^N(\widehat{\boldsymbol{\Gamma}}_{1sr})$ using $\widehat{\boldsymbol{\Gamma}}_{1sr}$ for all $i \in \mathcal{F}_{sr}$.
    \EndFor
    
    \State \textbf{Stage 2 Step 2:} Regress $\{\widehat{Den}_i,\widehat{Num}_i\}_{i \in \mathcal{F}_s^c}$ on $\mathbf{X}_i$ to obtain $\widehat{\boldsymbol{\Gamma}}_{2s}$.
    
    \State \textbf{Stage 3 Step 1:} Compute $\widehat{Y}^* = Y^*(\widehat{\boldsymbol{\Gamma}}_{1s}, \widehat{\boldsymbol{\Gamma}}_{2s})$ for all $i \in \mathcal{F}_s$.
\EndFor

\State \textbf{Stage 3 Step 2:} Regress $\{\widehat{Y}_i^*\}_{i=1}^n$ on $\{\mathbf{X}_i\}_{i=1}^n$ via penalized least-squares
series.

\end{algorithmic}
\end{algorithm}





We write $\tau_{d_0d_1}(\bx)$ with a lower-case argument $\bx \in \mathcal{X}$ to denote the CPCE evaluated at a fixed point, thereby distinguishing it from the random variable $\tau_{d_0d_1}(\bX)$. The least-squares series approach is a computationally efficient sieve estimator in which the conditional mean function $\tau_{d_0d_1}(\bx)$ is approximated by a linear span of the set of basis functions, $\bb(\bx)^\top\bbeta$ ($\bb(\bx)=(b_1(\bx),\ldots,b_K(\bx))^\top$ and $\bbeta=(\beta_1,\ldots,\beta_K)^\top$; also known as a \emph{linear sieve}), whose dimension $K=K_n$ grows slowly with the sample size \citep{chen2007large,newey1997convergence}. Common choices for basis functions include polynomials, splines, and wavelets. The quality of the resulting approximation is governed by the quantity $\xi_K := \sup_{\bx} \lVert \bb(\bx) \rVert_2$, where $\lVert \bullet \rVert_p$ is the $\ell^p$ norm. Of note, provided that the basis functions are not overly collinear (see Assumption~\ref{asp:regularityL2}(a)), the best attainable rate for $\xi_K$ is $\sqrt{K}$, i.e., $\sqrt{K}\lesssim\xi_K$, since $\xi_K^2 = \sup_{\bx} \lVert \bb(\bx) \rVert_2^2 \geq E\{\lVert \bb(\bX) \rVert_2^2\}  \asymp K$. For instance, with spline series, wavelet series, and Fourier series, the best rate $\xi_K\lesssim \sqrt{K}$ is attainable \citep{chen2007large}. With the above preliminaries, we then propose to estimate $\bbeta$ as the solution to the least squares loss under a generalized ridge penalty \citep{eilers1996flexible}:
\begin{align}
    \widehat{\bbeta}=&\arg\min_{\bbeta} \frac{1}{2}\left\{\mP_n\{(\widehat{Y}^\ast-\bb(\bX)^\top\bbeta)^2\}+\lambda\bbeta^\top \bP\bbeta\right\}=(\widehat{\bH}+\lambda\bP)^{-1}\widehat{\bh},\label{eq:criteria-PLSS}
\end{align}
where $\mP_n\{V\}=n^{-1}\sum_{i=1}^nV_i$ denotes the sample mean, $\widehat{\bH}=\mP_{n}\{\bb(\bX)\bb(\bX)^\top\}$, $\widehat{\bh}=\mP_n\{\bb(\bX)\widehat{Y}^\ast\}$, $\lambda$ is the ridge shrinkage parameter with $\lambda=\lambda_n=o(1)$, and $\bP\in\mathbb{R}^{K\times K}$ is the user-specified positive-definite quadratic penalty matrix. Table \ref{table:penalty-matrix} gives several possible specifications of the penalty matrix, including the standard ridge penalty with $\bP = \boldsymbol{I}_K$ ($\boldsymbol{I}_K$ denotes the $K \times K$ identity matrix), the weighted ridge penalty with $\bP = \mathrm{diag}(w_1, \ldots, w_K)$, and the smoothing penalty (thin plate) with $\bbeta^\top\bP \bbeta = \int \lVert \nabla^2 (\bbeta^\top\bb(\bx)) \rVert_\mathrm{F}^2 \, d\bx$ ($\lVert \bullet \rVert_\mathrm{F}$ denotes the Frobenius norm, and $\nabla^2 \bullet$ denotes the Hessian operator). 

\begin{table}[htbp!]
\centering
\caption{Several example choices for the quadratic penalty matrix, $\bP$, in Equation \eqref{eq:criteria-PLSS}.}
\label{table:penalty-matrix}
\begin{tabularx}{\textwidth}{p{3cm} p{5cm} p{6cm}}
\toprule
{Penalty Type} & {Penalty Matrix} & {Interpretation} \\
\midrule
Standard Ridge & $\bP = \mathbf{I}_K$ & Shrink all coefficients equally. \\
Weighted Ridge & $\bP = \mathrm{diag}(w_1, \dots, w_K)$ & Shrink coefficients unequally. \\
Smoothing Penalty & $\scriptsize \bbeta^\top \bP \bbeta = \int \lVert \nabla^2 (\bbeta^\top \bb(\bx)) \rVert_\mathrm{F}^2 \, d\bx$ & Shrink towards a smoother curve. \\
\bottomrule
\end{tabularx}
\end{table}

Finally, the proposed estimator for the CPCE is given by $\widehat{\tau}_{d_0d_1}(\bx)=\bb(\bx)^\top\widehat{\bbeta}$. 
Hereafter, we refer to the estimator $\widehat{\tau}_{d_0d_1}(\bx)$ as the \emph{penalized least-squares series DCDR learner (PLSS-DCDR learner)}. 
We provide a remark below on tuning the number of basis functions and the shrinkage parameter to guide practical implementation.
\begin{remark}\label{remark:CV-pick-K,lambda}
As demonstrated later, valid inference (based on either pointwise or uniform confidence bands) requires certain undersmoothing conditions, under which the bias must be of a smaller order relative to the estimated standard error \citep{belloni2015some}. To implement undersmoothing in practice, we follow \cite{eilers1996flexible,wood2017generalized}: (i) specify a moderate number of basis functions $K$; and (ii) empirically select $\lambda$ that minimizes the Generalized Cross Validation (GCV) score, with undersmoothing controlled by $\gamma<1$.
\end{remark}
The intuition behind the PLSS-DCDR learner can be understood as follows. Define the weighted $\mathcal{L}^2$ projection of $\tau_{d_0d_1}(\bx)$ on the linear sieve $\bb(\bx)^\top\bbeta$ as
\begin{align}\label{eq:projection-target}
    \bbeta_0=&\arg \min_{\bbeta} E\{r(\bX)^2\},
\end{align}
where $r(\bX)=\tau_{d_0d_1}(\bX)-\bb(\bX)^\top\bbeta$ is the approximation error. By similar arguments to those used in Theorem 1 in \cite{kennedy2023semiparametric}, it can be shown that $\widehat{\bbeta}$ (without regularization) is the semiparametrically efficient one-step estimator for $\bbeta_0$. Consequently, the total error splits into the estimation error $\widehat{\bbeta}-\bbeta_0$ (including the regularization bias if $\lambda\neq0$) and the approximation error $r(\bX)$.

To justify the optimality of the proposed DCDR learner, conditional on the training sample, we have the following decomposition of the expected conditional bias, defined as $\text{Bias}=E(\widehat{Y}^\ast-Y^\ast|\bX,\mathcal{F}_s^c)$:
\begin{align*}
\text{Bias}
=&\frac{E\{\widehat{\varphi}^N_{d_0d_1}-\varphi^N_{d_0d_1}|\bX,\mathcal{F}_s^c\}-\widehat{\tau}_{d_0d_1}E\{\widehat{\varphi}^D_{d_0d_1}-\varphi^D_{d_0d_1}|\bX,\mathcal{F}_s^c\}}{\widehat{\tau}_{d_0d_1}^D}+\frac{(\widehat{\tau}_{d_0d_1}-\tau_{d_0d_1})(\widehat{\tau}_{d_0d_1}^D-\tau_{d_0d_1}^D)}{\widehat{\tau}_{d_0d_1}^D},
\end{align*}
where $\widehat{\varphi}^N_{d_0d_1}:=\varphi^N_{d_0d_1}(\mO;\widehat{\bGamma}_s)$, $\widehat{\varphi}^D_{d_0d_1}:=\varphi^D_{d_0d_1}(D,Z,\bC;\widehat{\bGamma}_s)$, and the first and second terms represent the bias components arising from estimating the base and intermediate nuisance functions, respectively. The form of the above decomposition has two important implications. First, the second term demonstrates that the proposed DCDR learner satisfies the mixed-bias property \citep{rotnitzky2021characterization} with respect to the intermediate nuisance functions $\bGamma_2$. 
Consequently, the proposed DCDR learner is robust to intermediate nuisance estimation error, because the error propagation of the intermediate nuisance function $\widehat{\bGamma}_2$ is quadratic; in contrast, the counterpart in the direct ratio estimator of \cite{takatsu2025doubly} is linear and does not satisfy the mixed-bias property. Second, following \cite{tong2025semiparametric}, the first term demonstrates that the proposed DCDR learner also satisfies the mixed-bias property \citep{rotnitzky2021characterization} with respect to the base nuisance functions $\bGamma_1$, as $\varphi^N_{d_0d_1}$ and $\varphi^D_{d_0d_1}$ are uncentered EIFs. Therefore, the proposed DCDR learner is also robust to base nuisance estimation errors, depending only on second- and higher-order base nuisance estimation errors. Finally, a rigorous justification for the above nonlinear error propagations is encapsulated in the term $m_{2n}$ of Theorem \ref{thm:least-squares-series-density-ratio}.

\section{Asymptotic theory and statistical inference}\label{sec:large-sample-theory-inference-bands}
We examine two statistical properties of the proposed PLSS-DCDR learner: $\mathcal{L}^2$ convergence and uniform ($\mathcal{L}^\infty$) convergence. In brief, we show that, under suitable regularity conditions, our estimator is oracle efficient, achieving the same rate of convergence as if the true nuisance functions were known. For inferential purposes, we further develop pointwise normality and a strong mean-zero Gaussian process approximation, offering a theoretical foundation for constructing both pointwise and uniform confidence bands.


\subsection{Regularity conditions}\label{ss:regularity-conditions}
To establish the $\mathcal{L}^2$ and uniform limit theory, we first discuss several key regularity conditions that are typical for least-squares sieve methods \citep{belloni2015some} and debiased machine learning \citep{dml}. The first set of conditions facilitates the derivation of the $\mathcal{L}^2$ rate of convergence.

\begin{assumption}{(Regularity conditions for $\mathcal{L}^2$ convergence)}\label{asp:regularityL2} 
{(a) The eigenvalues of $\bH=E\{\bb(\bX)\bb(\bX)^\top\}$ are uniformly bounded above and away from zero.}\\
{(b) For all $n$ and $K$, there exist finite constants $c_{K}$ and $l_{K }$ such that $\lVert \tau_{d_0d_1}(\bX)-{\bb }(\bX)^\top\bbeta_0\rVert_{\mP,2}\leq c_{K}$ and $\lVert \tau_{d_0d_1}(\bX)-{\bb }(\bX)^\top\bbeta_0\rVert_{\mP,\infty}\leq l_{K }c_{K }$.}\\
{(c) The complexity of the basis functions satisfies $\xi_{K}^2 \log K / n = o(1)$, where $K = K_n$.}\\
{(d) The diameter of the covariate support $\calX$, defined as $\sup_{\bx_1,\bx_2}\|\bx_1-\bx_2\|_2$, is uniformly bounded above, and the true CPCE is square-integrable. For some $\epsilon_1\in(0,1/2]$, $\epsilon_2>0$, $z=0,1$, all $s\in[S]$, and all $\bX$ and all $\bC$, $\epsilon_1\leq\{\pi(\bC),\widehat{\pi}(\bC),p_z(\bC),\widehat{p}_z(\bC)\}\leq1-\epsilon_1$, $\epsilon_1\leq\{\tau^D_{d_0d_1}(\bX),\widehat{\tau}^D_{d_0d_1}(\bX)\}\leq\epsilon_2$, $\tau^N_{d_0d_1}(\bX)\leq \epsilon_2$, $E\{(Y-m_{zd_z}(\bC))^2|Z=z,D=d_z,\bC\}\leq\epsilon_2$, $E\{\Omega^2|\bX\}\leq \epsilon_2$, and $\sup_k\sup _{\bx}b_k(\bx)\leq \epsilon_2$, where $\Omega:=Y^\ast - \tau_{d_0d_1}(\bX)$ is the regression error. In addition, $\inf_{\bC\in\mathcal{C}} \theta(\bC)>0$ and $\inf_{\bC\in\mathcal{C}^\ast}|\theta(\bC)-1|>0$, where $\mathcal{C}^\ast:=\{\bC\in\mathcal{C}:\theta(\bC)\neq 1\}$.}
\end{assumption}


Assumption \ref{asp:regularityL2}(a) bounds the condition number of the design matrix, ensuring that the covariates do not exhibit excessive collinearity. To simplify the proof, and without loss of generality, $\bH$ can be orthonormalized to the identity matrix \citep{newey1997convergence}. Assumption \ref{asp:regularityL2}(b) specifies bounds on the approximation error between the projected pseudo-target function $\bb(\bX)^\top\bbeta_0$, as defined in \eqref{eq:projection-target}, and the true CPCE $\tau_{d_0d_1}(\bX)$. Here, the sequence $c_K$ bounds the $\mathcal{L}^2(\mP)$ norm, and $l_K$ denotes the modulus of continuity for the worst-case uniform normal ($\mathcal L^\infty(\mP)$) with respect to the $\mathcal{L}^2(\mP)$ norm. For commonly used basis functions, $l_{K}$ increases slowly with $K$ and can often be further bounded in terms of $\xi_{K}$ and other properties of the basis functions \citep{chen2007large}. Assumption \ref{asp:regularityL2}(c) ensures that the growth rate of the basis function complexity is sufficiently slow to maintain an appropriate bias–variance tradeoff. Assumption \ref{asp:regularityL2}(d) requires: (i) boundedness of certain base and intermediate nuisance functions; (ii) finite conditional variance of the outcome given $(D, Z, \bC)$; (iii) finite conditional second moments of the regression errors given $\bX$; (iv) uniform boundedness of the components of the basis functions; and (v) the odds ratio sensitivity parameter is regular and bounded (special cases of monotonicity and counterfactual intermediate independence are addressed separately in subsequent results). Of note, it is possible to drop the boundedness assumption (iv) by replacing the $\mathcal L^2$ metric of the nuisance estimation error with $\mathcal L^{2p}$ and $\mathcal L^{2q}$ norms, where $1/p + 1/q = 1$, via the Hölder's inequality. However, since the $\mathcal L^2$ metric is more commonly used, we retain the assumption in its current form.

To derive the uniform rate of convergence, we impose an additional condition. This assumption requires that the conditional $v$-th moment ($v>2$) of the regression error $\Omega$ is uniformly bounded, and that both the growth rate and smoothness of the basis functions are further restricted.

\begin{assumption}{(Regularity condition for $\mathcal{L}^\infty$ convergence)}\label{assump:uniform-further-boundedness} 
For some $\nu>2$,  $\sup_{\bx}E\{|\Omega|^\nu|\bX=\bx\}=O(1)$, $\xi_K^{2\nu/(\nu-2)}\log K/n=O(1)$, $\log \xi_K^L=O(\log K)$, and $\log \xi_K =O(\log K)$, where $\widetilde{\bb}(\bx) := \bb(\bx) / \lVert \bb(\bx) \rVert_2$ denotes the vector of normalized basis functions, and the Lipschitz constant for function $\widetilde{\bb}(\bx)$ is defined as $    \xi_K^L:= \sup_{\bx\neq\bx'}{\lVert \widetilde{\bb}(\bx)-\widetilde{\bb}(\bx')\rVert_2}/{\lVert \bx-\bx'\rVert_2}$.
\end{assumption}

\subsection{Oracle efficiency}\label{subsec:oracle-efficiency}
We derive the $\mathcal L^2$ (Theorem \ref{thm:least-squares-series-density-ratio}) and uniform convergence rates (Theorem \ref{thm:uniform-rate}) of the proposed PLSS-DCDR learner, and identify the conditions under which these rates are of the same order as if the first-stage and second-stage nuisance functions were known; that is, it achieves oracle efficiency. 

\begin{theorem}\label{thm:least-squares-series-density-ratio}
Under Assumptions \ref{assump:treatment-ignorability}-\ref{asp:regularityL2}, the error of the PLSS-DCDR learner in the $\mathcal{L}^2(\mathbb{P})$ norm is bounded by
\begin{align*}
&\lVert \widehat{\tau}_{d_0d_1}-\tau_{d_0d_1}\rVert_{\mP,2}\lessp (1+2\lambda\omega_{\min})^{-1}\left(\sqrt{\frac{K}{n}}+ m_{0n}+m_{1n}+m_{2n}+\lambda\omega_{\max}\right)+c_K,
\end{align*}
where $\omega_{\min}=\omega_1(\bP)$, $\omega_{\max}=\omega_K(\bP)$,
$\eta(\bullet)=\lVert \widehat{\bullet}-\bullet\rVert_{\mP,2}$ denotes the nuisance estimation error in $\mathcal L^2(\mP)$ norm, $m_{0n}=\min\{l_Kc_K\sqrt{{K}/{n}}, \allowbreak {\xi_Kc_K}/{\sqrt{n}}\}$, $m_{1n}=\xi_K/\sqrt{n}[\eta(\tau_{d_0d_1}^N)+     \eta(\tau_{d_0d_1}^D)+\eta(\pi)+\sum_{z=0,1}\{\eta(p_z)+\eta(m_{zd_z})\}]$, $m_{2n}=m_{2n}^{(1)}+m_{2n}^{(2)}$, $m_{2n}^{(1)}=\sqrt{K} [\sum_{z=0,1}\eta(m_{zd_z})\allowbreak\{\eta(\pi)+\sum_{j=0,1}\eta(p_j)\}+\eta(p_0)\eta(p_1)+\eta(\pi)\sum_{j=0,1}\eta(p_j)]$, and $m_{2n}^{(2)}=\sqrt{K}\{\eta(\tau_{d_0d_1}^N)\eta(\tau_{d_0d_1}^D)+\eta(\tau_{d_0d_1}^D)^2\}$.
\end{theorem}
In Theorem \ref{thm:least-squares-series-density-ratio},  the terms $m_{1n}$ and $m_{2n}$ summarize the impact of nuisance function estimation errors. The term $m_{1n}$, analogous to the empirical process term of the debiased machine learning estimators \citep{dml}, is typically of a smaller order than $m_{2n}$ provided the nuisance estimators are consistent, as $\xi_K \lesssim \sqrt{K}$ holds for many commonly used basis functions. In contrast, $m_{2n}$ is the critical term that typically dominates the convergence rate. 
The terms $m_{2n}^{(1)}$ and $m_{2n}^{(2)}$ both possess the mixed-bias property \citep{rotnitzky2021characterization}. Specifically, $m_{2n}^{(1)}$ depends on the second-order errors of the base nuisance estimators. The form of this term closely relates to robustness properties established in the literature. For example, it reflects the conditional double robustness property proved in \cite{tong2025semiparametric}, as the product term $\eta(p_0)\eta(p_1)$ requires that the principal score model be correctly specified. Under monotonicity, $\eta(p_0)\eta(p_1)$ vanishes, yielding the triple robustness property of \cite{JiangJRSSB2022}; under counterfactual intermediate independence, $\sum_{z=0,1}\eta(m_{zd_z})\eta(p_z)$ disappears, corresponding to the quadruple robustness property proved in \cite{tong2025semiparametric}. On the other hand, the term $m_{2n}^{(2)}$ represents the second-order errors of the intermediate nuisance estimators. 
In principle, the remainder term $m_{2n}^{(1)}$ is of smaller order than $m_{2n}^{(2)}$ if the estimation of the intermediate nuisance functions $\bGamma_2$ is not carefully handled, as $\bGamma_2$ depends on the unknown base nuisance functions $\bGamma_1$. However, because $\bGamma_2$ is estimated using an additional orthogonal learner with sample cross-fitting in Stage 2, the convergence rate of $m_{2n}^{(2)}$ is significantly improved. For example, in ideal cases where $\bGamma_2$ is sufficiently smooth and the Stage 2 orthogonal learner $\widehat{\bGamma}_2$ is oracle efficient, $m_{2n}^{(2)}$ achieves the same order as $m_{2n}^{(1)}$. Finally, the term $\lambda\omega_{\max}$ accounts for the regularization bias.
The following corollary further shows that the PLSS-DCDR learner achieves oracle efficiency in the $\mathcal L^2(\mP)$ norm under suitable conditions.

\begin{corollary}\label{corollary:L2-negligible}
    Assuming further  $c_{K}=o(1)$, $\xi_K\asymp \sqrt{K}$, $\lambda\omega_{\max}=O(n^{-1/2})$, $m_{1n}=O(m_{0n})$, and $m_{2n}=O(m_{0n})$, then $\lVert \widehat{\tau}_{d_0d_1}(\bX)-\tau_{d_0d_1}(\bX)\rVert_{\mP,2}=O_{\mP}(\delta_{n,2})$, where the oracle rate of convergence in $L^2(\mP)$ norm, $\delta_{n,2}$, is given by $\delta_{n,2}=\sqrt{\frac{K}{n}}+c_{K}$. 
\end{corollary}

In Corollary \ref{corollary:L2-negligible}, the assumption $c_K=o(1)$ requires that the basis functions are correctly specified. Below we discuss the potential misspecification of the basis functions.
\begin{remark}
    Importantly, the series approximation is said to be correctly specified if $c_{K}=o(1)$, whereas it is considered misspecified if $c_{K} \not\to 0$. For example, misspecification occurs when one incorrectly employs basis functions that are separately additive, whereas the true underlying function is not \citep{newey1997convergence}. Even when the basis functions are misspecified, the PLSS-DCDR learner (assuming $\lambda=0$) coincides with the debiased machine learning estimator \citep{dml} targeting the estimand of the projected surrogate parameter defined in \eqref{eq:projection-target}. It can be shown that this estimator achieves $\sqrt{n}$-consistency and semiparametric efficiency under mild conditions. 
\end{remark}
Furthermore, achieving oracle efficiency requires that (i) the first-stage nuisance estimators satisfy $m_{1n} = O(m_{0n})$ and $m_{2n} = O(m_{0n})$, which is weaker than the small bias assumption (Assumption 3.5) in \cite{semenova2021debiased}; and (ii) the shrinkage parameter satisfies $\lambda\omega_{\max}=O(n^{-1/2})$, or simply $\lambda=O(n^{-1/2})$ in cases where $\omega_{\max}$ is bounded (e.g., standard ridge penalty).
Following \cite{belloni2015some}, under the assumption that $\tau_{d_0d_1}(\bx) \in \mathcal{H}_s$ (the Hölder class of smoothness order $s$) and that the approximation error satisfies $c_K \lesssim K^{-s/p}$, the oracle $\mathcal{L}^2$ convergence rate achieves minimax optimality when the number of basis functions scales as $K \asymp n^{p/(p+2s)}$.
Next, we establish the uniform convergence rate for the PLSS-DCDR learner.

\begin{theorem}
\label{thm:uniform-rate}
Under Assumptions \ref{assump:treatment-ignorability}–\ref{assump:uniform-further-boundedness}, the error of the PLSS-DCDR learner in the $\mathcal L^\infty(\mathbb{P})$ norm is bounded by
\begin{align*}
\lVert \widehat{\tau}_{d_0d_1}-\tau_{d_0d_1}\rVert_{\mP,\infty}
\lessp& \frac{\xi_K}{\sqrt{n}(1+\lambda\omega_{\min})}\Biggl\{\left(\frac{\xi_K\sqrt{\log K}}{1+2\lambda\omega_{\min}}+\sqrt{n}\right)(\lambda\omega_{\max}/\sqrt{n}+m_{1n}+m_{2n})+\\&\frac{m_{3n}}{1+2\lambda\omega_{\min}}+\sqrt{\log K}(1+l_Kc_K)\Biggr\}+l_Kc_K.
\end{align*}
where 
$m_{3n}=\sqrt{\xi_K^2\log K/n}(n^{1/\nu}\sqrt{\log K}+\sqrt{K}l_Kc_K)$.
\end{theorem}
In Theorem \ref{thm:uniform-rate}, the term $m_{4n}:=\left\{\xi_K\sqrt{\log K}/(1+2\lambda\omega_{\min})+\sqrt{n}\right\}(m_{1n}+m_{2n})$ captures the effect of nuisance estimation errors on the worst-case estimation error for the PLSS-DCDR learner. 
Moreover, the regularization bias can be controlled by appropriately bounding the term $\lambda\omega_{\max} $. The corollary below shows that the PLSS-DCDR learner achieves oracle efficiency in $\mathcal L^\infty(\mP)$ norm under suitable conditions.
\begin{corollary}\label{corollary:uniform-negligible}
    Assuming $c_{K}=o(1)$, $\xi_K\asymp\sqrt{K}$, $m_{4n}=O(m_{3n}+\sqrt{\log K}(1+l_Kc_K))$, $\lambda\omega_{\max}=o(1)$, and $\lambda\omega_{\max}/\sqrt{n}=O(m_{1n}+m_{2n})$, then  $\lVert \widehat{\tau}_{d_0d_1}-\tau_{d_0d_1}\rVert_{\mP,\infty}=O_{\mP}(\delta_{n,\infty})$, where the oracle rate of convergence in $\mathcal L^\infty(\mP)$ norm, $\delta_{n,\infty}$, is given by $\delta_{n,\infty}= {\xi_K}/{\sqrt{n}}\left\{m_{3n}+\sqrt{\log K}(1+l_Kc_K)\right\}+l_Kc_K$.

\end{corollary}
We discuss conditions under which the oracle uniform rate of convergence is minimax optimal. Following Comment 4.5 in \cite{belloni2015some}, assuming $l_K c_K \lesssim K^{-s/p}$, $\xi_K \lesssim \sqrt{K}$, and $m_{3n} + \sqrt{\log K}\, l_K c_K \lesssim \sqrt{\log K}$, the optimal uniform rate of convergence for $\mathcal{H}_s$, $(\log n / n)^{s/(2s+p)}$, is achieved by setting $K\asymp(\log n / n)^{-p/(2s+p)}$. In summary, Corollaries \ref{corollary:L2-negligible} and \ref{corollary:uniform-negligible} show that, under suitable conditions, the PLSS-DCDR learner is oracle efficient in both the $\mathcal L^2(\mathbb{P})$ and $\mathcal L^\infty(\mathbb{P})$ norms.

\subsection{Asymptotically honest statistical inference}
\subsubsection{Pointwise normality and strong Gaussian process approximation}
To begin with, we prove in Supplementary Material Lemmas \ref{lemma:pointwise-linearization} and \ref{lemma:uniform-linearization} that $\widehat{\bbeta}$ admits an asymptotic linear representation in $\mathcal L^2(\mathbb{P})$ and uniform norms, respectively. Based on these pointwise and uniform linearization results, we show in the following results that the proposed PLSS-DCDR learner is (i) pointwise asymptotically normal and (ii) can be approximated by a mean-zero Gaussian process; the strong approximation result (ii) naturally requires stronger conditions. To proceed, we focus on inference based on the $t$-process $\{T_n(\bx):\bx\in\mathcal{X}\}$ constructed from the \emph{studentized} PLSS-DCDR learner, where
\begin{align*}
    T_n(\bx)=\sqrt{n}\frac{\widehat{\tau}_{d_0d_1}(\bx)-\tau_{d_0d_1}(\bx)}{\lVert \mathbb{V}^{1/2}\bb(\bx)\rVert_2},
\end{align*}
and $\mathbb{V}=\bH^{-1}E\{(\Omega+r_0)^2\bb(\bX)\bb(\bX)^\top\}\bH^{-1}$ is the covariance matrix, $r_0=\tau_{d_0d_1}(\bX)-\bb(\bX)^\top\bbeta_0$ is the approximation error for the projection $\bbeta_0$ in \eqref{eq:projection-target}, and $\Omega=Y^\ast-\tau_{d_0d_1}(\bX)$ is the regression error.

\begin{theorem}
\label{thm:pointwise-normal}
Suppose Assumptions \ref{assump:treatment-ignorability}-\ref{asp:regularityL2} hold. Furthermore, assume (i) $\lambda\omega_{\max}(1+m_{0n}+\sqrt{K/n})=o(1)$; (ii) $m_{1n}=o(n^{-1/2})$ and $m_{2n}=o(n^{-1/2})$; (iii) the 
Lindeberg's condition holds such that $\sup_{\bx} E\{\Omega^2\mathcal{I}\{|\Omega|>\iota\}|\bX=\bx\}\to 0$ as $\iota\to\infty$; (iv) $1\lesssim\inf_{\bx} E\{\Omega^2|\bX=\bx\}$; (v) $(\xi_K\log K/n)^{1/2}(1+\sqrt{K}l_Kc_K)=o(1)$; and (vi) $\sqrt{n}r(\bx)=o(\lVert \mathbb{V}^{1/2}\bb(\bx)\rVert_2)$. Then for any given design point $\bx\in\mathcal{X}$, $\lim_{n\to\infty}\sup_{t\in\mathbb{R}}\left|\Pr\left(T_n(\bx)<t\right)-\Phi(t)\right|=0$,
where $\Phi$ is the standard normal cumulative distribution function. 

\end{theorem}
Theorem \ref{thm:pointwise-normal} states that, under suitable conditions, the proposed PLSS-DCDR learner is pointwise asymptotically normal, which facilitates the construction of pointwise Wald confidence intervals for inference. Most conditions in Theorem \ref{thm:pointwise-normal} are typical in the least squares series sieve literature \citep{belloni2015some}, except for the following. First, condition (i) indicates that pointwise normality requires the shrinkage parameter to decay faster than what is needed for achieving oracle efficiency in the $\mathcal L^2$ norm. Similarly, condition (ii) implies that pointwise normality demands the nuisance estimation errors to potentially converge faster than required for oracle efficiency in the $\mathcal L^2$ norm; this rate is akin to the small bias condition (Assumption 3.5 in \cite{semenova2021debiased}). Second, condition (vi) is an undersmoothing condition, which essentially requires that the bias be of smaller order relative to the standard error. As pointed out earlier, fine-tuning through the shrinkage parameter given moderate $K$ may practically outperform the ordinary least squares series sieve method in aligning with this condition. Practical considerations regarding tuning are further discussed in Section \ref{sec:sim_exp_main}.

For simultaneous inference over $\bx \in \mathcal{X}$, a stronger result involving Gaussian process approximation is required, as stated below.

\begin{theorem}
\label{thm:uniform-strong-gaussian}
Suppose Assumptions \ref{assump:treatment-ignorability}–\ref{assump:uniform-further-boundedness} and all conditions in Corollary \ref{corollary:uniform-negligible} hold. Furthermore, assume (i) $\lambda\omega_{\max}(1+m_{0n}+\sqrt{K/n})=o(1)$; (ii) $m_{3n}=o_{\mP}(a_n^{-1})$; (iii) $1\lesssim\inf_{\bx} E\{\Omega^2|\bX=\bx\}$; (iv) $a_n^6K^4\xi_K^2(1+l_K^3c_K^3)^2(\log n)^2/n=o(1)$; and (v) $\sup_{\bx}\sqrt{n}|r(\bx)|/\lVert \mathbb{V}^{1/2}\bb(\bx)\rVert_2=o(a_n^{-1})$. Then the following Gaussian process strong approximation holds in $\mathcal L^\infty$ norm $T_n(\bx)=_d{\bb(\bx)^\top\mathbb{V}^{1/2}}/{\lVert \mathbb{V}^{1/2}\bb(\bx)\rVert_2}\mathcal{N}(0,\boldsymbol{I}_K)+o_{\mP}(a_n^{-1})$.
\end{theorem}
Theorem \ref{thm:uniform-strong-gaussian} states that the studentized statistic from the PLSS-DCDR learner can be well approximated by a mean-zero Gaussian process up to a stochastic order of $a_n^{-1}$, which forms the basis for simultaneous inference and the construction of honest uniform confidence bands, as detailed in the following. 

\subsubsection{Construction of pointwise and uniform confidence bands}

To facilitate statistical inference and construct confidence bands, we propose using the following asymptotic covariance matrix estimator: $\widehat{\mathbb{V}}=\widehat{\bH}^{-1}\mP_n\{(\widehat{Y}^\ast-\bb(\bX)^\top\widehat{\bbeta})^2\bb(\bX)\bb(\bX)^\top\}\widehat{\bH}^{-1}$.
In Supplementary Material Theorem \ref{thm:cov-matrix-rate}, we derive the rate of convergence for the covariance matrix estimator and show that, under certain regularity conditions, it is uniformly consistent with $\sup_{\bx} | {\lVert \widehat{\mathbb{V}}^{1/2}\bb(\bx)\rVert_2}/{\lVert \mathbb{V}^{1/2}\bb(\bx)\rVert_2} - 1 | 
\lessp \lVert \widehat{\mathbb{V}} - \mathbb{V} \rVert_{\mathrm{op}} = o_{\mP}(1)$.
Therefore, we construct the confidence bands at significance level $\alpha \in (0,1)$ as
\begin{align*}
    [\mathcal{L}_n,\mathcal{U}_n]:=[\widehat{\tau}_{d_0d_1}(\bx)-c_n(1-\alpha)\lVert \widehat{\mathbb{V}}^{1/2}\bb(\bx)\rVert_2/\sqrt{n},\widehat{\tau}_{d_0d_1}(\bx)+c_n(1-\alpha)\lVert \widehat{\mathbb{V}}^{1/2}\bb(\bx)\rVert_2/\sqrt{n}].
\end{align*}
To perform pointwise inference at a given covariate value $\bx = \bx_0$, the threshold $c_n(1-\alpha)$ for the pointwise confidence band is the $(1-\alpha)$-quantile of $|\mathcal{N}(0,1)|$, or equivalently, the $(1-\alpha/2)$-quantile of $\mathcal{N}(0,1)$ by Theorem \ref{thm:pointwise-normal}. In contrast, to perform simultaneous inference with uniform confidence bands, one sets $c_n(1-\alpha)$ as the $(1-\alpha)$ quantile of $\sup_{\bx} |T_n(\bx)|$. The distribution of $\sup_{\bx} |T_n(\bx)|$ is analytically untractable and must be approximated. 

There are two possible approximation methods to obtain the threshold: the {Gaussian bootstrap} \citep{chernozhukov2014gaussian} and the {multiplier (weighted) bootstrap} \citep{chernozhukov2014anti}. The Gaussian bootstrap is based on the Gaussian process approximation established in Theorem \ref{thm:uniform-strong-gaussian}. In contrast, the multiplier bootstrap \citep{praestgaard1993exchangeably} avoids direct simulation of the limiting Gaussian process, employing instead a weighted bootstrap of the data and estimators. While the latter may exhibit superior finite-sample performance in the presence of heavy-tailed distributions, we focus on the Gaussian bootstrap due to its computational efficiency. A detailed theoretical treatment and corresponding algorithms for the multiplier bootstrap are provided in Sections \ref{sec:bayesian-boot-ci-uniform} and \ref{supp;subsec:valid-weighted-boot} of the Supplementary Material. To implement the Gaussian bootstrap, we approximate the critical value using the $(1-\alpha)$-quantile of the supremum of the Gaussian process with the estimated covariance. Algorithm \ref{alg:gaussian-threshold} outlines the simulation for $c^{(1)}_n(1-\alpha)$, justified by Theorem \ref{thm:strong-gaussian-appro-suprema} in the Supplementary Material.

\begin{algorithm}[ht!] 
\caption{Computation of the threshold $c^{(1)}_n(1-\alpha)$ for uniform confidence bands.}
\label{alg:gaussian-threshold}
\begin{algorithmic}[1]
\Require Estimated covariance matrix $\widehat{\mathbb{V}}$, sample grid $\{\bx_j\}_{j=1}^J$, number of simulations $B$
\Ensure Threshold $c^{(1)}_n(1-\alpha)$

\For{$b = 1$ to $B$}
    \State Generate a standard normal random vector $\boldsymbol{Z}_b \sim \mathcal{N}(0, \boldsymbol{I}_K)$
    \State Compute $\mathfrak{S}_b = \sup_{\bx_j} | {\bb(\bx_j)^\top \widehat{\mathbb{V}}^{1/2} \boldsymbol{Z}_b}|/{\lVert \widehat{\mathbb{V}}^{1/2}\bb(\bx_j)\rVert_2} $
\EndFor
\State Compute the empirical conditional $(1-\alpha)$-quantile of $\{\mathfrak{S}_b\}_{b=1}^B$ given the data
\State \Return $c^{(1)}_n(1-\alpha)$ as this empirical conditional quantile
\end{algorithmic}
\end{algorithm}

Finally, Theorem \ref{thm:honest-cov-gauss-boot-ci} proves that the confidence bands $[\mathcal{L}_n, \mathcal{U}_n]$ based on $c_n^{(1)}(1-\alpha)$ exhibit {asymptotically honest} uniform coverage; that is, they achieve the nominal coverage probability of $1 - \alpha$ asymptotically.

\begin{theorem}
\label{thm:honest-cov-gauss-boot-ci}
Suppose that all the conditions in Theorem~\ref{thm:uniform-strong-gaussian} hold, with some $\nu \geq 4$ and $a_n = \sqrt{\log K}$. In addition, assume that $\xi_K(\log K)^2/n^{1/2-1/\nu}=o(1)$. Then $\Pr(\tau_{d_0d_1}(\bx)\in[\mathcal{L}_n,\mathcal{U}_n]\text{ for all $\bx\in\mathcal{X}$})\to1-\alpha$.
Moreover, the convergence rate for the width of the uniform confidence bands coincides with the uniform rate of convergence of $c^{(1)}_n(1-\alpha)\lVert \widehat{\mathbb{V}}^{1/2}\bb(\bx)\rVert_2/\sqrt{n}\lessp \sqrt{{\xi_K^2\log K}/{n}}$.
    
\end{theorem}

\section{Simulation experiments}\label{sec:sim_exp_main}

To demonstrate the empirical performance of our proposed approach, we simulate a sample of $n=3000$ individuals to evaluate the CPCE conditional on a scalar covariate $X\sim\text{Uniform}[-1,1]$. 
Briefly, the data-generating process incorporates non-linear dependence on the full covariate vector 
via a mixture of trigonometric functions and polynomials. We consider four scenarios: a non-monotonic data-generating process fitted assuming non-monotonicity or monotonicity, and a monotonic data-generating process fitted assuming monotonicity or non-monotonicity. We employ the Super Learner for first- and second-stage nuisance estimation and evaluate nine smoothers for the final-stage regression. These smoothers consist of factorial combinations of three knot choices ($\{5, 10, 20\}$) and three penalty tuning strategies: (i) unpenalized B-splines, (ii) penalized splines (P-splines) selected via GCV using the \texttt{mgcv} package, and (iii) an undersmoothed P-spline variant with $\gamma=0.4$ ($\gamma=1$ corresponds to the standard GCV; $\gamma<1$ indicates the undersmoothed variant). We evaluate uniform inference performance based on the root mean integrated squared error (MISE), empirical uniform coverage probability, and average integrated uniform band width. We assess pointwise inference performance by computing the bias, Monte Carlo standard errors (MCSD), average empirical standard errors (AESE), and empirical pointwise coverage probability at four representative points: $\{-0.5, -0.25, 0.25, 0.5\}$. Supplementary Material Section \ref{sec:supp-simulation} details the complete setup and results of the simulation experiments. We briefly comment on several key findings below.

In general, when non-monotonicity is present, the PLSS-DCDR learner fitted assuming monotonicity will generally lead to biased inference. Moreover, the CLATE estimates become extremely unstable and invalid, as evidenced by extremely large MISE, uniform band widths, bias, and pointwise AESEs. This instability arises because the principal score derived under monotonicity, $\widehat{p}_1(\mathbf{C})-\widehat{p}_0(\mathbf{C})$, can become extremely small or even negative when the data exhibit non-monotonicity. On the other hand, when monotonicity holds but the PLSS-DCDR learner is fitted assuming a large constant odds ratio (e.g., $\theta=5$), both pointwise and uniform inference remain approximately valid, and the impact of such misspecification is not significant. To assess the impact of violations of the exclusion restriction and monotonicity, we further compute the ratio estimator of \cite{takatsu2025doubly} for estimating CLATE and find that their approach is generally biased. 


\begin{figure}[ht!]
    \centering
    \includegraphics[width=\linewidth]{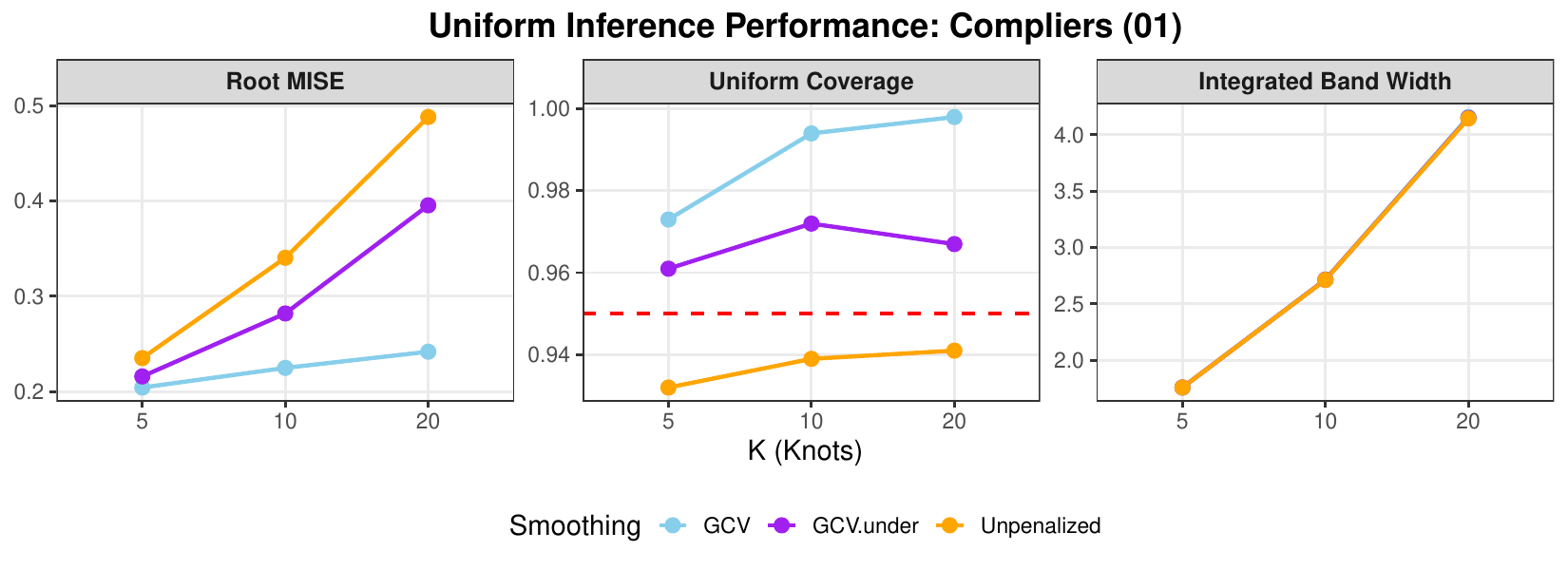}
    \caption{Simulation results presenting the root mean integrated squared error (MISE), uniform empirical coverage probability, and average integrated width of the uniform confidence bands for the proposed PLSS-DCDR learner of the CPCE among compliers when non-monotonicity holds. The basis functions are constructed using P-splines with $K \in \{5, 10, 20\}$ knots. We employ three smoothers: a penalized estimator with GCV-selected smoothing, an undersmoothed GCV variant, and an unpenalized least squares series estimator.}
    \label{fig:sim-results-01}
\end{figure}

\begin{table}[ht!]
\centering
\caption{ Simulation results presenting pointwise inference metrics, including bias, Monte Carlo standard deviation (MCSD), average estimated standard errors (AESE), and empirical coverage probabilities (CP), for the proposed PLSS-DCDR learner for the complier stratum ($01$) at four representative points $X=x_0\in\{-0.50,-0.25,0.25,0.50\}$ when non-monotonicity holds. The basis functions are constructed using P-splines with $K \in \{5, 10, 20\}$ knots. We employ three smoothers: a penalized estimator with GCV-selected smoothing, an undersmoothed GCV variant, and an unpenalized least squares series estimator.}
\label{tab:pointwise-01-compare}
\resizebox{\textwidth}{!}{%
\begin{tabular}{lrr cccc cccc cccc}
\toprule
 & & \multicolumn{4}{c}{GCV} & \multicolumn{4}{c}{GCV.under} & \multicolumn{4}{c}{Unpenalized} \\
\cmidrule(lr){3-6} \cmidrule(lr){7-10} \cmidrule(lr){11-14}
$K$ & $x_0$ & BIAS & MCSD & AESE & CP(\%) & BIAS & MCSD & AESE & CP(\%) & BIAS & MCSD & AESE & CP(\%) \\
\hline
5 & -0.50 & 0.02 & 0.10 & 0.13 & 98.8 & 0.01 & 0.11 & 0.13 & 97.7 & 0.01 & 0.13 & 0.13 & 96.0 \\ 
  & -0.25 & 0.01 & 0.12 & 0.12 & 97.0 & 0.00 & 0.12 & 0.12 & 96.5 & -0.00 & 0.13 & 0.12 & 95.5 \\ 
  & 0.25 & 0.02 & 0.12 & 0.13 & 96.4 & 0.00 & 0.12 & 0.13 & 95.7 & -0.01 & 0.13 & 0.13 & 94.1 \\ 
  & 0.50 & 0.03 & 0.17 & 0.16 & 98.3 & 0.01 & 0.18 & 0.16 & 97.4 & 0.00 & 0.20 & 0.16 & 94.6 \\ 
[0.5ex]
10 & -0.50 & 0.02 & 0.11 & 0.17 & 99.3 & 0.01 & 0.14 & 0.17 & 97.1 & 0.01 & 0.17 & 0.17 & 95.2 \\ 
  & -0.25 & 0.02 & 0.11 & 0.17 & 99.3 & 0.01 & 0.15 & 0.17 & 97.1 & 0.00 & 0.17 & 0.17 & 95.7 \\ 
  & 0.25 & 0.02 & 0.13 & 0.18 & 99.5 & -0.01 & 0.16 & 0.18 & 96.7 & -0.01 & 0.18 & 0.18 & 94.9 \\ 
  & 0.50 & 0.02 & 0.14 & 0.20 & 98.8 & 0.01 & 0.18 & 0.20 & 96.5 & -0.00 & 0.20 & 0.20 & 95.2 \\ 
[0.5ex]
20 & -0.50 & 0.02 & 0.12 & 0.26 & 99.8 & 0.00 & 0.22 & 0.26 & 97.6 & 0.00 & 0.26 & 0.26 & 95.4 \\ 
  & -0.25 & 0.02 & 0.13 & 0.25 & 99.6 & -0.00 & 0.21 & 0.25 & 96.9 & -0.00 & 0.25 & 0.25 & 94.7 \\ 
  & 0.25 & 0.02 & 0.14 & 0.28 & 99.7 & -0.02 & 0.24 & 0.28 & 96.8 & -0.02 & 0.28 & 0.28 & 94.2 \\ 
  & 0.50 & 0.02 & 0.16 & 0.31 & 99.6 & -0.00 & 0.27 & 0.31 & 96.8 & -0.01 & 0.31 & 0.31 & 94.5 \\ 
\bottomrule
\end{tabular}%
}
\end{table}

Next, we comment on the hyperparameter tuning of the proposed PLSS-DCDR learner. Figure \ref{fig:sim-results-01} and Table \ref{tab:pointwise-01-compare} present the simulation results for uniform inference and pointwise inference within the complier stratum, respectively, when non-monotonicity holds. In general, the proposed PLSS-DCDR learner, implemented with a suitable choice of knots and penalty, provides valid uniform and pointwise inference characterized by low MISE and bias, and nominal uniform and pointwise coverage. Several additional observations emerge. First, the P-spline basis with the penalty selected via GCV yields the lowest MISE, whereas the unpenalized method yields the highest. Second, bias is negligible across all tuning configurations considered, suggesting that the impact of tuning on pointwise bias appears to be limited. Third, regarding uniform and pointwise coverage probabilities, the GCV-selected P-spline smoother tends to be conservative, whereas the unpenalized method tends to be slightly anti-conservative. Hence, there exists a trade-off between optimal estimation accuracy (minimum MISE) and honest valid inference, as undersmoothing is typically required to ensure the latter. Fourth, the average integrated width of the uniform bands is relatively insensitive to the choice of smoother but is notably sensitive to the number of knots. As the number of knots increases, the MISE rises due to overfitting, and the uniform band width also increases, leading to more conservative inference. In conclusion, to achieve valid uniform and pointwise inference without incurring a substantial penalty in estimation accuracy (MISE), we recommend using the undersmoothed P-spline variant with a moderate number of knots and appropriate undersmoothing, such as $\gamma=0.4$.

\section{Data example: analysis of the ARDS trial}
\label{s:Data_Example}
We reanalyze data from the Acute Respiratory Distress Syndrome (ARDS) Network randomized trial \citep{ARDS04}, to illustrate the use of the PLSS-DCDR learner. This study involved $549$ mechanically ventilated patients with acute lung injury and ARDS, randomized to receive either lower ($Z=1$) or higher ($Z=0$) positive end-expiratory pressure (PEEP). During the trial, $68$ patients in the lower PEEP group and $76$ patients in the higher PEEP group died. Our objective is to assess the CSACE of PEEP levels on the number of days to returning home (DTRH) within $60$ days. This outcome is ambiguously defined if death occurs prior to the $60$-day cutoff. Similar to \cite{amato2015driving}, we assess the CSACE conditional on driving pressure, a scalar covariate that is arguably the strongest predictor of treatment benefit. To make the principal ignorability assumption more plausible, we conduct a comprehensive analysis adjusting for $29$ baseline covariates, including age, gender, ethnicity, and other clinical variables. In this trial, monotonicity is questionable, as neither PEEP level represents the definitive standard of care, and their relative impact on survival remains controversial. This doubt is further substantiated by the fact that, following \citet{tong2025semiparametric}, principal score estimates for compliers based on identification formulas assuming monotonicity yield negative values for nearly half of the study population, which contradicts the fact that principal scores represent probabilities. Therefore, existing approaches are not applicable to this setting, as IV methods such as that of \cite{takatsu2025doubly} assume a null CSACE, and other principal-stratification-based methods, such as \cite{chen2024bayesian}, typically assume monotonicity.
To assess sensitivity to violations of monotonicity, we consider six constant values of $\theta \in \{0.2, 0.5, 1, 2, 5, \infty\}$ for illustrative purposes, where $\theta=\infty$ corresponds to the monotonicity assumption $D(1) \geq D(0)$.

\begin{figure}[ht!]
    \centering
    \includegraphics[width=\linewidth]{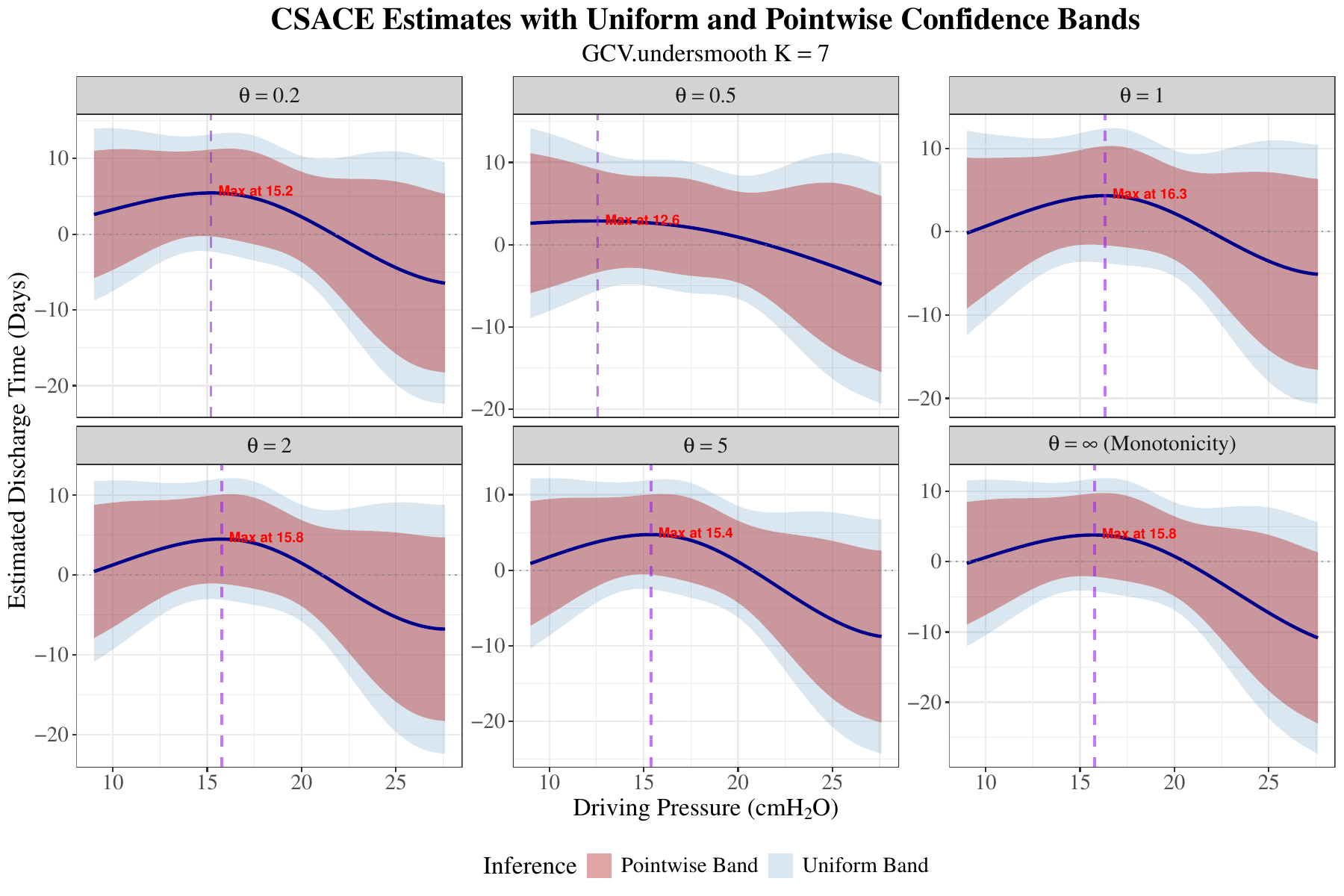}
    \caption{The estimated CSACE curves and associated 95\% pointwise and uniform confidence bands for the PLSS-DCDR learner, utilizing P-spline smoothing with $K=7$ knots, are presented in the ARDS study across six values of the conditional odds ratio $\theta\in\{0.2,0.5,1,2,5,\infty\}$ (where $\theta=\infty$ corresponds to monotonicity). The vertical purple dashed lines indicate the maximizer or turning point of the estimated CSACE curves.}
    \label{fig:ARDS-CSACE-7}
\end{figure}

Figure \ref{fig:ARDS-CSACE-7} shows the estimated CSACE curves and associated 95\% pointwise and uniform confidence bands, utilizing the PLSS-DCDR learner with P-spline smoothing ($K=7$ knots), consistent with the smoother specification used in the simulation study. We first comment on the impact of the odds ratio on the CSACE estimates. 
The harmful CSACE of higher PEEP levels on DTRH for those with very high driving pressure increases significantly as $\theta$ increases, which aligns with the interpretation that increasing $\theta$ roughly implies that higher PEEP levels tend to be stochastically harmful to survival. For example, the CSACE for those with a driving pressure of $27$ cmH$_2$O decreases from approximately $-5.5$ days to $-10$ days as $\theta$ increases from $0.2$ to $\infty$ (monotonicity).
Moreover, several additional interesting clinical observations robust to the odds ratio specification emerge. First, the estimated benefit of higher PEEP is non-monotonic, consistently peaking at driving pressures between approximately $12$ and $16~\text{cmH}_2\text{O}$. This peak aligns remarkably well with the risk threshold (approximately $15~\text{cmH}_2\text{O}$) identified by \cite{amato2015driving}, suggesting a ``sweet spot'' of recruitability where the lungs are sufficiently stiff to benefit from stabilization via higher PEEP, yet compliant enough to avoid injury. Second, as driving pressure exceeds $20~\text{cmH}_2\text{O}$---indicating significantly reduced respiratory-system compliance and small functional lung size---the treatment effect declines sharply and becomes negative. This supports the physiological hypothesis that in patients with high baseline cyclic strain, higher PEEP fails to recruit collapsed lung units and instead causes harmful overdistension. Third, the diminished benefit observed at lower driving pressures ($<10~\text{cmH}_2\text{O}$) suggests that patients with relatively healthy lung compliance derive little treatment benefit from higher PEEP. This reinforces the view that higher PEEP is not a universal intervention, but rather one that is protective primarily when it ameliorates high driving pressures. Finally, we comment on certain statistical features. First, the confidence bands 
widen at the margins due to the well-known boundary effect resulting from data sparsity. Second, a similar analysis using a higher number of knots ($K=12$) is presented in Supplementary Material Figure \ref{fig:ARDS-CSACE-12}, where the uniform bands exhibit increased width (consistent with the simulation results).

\section{Conclusion}\label{sec:conclusion}

We seek to advance the principal stratification methodology by providing a theoretically grounded framework for learning within–principal-stratum treatment effect heterogeneity under substantially weaker structural assumptions than those typically imposed for similar purposes. By embedding principal ignorability within an odds ratio sensitivity formulation, our approach moves beyond monotonicity and enables identification and efficient learning of conditional principal causal effects across all strata, including settings where non-monotonicity is scientifically plausible. From a methodological perspective, the proposed DCDR learner resolves a fundamental challenge in orthogonal learning with nested nuisance structures, yielding oracle-efficient estimation and valid pointwise and uniform inference for complex conditional principal causal effect estimands. While the validity of these results hinges on the principal ignorability assumption, in many biomedical applications including our ARDS trial reanalysis, this assumption is considered plausible through rich baseline adjustment, incorporating an extensive set of prognostic and mechanistic covariates measured prior to treatment assignment. At the same time, sensitivity analysis with respect to departures from principal ignorability assumption for assessing treatment effect heterogeneity within principal strata remains an important direction for future research. 

{Although our presentation has focused primarily on CLATE ($\tau_{01}(\bX)$) and CSACE ($\tau_{11}(\bX)$), the proposed framework is not limited to these two estimands. In general, the scientific interpretation of $\tau_{d_0d_1}(\bX)$ depends on the role of the intermediate variable in the application. For example, in mediation analysis, both $\tau_{11}(\bx)$ and $\tau_{00}(\bx)$ contribute to the assessment of direct effects because treatment does not change the mediator \citep{gallop2009mediation}. Similarly, $\tau_{10}(\bx)$ and $\tau_{01}(\bx)$ can both be scientifically meaningful when monotonicity is not imposed, and may arise in applications involving principal surrogate evaluation under causal sufficiency \citep{gilbert2008evaluating}. Thus, while CLATE and CSACE provide two canonical examples, the CPCE framework applies more broadly to heterogeneous causal effects within any scientifically interpretable principal stratum.}

\bigskip
\begin{center}
{\large\bf SUPPLEMENTARY MATERIAL}
\end{center}
\setcounter{section}{0}
\renewcommand{\thesection}{S\arabic{section}}

\begin{spacing}{0.9} 

\section{Summary}
\label{sec:intro_supple}
This supplementary material is organized as follows. 

Section \ref{sec:notation-summary} summarizes all necessary technical notation used in the main manuscript and the supplementary material for ease of reference. 

Section \ref{sec:bayesian-boot-ci-uniform} presents additional details regarding the construction of uniform confidence bands based on the Bayesian bootstrap.

Section \ref{sec:derivation-identification-CPCE} derives the nonparametric identification formulas for the CPCE under the proposed structural assumptions.

Section \ref{sec:derivation-EIF} presents the derivation of the efficient influence functions (EIFs) for the smoothed CPCE parameters $\widetilde{\tau}_{d_0d_1} = E\{\tau_{d_0d_1}(\bX)\}$.

Section \ref{sec:additional-tech-results} provides additional technical results referenced in the main manuscript, including covariance matrix estimation, linearization and strong approximation for the multiplier bootstrap, and strong approximation of suprema for the Gaussian bootstrap. 

Section~\ref{sec:proofs-all} presents the proofs of all technical results in the main manuscript and in Section~\ref{sec:additional-tech-results} of the Supplementary Material. 

Section \ref{sec:mono-supp-simulation} presents results assuming monotonicity ($D(1) \geq D(0)$), recovered as a limiting case where $\theta \to \infty$ subject to the constraint $p_1(\mathbf{C}) > p_0(\mathbf{C})$, or counterfactual intermediate independence ($D(0) \perp D(1) \mid \bC$), which corresponds to the special case $\theta(\bC)=1$.


Section~\ref{sec:supp-simulation} presents additional details of the simulation experiments, including the simulation design and supplementary results.

Section~\ref{sec:additional-table-figures} includes additional tables and figures.

\section{Notation}\label{sec:notation-summary}
Let $\| \bullet \|_{\mathrm{op}}$ denote the operator norm for a matrix, $\| \bullet \|_p$ the usual Euclidean $\ell^p$ norm for a vector, and $\| f \|_{\mP,q} = \left( \int |f(\bx)|^p d\mP(\bx) \right)^{1/p}$ the $L^p(\mP)$ norm. In particular, $\| f \|_{\mP,\infty}=\sup_{\bx} |f(\bx)|$ is the uniform norm. Let $\omega_k(\bullet)$ be the $k$th eigenvalue of a generic matrix $\bullet\in\mathbb{R}^{K\times K}$ such that $\omega_1\leq\ldots\leq\omega_K$. We say $a_n \lesssim b_n$ if there exists a constant $c$ independent of $n$ such that $a_n \leq c\,b_n$ and $a_n\asymp b_n$ if $a_n\lesssim b_n$ and $b_n\lesssim a_n$ hold simultaneously. We use $X =_d Y$ to denote that two random variables $X$ and $Y$ have the same distribution, and write $X =_d Y + o_{\mP}(a_n)$ if $(X - Y)/a_n$ converges to zero in probability. We also use $a_n \lessp b_n$ to denote that the stochastic sequence $a_n$ converges no faster than $b_n$, i.e., $a_n = O_{\mP}(b_n)$. The notations $a_n \lessp b_n$ and $a_n = O_{\mP}(b_n)$ are used interchangeably whenever convenient. We say that an event $A_n$ holds $\text{wp} \to 1$ if $\Pr(A_n) \to 1$. Consistent with standard notation in the empirical process literature, we occasionally employ $\mathbb{P}$ to denote the expectation operator $E$.

\section{Constructing uniform confidence bands using the alternative Bayesian bootstrap}\label{sec:bayesian-boot-ci-uniform}
The multiplier bootstrap \citep{praestgaard1993exchangeably} does not directly simulate the approximating Gaussian process. Instead, it performs a weighted bootstrap on the data and estimators. When the weights (or multipliers) $\{\mathfrak{w}\}_i^n$ follow a multinomial distribution $\text{Multinomial}(n,(1/n,\ldots,1/n))$, the procedure reduces to the standard nonparametric bootstrap with resampling with replacement. When the weights follow the standard exponential distribution, the resulting procedure is commonly referred to as the \emph{Bayesian bootstrap}. Following \cite{belloni2015some}, we adopt the Bayesian bootstrap. In brief, inference is conducted by approximating $c_n(1-\alpha)$ using the bootstrap analogues of the target $t$-process, given by $T^b_n(\bx)=\sqrt{n}{(\widehat{\tau}^b_{d_0d_1}(\bx)-\widehat{\tau}_{d_0d_1}(\bx))}/{\| \mathbb{V}^{1/2}\bb(\bx)\|_2}$, where $\widehat{\tau}^b_{d_0d_1}(\bx)=\bb(\bx)^\top\widehat{\bbeta}^b$ and
\begin{align*}
        \widehat{\bbeta}^b:=\arg\min_{\bbeta} \frac{1}{2}\left\{\mP_n\{\mathfrak{w}(\widehat{Y}^\ast-\bb(\bX)^\top\bbeta)^2\}+\lambda\bbeta^\top \bP\bbeta\right\}.
\end{align*}
In Theorem~\ref{thm:uniform-strong-gaussian-boot} of the Supplementary Material, we demonstrate that $T_n^b(\bx)$ admits a strong approximation by the same mean-zero Gaussian process as in Theorem~\ref{thm:uniform-strong-gaussian}, thereby justifying simultaneous inference based on the multiplier bootstrap. The computation of the threshold $c^{(2)}_n(1-\alpha)$ via the Bayesian bootstrap is detailed in Algorithm \ref{alg:multiplier-threshold} below.

\begin{algorithm}[H]
\caption{Computation of the threshold $c^{(2)}_n(1-\alpha)$ for uniform confidence bands based on Bayesian bootstrap.}
\label{alg:multiplier-threshold}
\begin{algorithmic}[1]
\Require Estimated covariance matrix $\widehat{\mathbb{V}}$, sample grid $\{\bx_j\}_{j=1}^J$, number of simulations $B$
\Ensure Threshold $c^{(2)}_n(1-\alpha)$

\For{$b = 1$ to $B$}
    \State Generate $n$ independent standard exponential random weights $\mathfrak{w}^b_i \sim \text{exp}(1)$
    \State Compute the weighted bootstrap draw $\widehat{\tau}^b_{d_0d_1}(\bx)$ 
    \State Compute 
    \[
        \mathfrak{S}_b = \sup_{\bx_j} \left| \sqrt{n}\frac{\widehat{\tau}^b_{d_0d_1}(\bx_j)-\widehat{\tau}_{d_0d_1}(\bx_j)}{\| \widehat{\mathbb{V}}^{1/2}\bb(\bx_j)\|_2} \right|
    \]
\EndFor
\State Compute the empirical conditional $(1-\alpha)$-quantile of $\{\mathfrak{S}_b\}_{b=1}^B$ given the data
\State \Return $c^{(2)}_n(1-\alpha)$ as this empirical conditional quantile
\end{algorithmic}
\end{algorithm}
We conjecture that confidence bands based on $c_n^{(2)}(1-\alpha)$ are also asymptotically honest, analogous to Theorem \ref{thm:honest-cov-gauss-boot-ci}, as the key conditions in \cite{chernozhukov2014anti} are verified, although the remaining conditions have yet to be checked.

\section{Derivation of identification formulas for the CPCEs}\label{sec:derivation-identification-CPCE}
\begin{proof}
For illustration, we consider the case $d_0 = d_1 = 1$; the results for the remaining three principal strata follow analogously. We first let $\bX=\bC$. Then, by the law of total expectation (LOTE), it follows that
\begin{align}
    E\{Y(1)|D(1)=1,\bC\}=&\sum_{d_0=0,1}\Pr(D(0)=d_0|D(1)=1,\bC)E\{Y(1)|D(0)=d_0,D(1)=1,\bC\}\nonumber\\
    =&E\{Y(1)|D(0)=1,D(1)=1,\bC\}~~~(\text{Assumption \ref{assump:PI_weak}.a}).\label{eq:proof-ident-1}
\end{align}
Similar arguments show that 
\begin{align}
    E\{Y(0)|D(0)=1,\bC\}=&E\{Y(0)|D(0)=1,D(1)=1,\bC\}.\label{eq:proof-ident-2}
\end{align}
By Treatment Ignorability, SUTVA, and Equation \eqref{eq:proof-ident-1}, we have that
\begin{align*}
E\{Y(1)|D(0)=1,D(1)=1,\bC\}=&E\{Y(1)|D(1)=1,\bC\}\\
=&E\{Y(1)|D(1)=1,Z=1,\bC\}~~~(\text{Treatment Ignorability})\\
    =&E\{Y|Z=1,D=1,\bC\}=m_{11}(\bC)~~~(\text{SUTVA}).
\end{align*}
Similarly, by SUTVA, Treatment Ignorability, Assumption \ref{assump:PI_weak}.c, and Equation \eqref{eq:proof-ident-2}, we have that $E\{Y(0)|D(0)=1,D(1)=1,\bC\}=m_{01}(\bC)$. When $\bX$ is a strict subset of $\bC$, we have that
\begin{align*}
    &E\{Y(1)|D(0)=1,D(1)=1,\bX\}\\
    =&E\{E\{Y(1)|D(0)=1,D(1)=1,\bC\}|D(0)=1,D(1)=1,\bX\}\\
    =&E\{m_{11}(\bC)|G=11,\bX\}\\
    =&E\left\{\frac{\mathcal{I}(G=11)m_{11}(\bC)}{\Pr(G=11|\bX)}|\bX\right\}\\
    =&\frac{E\{e_{11}(\bC)m_{11}(\bC)|\bX\}}{E\{e_{11}(\bC)|\bX\}},
\end{align*}
where the last two equalities follow from the LOTE. Analogous arguments demonstrate that
\begin{align*}
    E\{Y(0)|D(0)=1,D(1)=1,\bX\}=\frac{E\{e_{11}(\bC)m_{01}(\bC)|\bX\}}{E\{e_{11}(\bC)|\bX\}}.
\end{align*}
\end{proof}

\section{Derivation of EIF for the smoothed CPCE parameters}\label{sec:derivation-EIF}
\subsection{Preliminaries}
Following \cite{kennedy2022semiparametric}, we derive the EIF under a nonparametric model that imposes no restrictions on the joint density of the observed data vector $\mO=(Y,D,Z,\bC^\top)^\top$. To facilitate the derivations, let $f(\mO)$ denote the joint density function of $\mO$, which we factorize as follows:
\begin{equation*}
f(\mO)=f(\bX)f(\widetilde{\bX}|\bX)f(Z|\bC)f(D|Z,\bC)f(Y|D,Z,\bC).
\end{equation*}
Following Theorems 4.4 and 4.5 of \cite{tsiatis2006semiparametric}, the tangent space $\mathcal{F}$ for this nonparametric model is the entire Hilbert space $\mathcal{H}:=\{h(\mO):E\{h(\mO)\}=0, E\{h^2(\mO)\}<\infty\}$. Furthermore, $\mathcal{F}$ admits the orthogonal decomposition
\begin{equation*}
 \mathcal{F}=\mathcal{F}_1\oplus\mathcal{F}_2\oplus\mathcal{F}_3\oplus\mathcal{F}_4\oplus\mathcal{F}_5,
\end{equation*}
where $\mathcal{F}_1,\dots,\mathcal{F}_5$ are mutually orthogonal subspaces defined by
\begin{align*}
  &\mathcal{F}_1=\{h(\bX)\in\mathcal{H}:E\{h(\bX)\}=0\},\\
  &\mathcal{F}_2=\{h(\bC)\in\mathcal{H}:E\{h(\bC)|\bX\}=0\},\\
  &\mathcal{F}_3=\{h(Z,\bC)\in\mathcal{H}:E\{h(Z,\bC)|\bC\}=0\},\\
  &\mathcal{F}_4=\{h(D,Z,\bC)\in\mathcal{H}:E\{h(D,Z,\bC)|Z,\bC\}=0\},\\
  &\mathcal{F}_5=\{h(\mO)\in\mathcal{H}:E\{h(\mO)|D,Z,\bC\}=0\}.
\end{align*}
We consider a parametric submodel for the distribution of $\mO$, $\mathcal{P}=\{f_{\upsilon}(\mO):\upsilon\in(-\epsilon_s,\epsilon_s)\}$ for some $\epsilon_s>0$, where $f_{0}=f$ denotes the true density and $E_{0}=E$ denotes the expectation under $f$. The score vector $S(\mO)$ admits the following orthogonal decomposition:
\begin{equation*}
S(\mO)=S(\bX)+S(\widetilde{\bX}|\bX)+S(Z|\bC)+S(D|Z,\bC)+S(Y|D,Z,\bC),
\end{equation*}
where the components are the scores for the corresponding conditional densities:
\begin{align*}
  S(\mO) &= \frac{\partial \log f_{\upsilon}(\mO)}{\partial \upsilon}\Big|_{\upsilon=0}, \\
  S(Y|D,Z,\bC) &= \frac{\partial \log f_{\upsilon}(Y|D,Z,\bC)}{\partial \upsilon}\Big|_{\upsilon=0}, \\
  S(D|Z,\bC) &= \frac{\partial \log f_{\upsilon}(D|Z,\bC)}{\partial \upsilon}\Big|_{\upsilon=0}, \\
  S(Z|\bC) &= \frac{\partial \log f_{\upsilon}(Z|\bC)}{\partial \upsilon}\Big|_{\upsilon=0}, \\
  S(\widetilde{\bX}|\bX) &= \frac{\partial \log f_{\upsilon}(\widetilde{\bX}|\bX)}{\partial \upsilon}\Big|_{\upsilon=0},\\
  S(\bX) &=\frac{\partial \log f_{\upsilon}(\bX)}{\partial \upsilon}\Big|_{\upsilon=0}.
\end{align*}
We define $\widetilde{\tau}_{d_0d_1,\upsilon}$ as the value of $\widetilde{\tau}_{d_0d_1}$ under the submodel, with the true value attained at $\upsilon=0$, i.e., $\widetilde{\tau}_{d_0d_1,0}=\widetilde{\tau}_{d_0d_1}$. Theorem 3.2 of \cite{tsiatis2006semiparametric} implies that the influence function $\varphi_{d_0d_1}(\mO)\in\mathcal{H}$ for the submodel must satisfy
\begin{equation}\label{eq:influfuncchara}
  E\{\varphi_{d_0d_1}(\mO)S(\mO)\}=\frac{\partial \widetilde{\tau}_{d_0d_1,\upsilon}}{\partial {\upsilon}}\Big|_{{\upsilon}=0}.
\end{equation}
We adopt the notation $\dot{\tau}_{d_0d_1,\upsilon}|_{\upsilon=0}$ for the pathwise derivative ${\partial \widetilde{\tau}_{d_0d_1,\upsilon}}/{\partial {\upsilon}}|_{{\upsilon}=0}$, and apply this convention to all pathwise derivatives with respect to $\upsilon$. Following \cite{kennedy2022semiparametric}, there is at most one solution to the differential equation \eqref{eq:influfuncchara} under the nonparametric model. Because the tangent space for this model is the entire Hilbert space $\mathcal{H}$, the EIF $\varphi_{d_0d_1}(\mO)$ is the unique solution to \eqref{eq:influfuncchara} \cite[Theorem 4.3]{tsiatis2006semiparametric}.

\subsection{Some useful lemmas}
We begin by presenting three lemmas that simplify the subsequent derivation of the EIF for $\widetilde{\tau}_{d_0d_1}=E\{\tau_{d_0d_1}(\bX)\}$. The proofs are omitted as they follow from existing literature; the first two lemmas are adapted from Section S6.1.1 of the Supplementary Material in \cite{JiangJRSSB2022}.
\begin{lemma}\label{lemma:eif-pz(c)}
 Let $\Xi_z(\bC)=E\{F(Y,D,\bC)|Z=z,\bC\},z\in\{0,1\}$, where $F$ is a random function of $(Y,D,\bC)$. Then
 \begin{align*}
     \dot{\Xi}_{z,\upsilon}(\bC)|_{\upsilon=0}=E\{(\psi_{F(Y,D,\bC),z}-\Xi_z(\bC))S(Y,D|Z,\bC)|\bC\}.
 \end{align*}
 In particular, setting $F(Y,D,\bC)=D$ yields 
 \begin{align*}
          \dot{p}_{z,\upsilon}(\bC)|_{\upsilon=0}=E\{(\psi_{D,z}-p_z(\bC))S(Y,D|Z,\bC)|\bC\}.
 \end{align*}
\end{lemma}
\begin{lemma}\label{lemma:mzdz(c)}
For $z\in\{0,1\}$, we have that
 \begin{align*} 
&\displaystyle   \dot{m}_{z1,\upsilon}(\bC)|_{\upsilon=0}=E\left\{\phi^m_{z1}(\mO)S(Y|Z,D,\bC)|\bC\right\},\\
&\displaystyle   \dot{m}_{z0,\upsilon}(\bC)|_{\upsilon=0}=E\left\{\phi^m_{z0}(\mO)S(Y|Z,D,\bC)|\bC\right\},
 \end{align*}
 where
 \begin{align*}
     &\phi^m_{z1}(\mO)=\frac{\psi_{YD,z}-m_{z1}(\bC)\psi_{D,z}}{p_z(\bC)},\\
     &\phi^m_{z0}(\mO)=\frac{\psi_{Y(1-D),z}-m_{z0}(\bC)\psi_{1-D,z}}{1-p_z(\bC)}.
 \end{align*}
\end{lemma}
The third lemma is adapted from Equation (S6.9) of the Supplementary Material to \cite{tong2025semiparametric}.
\begin{lemma}\label{lemma:eif-eg(x),pswith theta}
The pathwise derivative of the principal score under the parametric submodel, $e_{d_0d_1,\upsilon}(\bC)$, evaluated at $\upsilon=0$, is given by
\begin{align*}
    \dot{e}_{d_0d_1,\upsilon}(\bC)|_{\upsilon=0}=&E\{\phi_{d_0d_1}^e(D,Z,\bC)S(D|Z,\bC)|\bC\},
\end{align*}
where
    \begin{align*}
        \phi_{d_0d_1}^e(D,Z,\bC)=&\frac{(-1)^{d_0+d_1}}{2\sqrt{\delta(\bC)}}\times\left\{\sum_{z=0}^1(\psi_{D,z}-p_z(\bC))\times\right.\\
&\left.\left\{(2d_{1-z}-1)\sqrt{\delta(\bC)}-\left[1+(\theta(\bC)-1)p_z(\bC)-(\theta(\bC)+1)p_{1-z}(\bC)\right]\right\}\right\}
    \end{align*}
\end{lemma}

\subsection{Derivation of the EIF}
Given these preliminaries, we now derive the EIF. For notational simplicity, we omit the subscript $d_0d_1$. By the chain rule, it follows that
\begin{align*}
    \widetilde{\tau}=&\underbrace{E\left\{\tau(\bX)S(\bX)\right\}}_{(\ast_1)}+\underbrace{E\left\{\frac{\dot{\tau}^N_{\upsilon}(\bX)}{\tau^D(\bX)}\right\}|_{\upsilon=0}}_{(\ast_2)}-\underbrace{E\left\{\frac{\tau^N(\bX)\dot{\tau}_{\upsilon}^D(\bX)}{\tau^D(\bX)^2}\right\}|_{\upsilon=0}}_{(\ast_3)}.
\end{align*}
For the term $(\ast_1)$, we obtain that 
\begin{align*}
    (\ast_1)=&E\left\{(\tau(\bX)-\widetilde{\tau})S(\bX)\right\},
\end{align*}
as $E\{S(\bX)\}=0$. For the term $(\ast_2)$, applying the chain rule again yields
\begin{align*}
    (\ast_2)=&\underbrace{E\left\{\frac{e(\bC)\Delta_m (\bC)}{\tau^D(\bX)}S(\widetilde{\bX}|\bX)\right\}}_{(\ast_4)}+\underbrace{E\left\{\frac{\dot{e}_{\upsilon}(\bC)\Delta_m(\bC)}{\tau^D(\bX)}\right\}|_{\upsilon=0}}_{(\ast_5)}+\underbrace{E\left\{\frac{e(\bC)\dot{\Delta}_{m,\upsilon}(\bC)}{\tau^D(\bX)}\right\}|_{\upsilon=0}}_{(\ast_6)},
\end{align*}
where $\Delta_m(\bC):=m_{1d_1}(\bC)-m_{0d_0}(\bC)$ denotes the contrast between the two conditional outcome-regression functions. Similarly, applying the chain rule to $(\ast_3)$ yields
\begin{align*}
    (\ast_3)=&\underbrace{E\left\{\frac{\tau^N(\bX)}{\tau^D(\bX)^2}e(\bC)S(\widetilde{\bX}|\bX)\right\}}_{(\ast_7)}+\underbrace{E\left\{\frac{\tau^N(\bX)}{\tau^D(\bX)^2}\dot{e}_{\upsilon}(\bC)|_{\upsilon=0}\right\}}_{(\ast_8)}.
\end{align*}
Combining the terms $(\ast_4)$ and $(\ast_7)$ yields
\begin{align*}
    (\ast_4)-(\ast_7)=&E\left\{\frac{\Delta_m(\bC)-\tau(\bX)}{\tau^D(\bX)}e(\bC)S(\widetilde{\bX}|\bX)\right\}.
\end{align*}
For the term $(\ast_5)$, Lemma \ref{lemma:eif-eg(x),pswith theta} implies 
\begin{align*}
    (\ast_5)=&E\left\{\frac{\Delta_m(\bC)}{\tau^D(\bX)}E\{\phi^e(D,Z,\bC)S(D|Z,\bC)|\bC\}\right\}\\
    =&E\left\{\frac{\Delta_m(\bC)}{\tau^D(\bX)}\phi^e(D,Z,\bC)S(D|Z,\bC)\right\}.
\end{align*}
Similarly, Lemma \ref{lemma:eif-eg(x),pswith theta} also implies that the term $(\ast_8)$ can be further expressed as
\begin{align*}
    (\ast_8)=&E\left\{\frac{\tau^N(\bX)}{\tau^D(\bX)^2}E\{\phi^e(D,Z,\bC)S(D|Z,\bC)|\bC\}\right\}\\
    =&E\left\{\frac{\tau(\bX)}{\tau^D(\bX)}\phi^e(D,Z,\bC)S(D|Z,\bC)\right\}.
\end{align*}
For the term $(\ast_6)$, Lemma \ref{lemma:mzdz(c)} implies
\begin{align*}
    (\ast_6)=&E\left\{\frac{e(\bC)E\{\phi_{d_0d_1}^m(\mO)S(Y|Z,D,\bC)|\bC\}}{\tau^D(\bX)}\right\}\\
    =&E\left\{\frac{e(\bC)\phi_{d_0d_1}^m(\mO)}{\tau^D(\bX)}S(Y|Z,D,\bC)\right\}.
\end{align*}
It is straightforward to verify that 
\begin{align*}
   &\tau(\bX)-\widetilde{\tau}\in \mathcal{F}_1,\\
   &\frac{\Delta_m(\bC)-\tau(\bX)}{\tau^D(\bX)}e(\bC)\in \mathcal{F}_2,\\
   &\frac{\Delta_m(\bC)}{\tau^D(\bX)}\phi^e(D,Z,\bC)\in \mathcal{F}_4,\\
   &-\frac{\tau(\bX)}{\tau^D(\bX)}\phi^e(D,Z,\bC)\in \mathcal{F}_4,\\
   &\frac{e(\bC)\phi_{d_0d_1}^m(\mO)}{\tau^D(\bX)}\in \mathcal{F}_5.
\end{align*}
Therefore, the efficient influence function for the smoothed CPCE parameter $\widetilde{\tau}_{d_0d_1}$ is given by
\begin{align*}
    \varphi_{d_0d_1}(\mO) &= -\widetilde{\tau}_{d_0d_1}+\tau_{d_0d_1}(\bX)+\frac{1}{\tau_{d_0d_1}^D(\bX)}\left[e_{d_0d_1}(\bC)\phi_{d_0d_1}^m(\mO)+ \right. \\
    & \qquad \left.  (e_{d_0d_1}(\bC)+\phi_{d_0d_1}^e(D,Z,\bC))\left(m_{1d_1}(\bC)-m_{0d_0}(\bC)-\tau_{d_0d_1}(\bX)\right)\right].
\end{align*}

\section{Additional technical results}\label{sec:additional-tech-results}
\subsection{Covariance matrix estimation}
The theorem below provides the rate of convergence for the covariance matrix estimator $\widehat{\mathbb{V}}$ and $\| \widehat{\mathbb{V}}^{1/2}\bb(\bx)\rVert$.

\begin{theorem}\label{thm:cov-matrix-rate}
Suppose that all the conditions in Corollary \ref{corollary:uniform-negligible} hold. In addition, assume that (i) $\sup\{t:F_{|\widehat{Y}^\ast-Y||\mathcal{F}_s^c}(t)<1\}<\infty$ and $\sup\{t:F_{(\widehat{Y}^\ast-Y)^2|\mathcal{F}_s^c}(t)<1\}<\infty$, where $F_{\bullet|\mathcal{F}_s^c}$ is the cumulative distribution function conditional on the training data; (ii) $m_{3n}+\sqrt{\log K}l_Kc_K\lesssim\sqrt{\log K}$; (iii) $(n^{1/\nu}+l_Kc_K)(\sqrt{\xi_K^2\log K/n}+\kappa_n^\ast)=o_{\mP}(1)$; and (iv) $\kappa_n^\dagger=o_{\mP}(1)$.
\begin{align*}  
\sup_{\bx}\left|\frac{\| \widehat{\mathbb{V}}^{1/2}\bb(\bx)\|_2}{\| \mathbb{V}^{1/2}\bb(\bx)\|_2}-1\right|\lessp \| \widehat{\mathbb{V}}-\mathbb{V}\|_{\text{op}}\lessp\kappa_n,
\end{align*}
where $\kappa_n^\dagger:=E\left\{\max_{1\leq i\leq n}(\widehat{Y}_i^\ast-Y^\ast_i)^2|\widehat{\bGamma}_s\right\}$, $\kappa_n^\ast:=E\left\{\max_{1\leq i\leq n}|\widehat{Y}^\ast-Y^\ast||\widehat{\bGamma}_s\right\}$, and 
\begin{align*}
    \kappa_n:=\left(\sqrt{\frac{\xi_K^2\log K}{n}}+\kappa_n^\ast\right)(n^{1/\nu}+l_Kc_K)+\kappa_n^\dagger=o_{\mP}(1).
\end{align*}    
\end{theorem}

\subsection{Validity of Bayesian bootstrap method}\label{supp;subsec:valid-weighted-boot}
The theorem below establishes the uniform linearization results for the multiplier bootstrap distribution and, moreover, shows that it can be approximated by the same Gaussian process as presented in Theorem~\ref{thm:uniform-strong-gaussian} of the main manuscript. To proceed, suppose that the multipliers $\{\mathfrak{w}_i\}_{i=1}^n$ are independent of the data-generating process and are independent and identically distributed as $\mathfrak{w}_i \sim \exp(1)$.

\begin{lemma}[\emph{Uniform linearization for multiplier bootstrap}]\label{lemma:uniform-linearization-boot}
Under Assumptions \ref{asp:regularityL2}–\ref{assump:uniform-further-boundedness} and that $\allowbreak (\xi_k(\log n)^{1/2})^{2\nu/(\nu-2)}\lesssim1$, the estimator $\widehat{\bbeta}$ is asymptotically linear with 
\begin{align*}
     &\sqrt{n}\widetilde{\bb}(\bx)^\top(\widehat{\bbeta}^b-\widehat{\bbeta})=\widetilde{\bb}(\bx)^\top(\bH+\lambda\bP)^{-1}\mG_n\{(\mathfrak{w}-1)\bb(\bX)(Y^\ast-\bb(\bX)^\top\bbeta_0)\}+\text{Rem}^b_{1n}(\widetilde{\bb}(\bx)),
\end{align*}
where the remainder term $\text{Rem}^b_{1n}(\widetilde{\bb}(\bx))$ satisfies
\begin{align*}
    &\sup_{\bx}|\text{Rem}^b_{1n}(\widetilde{\bb}(\bx))|\\
    \lessp&\frac{1}{1+\lambda\omega_1(\bP)}\left\{\left(\frac{\xi_K\sqrt{\log K}}{1+2\lambda\omega_1(\bP)}+\sqrt{n}\right)(\lambda\omega_K(\bP)/\sqrt{n}+\log n(m_{1n}+m_{2n}))+\frac{m_{3n}\log n }{1+2\lambda\omega_1(\bP)}\right\}.
\end{align*}
\end{lemma}

\begin{theorem}[\emph{Strong Gaussian process approximation for Bayesian bootstrap}]\label{thm:uniform-strong-gaussian-boot}
Suppose Assumptions \ref{asp:regularityL2}–\ref{assump:uniform-further-boundedness} and all conditions in Corollary \ref{corollary:uniform-negligible} hold. Further, assume (i) $\lambda\omega_{\max}(1+m_{0n}+\sqrt{K/n})=o(1)$; (ii) $m_{3n}=o_{\mP}(a_n^{-1})$; (iii) $1\lesssim\inf_{\bx} E\{\Omega^2|\bX=\bx\}$; (iv) $a_n^6K^4\xi_K^2(1+l_K^3c_K^3)^2(\log n)^2/n=o(1)$; and (v) $\sup_{\bx}\sqrt{n}|r(\bx)|/\| \mathbb{V}^{1/2}\bb(\bx)\|_2=o(a_n^{-1})$. Then the following Gaussian process strong approximation holds in $L^\infty$ norm
\begin{align*}
T^b_n(\bx)=_d\frac{\bb(\bx)^\top\mathbb{V}^{1/2}}{\| \mathbb{V}^{1/2}\bb(\bx)\|_2}\mathcal{N}(0,\boldsymbol{I}_K)+o_{\mP}(a_n^{-1}),
\end{align*}
where $\{T^b_n(\bx):\bx\in\mathcal{X}\}$ is the bootstrap analogue of the targeting $t$-process $\{T_n(\bx):\bx\in\mathcal{X}\}$ with 
\begin{align*}
    T^b_n(\bx)=\sqrt{n}\frac{\widehat{\tau}^b_{d_0d_1}(\bx)-\widehat{\tau}_{d_0d_1}(\bx)}{\| \mathbb{V}^{1/2}\bb(\bx)\|_2}.
\end{align*}
\end{theorem}

\subsection{Validity of Gaussian bootstrap}
\begin{theorem}[\emph{Strong approximation of suprema}]\label{thm:strong-gaussian-appro-suprema}
Under the same set of conditions as outlined in Theorem~\ref{thm:honest-cov-gauss-boot-ci} of the main manuscript. Then 

\begin{align*}
    \sup_{\bx}|T_n(\bx)|=_d\sup_{\bx}\left|\frac{\bb(\bx)^\top\mathbb{V}^{1/2}}{\| \mathbb{V}^{1/2}\bb(\bx)\|_2}\mathcal{N}(0,\boldsymbol{I}_K)\right|+o_{\mP}\left(\frac{1}{\sqrt{\log K}}\right).
\end{align*}
    
\end{theorem}

\section{Proofs of Theorems and Propositions}\label{sec:proofs-all}
\subsection{Proof of Theorem \ref{thm:least-squares-series-density-ratio}}\label{sec:proofs-all;ss:proof-thm1}
\begin{proof}
By triangle inequality and Assumption \ref{asp:regularityL2}(b),
\begin{align*}
    &\| \widehat{\tau}_{d_0d_1}-\tau_{d_0d_1}\|_{\mP,2}\leq \| \bb(\bX)^\top(\widehat{\bbeta}-\bbeta_0)\|_{\mP,2} +\| \bb(\bX)^\top\bbeta_0-\tau_{d_0d_1}(\bX)\|_{\mP,2}\leq \| \bb(\bX)^\top (\widehat{\bbeta}-\bbeta_0) \|_{\mP,2}+c_{K}.
\end{align*}
Without loss of generality, Assumption \ref{asp:regularityL2}(a) implies that the basis functions can be orthonormalized, allowing us to assume $\bH= \boldsymbol{I}_{K}$ to simplify the exposition \citep{newey1997convergence}. Consequently, $\| \bb(\bX)^\top (\widehat{\bbeta}-\bbeta_0) \|_{\mP,2}=\| \widehat{\bbeta}-\bbeta_0\|_2$. Under Assumption \ref{asp:regularityL2}(c), the Rudelson’s law of large numbers for matrices (referred to as \emph{Rudelson’s LLN} hereafter) or Lemma 6.2 in \cite{belloni2015some} implies that 
\begin{align*}
    \| \widehat{\bH}-\bH\|_\text{op}\lessp\sqrt{\frac{\xi_K^2\log K}{n}}=o_{\mP}(1).
\end{align*}
Then, by Weyl’s inequality, we have $|\omega_k(\widehat{\bH}) - \omega_k(\bH)| \leq \| \widehat{\bH} - \bH \|_{\text{op}}$ for all $1\leq k\leq K$. Combined with Assumption \ref{asp:regularityL2}(a) and $\| \widehat{\bH}-\bH\|_\text{op}=o_{\mP}(1)$, this implies that 
\begin{align*}
    |\omega_k(\widehat{\bH})|\leq |\omega_k(\widehat{\bH}) - \omega_k(\bH)|+|\omega_k(\bH)|\leq |\omega_k(\widehat{\bH}) - \omega_k(\bH)|+1\lessp 1+o_{\mP}(1)=O_{\mP}(1).
\end{align*}
Moreover, \cite{belloni2015some} showed in the proof of Theorem 4.1 that, $\text{wp}\to1$, all eigenvalues of \(\widehat{\bH}\) are bounded below by $1/2$. Thus, $\text{wp}\to1$, the eigenvalues of $\widehat{\bH}$ remain uniformly bounded above and away from zero, and hence, invertible. 
Thus, $\text{wp}\to1$, the eigenvalues for $\widehat{\bH}+\lambda \bP$ are also all bounded above and away from zero, provided that $\lambda$ is bounded above. Define $\bh$ as the population analogue of $\widehat{\bh}$, i.e., $\bh=E\{\bb(\bX)Y^\ast\}=E\{\bb(\bX)\tau_{d_0d_1}(\bX)\}$. By triangle inequality and property of the operator norm, it follows that
\begin{align}
        \| \widehat{\bbeta}-\bbeta_0\|_2
        =&\| (\widehat{\bH}+\lambda \bP)^{-1} [\widehat{\bh}-(\widehat{\bH}+\lambda \bP)\bbeta_0]\|_2\nonumber\\
        \lessp&\| (\widehat{\bH}+\lambda \bP)^{-1}\|_{\text{op}}\| \widehat{\bh}-(\widehat{\bH}+\lambda \bP)\bbeta_0\|_2\nonumber\\
        \lessp&\frac{\| \widehat{\bh}-\widehat{\bH}\bbeta_0\|_2+\lambda\| \bP\bbeta_0\|_2}{\omega_1(\widehat{\bH}+\lambda\bP)}\nonumber\\
        \lessp&\frac{\| \widehat{\bh}-\widehat{\bH}\bbeta_0\|_2+\lambda\| \bP\bbeta_0\|_2}{\omega_1(\widehat{\bH})+\lambda\omega_1(\bP)}\nonumber\\
        \lessp&2\frac{\| \widehat{\bh}-\widehat{\bH}\bbeta_0\|_2+\lambda\| \bP\bbeta_0\|_2}{1+2\lambda\omega_1(\bP)}.\label{eq:two-components-L2-rate-bound}
\end{align}
To proceed, we first bound the term $\| \widehat{\bh}-\widehat{\bH}\bbeta_0\|_2$. Triangle inequality implies that
\begin{align*}
     \| \widehat{\bh}-\widehat{\bH}\bbeta_0\|_2
     =&\| \mP_n\{\bb(\bX)(\widehat{Y}^\ast-Y^\ast+Y^\ast-\tau_{d_0d_1}(\bX)+\tau_{d_0d_1}(\bX)-\bb(\bX)^\top\bbeta_0)\}\|_2\\
     \leq&\underbrace{\| \mP_n\{\bb(\bX)(\widehat{Y}^\ast-Y^\ast)\}\|_2}_{(\ast_1)} +\underbrace{\| \mP_n\{\bb(\bX)(Y^\ast-\tau_{d_0d_1}(\bX))\}\|_2}_{(\ast_2)}+\underbrace{\| \mP_n\{\bb(\bX)r_\tau(\bX)\}\|_2}_{(\ast_3)}.
\end{align*}
By sample-cross-fitting and triangle inequality, the term ($\ast_1$) can be written as
\begin{align*}
    (\ast_1)=& \| \sum_{s=1}^S\frac{n_s}{n}\mP_{n_s}\{\bb(\bX)(\widehat{Y}^\ast-Y^\ast)\}\|_2\\
\leq& \underbrace{\| \sum_{s=1}^S\frac{n_s}{n}(\mP_{n_s}\{\bb(\bX)(\widehat{Y}^\ast-Y^\ast)\}-\mP\{\bb(\bX)(\widehat{Y}^\ast-Y^\ast)|\widehat{\bGamma}_s\})\|_2}_{(\ast_4)}+\underbrace{\| \sum_{s=1}^S\frac{n_s}{n}\mP\{\bb(\bX)(\widehat{Y}^\ast-Y^\ast)|\widehat{\bGamma}_s\}\|_2}_{(\ast_5)}.
\end{align*}

Let $Y^\ast(\bGamma_1,\widehat{\bGamma}_2)$ denote the fitted pseudo-outcome under the true base nuisance parameters $\bGamma_1$ and the estimated intermediate nuisance parameters $\widehat{\bGamma}_2$. To bound $(\ast_5)$, 
 we can decompose the expected bias as follows:
\begin{align*}
&\mP\{\bb(\bX)(\widehat{Y}^\ast-Y^\ast)|\widehat{\bGamma}_s\}\\
=&\mP\{\bb(\bX)(\widehat{Y}^\ast-Y^\ast(\bGamma_1,\widehat{\bGamma}_2))|\widehat{\bGamma}_s\}+\mP\{\bb(\bX)(Y^\ast(\bGamma_1,\widehat{\bGamma}_2)-Y^\ast)|\widehat{\bGamma}_s\}\\
=&\underbrace{E\left\{\bb(\bX)\frac{E\{\widehat{\varphi}^N_{d_0d_1}-\varphi^N_{d_0d_1}|\bX,\mathcal{F}_s^c\}-\widehat{\tau}_{d_0d_1}E\{\widehat{\varphi}^D_{d_0d_1}-\varphi^D_{d_0d_1}|\bX,\mathcal{F}_s^c\}}{\widehat{\tau}_{d_0d_1}^D}|\mathcal{F}_s^c\right\}}_{(\ast_8)}+\\
&\underbrace{E\left\{\bb(\bX)\frac{(\widehat{\tau}^N_{d_0d_1}-\tau^N_{d_0d_1})(\widehat{\tau}_{d_0d_1}^D-\tau_{d_0d_1}^D)}{(\widehat{\tau}_{d_0d_1}^D)^2}|\mathcal{F}_s^c\right\}-E\left\{\bb(\bX)\frac{\tau^N_{d_0d_1}(\widehat{\tau}_{d_0d_1}^D-\tau_{d_0d_1}^D)^2}{(\widehat{\tau}_{d_0d_1}^D)^2\tau_{d_0d_1}^D}|\mathcal{F}_s^c\right\}}_{(\ast_9)}
\end{align*}
where $\widehat{\varphi}^N_{d_0d_1}:=\varphi^N_{d_0d_1}(\mO;\widehat{\bGamma}_s)$ and $\widehat{\varphi}^D_{d_0d_1}:=\varphi^D_{d_0d_1}(D,Z,\bC;\widehat{\bGamma}_s)$. For the term $(\ast_8)$, it follows that
\begin{align*}
      &\| (\ast_8)\|_2^2\\
    \leq&\epsilon_1^2(\| E\left\{\bb(\bX)(\widehat{\varphi}^N_{d_0d_1}-\varphi^N_{d_0d_1})|\mathcal{F}_s^c\}\|_2^2+\epsilon_1^2\epsilon_2^2\| E\{\bb(\bX)(\widehat{\varphi}^D_{d_0d_1}-\varphi^D_{d_0d_1})|\mathcal{F}_s^c\right\}\|_2^2)\\
    \lessp&\sum_{k=1}^K[(E\{b_k(\bX)(\widehat{\varphi}^N_{d_0d_1}-\varphi^N_{d_0d_1})|\mathcal{F}_s^c\})^2+(E\{b_k(\bX)(\widehat{\varphi}^D_{d_0d_1}-\varphi^D_{d_0d_1})|\mathcal{F}_s^c\})^2]\\
    \lessp&\sum_{k=1}^K[(E\{(\widehat{\varphi}^N_{d_0d_1}-\varphi^N_{d_0d_1})|\mathcal{F}_s^c\})^2+(E\{(\widehat{\varphi}^D_{d_0d_1}-\varphi^D_{d_0d_1})|\mathcal{F}_s^c\})^2]\\
    \lessp&K[\sum_{z=0,1}\eta(m_{zd_z})^2(\eta(\pi)^2+\sum_{j=0,1}\eta(p_j)^2)+\eta(p_0)^2\eta(p_1)^2+\eta(\pi)^2\sum_{j=0,1}\eta(p_j)^2]
\end{align*}
where the first and third inequalities follow from the triangle inequality and 
Assumption \ref{asp:regularityL2}(d), and the final inequality follows from the Cauchy–Schwarz inequality and Equation (S8.21) in the Supplementary Material to \cite{tong2025semiparametric}. For the term $(\ast_9)$, we have that
\begin{align*}
    &\| (\ast_9)\|_2^2\\
    \leq&\left\| E\left\{\bb(\bX)\frac{(\widehat{\tau}^N_{d_0d_1}-\tau^N_{d_0d_1})(\widehat{\tau}_{d_0d_1}^D-\tau_{d_0d_1}^D)}{(\widehat{\tau}_{d_0d_1}^D)^2}|\mathcal{F}_s^c\right\}\right\|_2^2+\left\| E\left\{\bb(\bX)\frac{\tau^N_{d_0d_1}(\widehat{\tau}_{d_0d_1}^D-\tau_{d_0d_1}^D)^2}{(\widehat{\tau}_{d_0d_1}^D)^2\tau_{d_0d_1}^D}|\mathcal{F}_s^c\right\}\right\|_2^2\\
    =&\sum_{k=1}^K\left[\left(E\left\{b_k(\bX)\frac{(\widehat{\tau}^N_{d_0d_1}-\tau^N_{d_0d_1})(\widehat{\tau}_{d_0d_1}^D-\tau_{d_0d_1}^D)}{(\widehat{\tau}_{d_0d_1}^D)^2}|\mathcal{F}_s^c\right\}\right)^2+\left(E\left\{b_k(\bX)\frac{\tau^N_{d_0d_1}(\widehat{\tau}_{d_0d_1}^D-\tau_{d_0d_1}^D)^2}{(\widehat{\tau}_{d_0d_1}^D)^2\tau_{d_0d_1}^D}|\mathcal{F}_s^c\right\}\right)^2\right]\\
    \leq &\epsilon_2^2\sum_{k=1}^K\left[\left(\epsilon_1^2E\left\{(\widehat{\tau}^N_{d_0d_1}-\tau^N_{d_0d_1})(\widehat{\tau}_{d_0d_1}^D-\tau_{d_0d_1}^D)|\mathcal{F}_s^c\right\}\right)^2+\left(\epsilon_1^3\epsilon_2^2E\left\{(\widehat{\tau}_{d_0d_1}^D-\tau_{d_0d_1}^D)^2|\mathcal{F}_s^c\right\}\right)^2\right]\\
    \lessp& K\left(\eta(\tau_{d_0d_1}^N)^2\eta(\tau_{d_0d_1}^D)^2+\eta(\tau_{d_0d_1}^D)^4\right),
\end{align*}
where the second inequality follows from Assumption \ref{asp:regularityL2}(d) and the last inequality follows from Cauchy-Schwarz. Thus, it follows that 
\begin{align*}
   (\ast_5)\lessp&\underbrace{\sqrt{K} \left[\sum_{z=0,1}\eta(m_{zd_z})\left(\eta(\pi)+\sum_{j=0,1}\eta(p_j)\right)+\eta(p_0)\eta(p_1)+\eta(\pi)\sum_{j=0,1}\eta(p_j)\right]}_{m_{2n}^{(1)}}+\\
   &\underbrace{\sqrt{K}\left(\eta(\tau_{d_0d_1}^N)\eta(\tau_{d_0d_1}^D)+\eta(\tau_{d_0d_1}^D)^2\right)}_{m_{2n}^{(2)}}\\
   =&m_{2n}.
\end{align*}

To bound the term $(\ast_4)$, we note that
\begin{align*}
&E\{\| (\mP_{n_s}\{\bb(\bX)(\widehat{Y}^\ast-Y^\ast)\}-\mP\{\bb(\bX)(\widehat{Y}^\ast-Y^\ast)|\widehat{\bGamma}_s\})\|_2^2|\widehat{\bGamma}_s\}\\
\leq&n^{-1}E\{\| \bb(\bX)(\widehat{Y}^\ast-Y^\ast)\|_2^2|\widehat{\bGamma}_s\}\\
\leq&n^{-1}\xi_K^2E\{(\widehat{Y}^\ast-Y^\ast)^2|\widehat{\bGamma}_s\},
\end{align*}
where the first inequality follows from the fact that $\{\widehat{Y}^\ast_i\}_{i\in\mathcal{F}_s}$ are independent and identically distributed conditional on the training sample $\mathcal{F}_s^c$. We consider the following decomposition for the second moment of the bias, conditional on the training sample:
\begin{align*}
    &E\{(\widehat{Y}^\ast-Y^\ast)^2|\widehat{\bGamma}_s\}\\
=&\underbrace{E\{(E\{\widehat{Y}^\ast|\bX,\widehat{\bGamma}_s\}-E\{Y^\ast|\bX,\widehat{\bGamma}_s\})^2|\widehat{\bGamma}_s\}}_{(\ast_6)}+\\
&\underbrace{E\{[(\widehat{Y}^\ast-E\{\widehat{Y}^\ast|\bX,\widehat{\bGamma}_s\})-(Y^\ast-E\{Y^\ast|\bX,\widehat{\bGamma}_s\})]^2|\widehat{\bGamma}_s\}}_{(\ast_7)},
\end{align*}
where the expected cross-product term vanishes by LOTE. We can bound the term $(\ast_6)$ as follows:
\begin{align*}
    (\ast_6)=&E\{(E\{\widehat{Y}^\ast-Y^\ast(\bGamma_1,\widehat{\bGamma}_2)|\bX,\widehat{\bGamma}_s\}+E\{Y^\ast(\bGamma_1,\widehat{\bGamma}_2)-Y^\ast|\bX,\widehat{\bGamma}_s\})^2|\widehat{\bGamma}_s\}\\
    \leq&2E\{E\{\widehat{Y}^\ast-Y^\ast(\bGamma_1,\widehat{\bGamma}_2)|\bX,\widehat{\bGamma}_s\}^2+E\{Y^\ast(\bGamma_1,\widehat{\bGamma}_2)-Y^\ast|\bX,\widehat{\bGamma}_s\}^2|\widehat{\bGamma}_s\}\\
    \lessp&E\left\{\left(\frac{E\{\widehat{\varphi}^N_{d_0d_1}-\varphi^N_{d_0d_1}|\bX,\mathcal{F}_s^c\}-\widehat{\tau}_{d_0d_1}E\{\widehat{\varphi}^D_{d_0d_1}-\varphi^D_{d_0d_1}|\bX,\mathcal{F}_s^c\}}{\widehat{\tau}_{d_0d_1}^D}\right)^2|\mathcal{F}_s^c\right\}+\\
    &E\left\{\left(\frac{(\widehat{\tau}^N_{d_0d_1}-\tau^N_{d_0d_1})(\widehat{\tau}_{d_0d_1}^D-\tau_{d_0d_1}^D)}{(\widehat{\tau}_{d_0d_1}^D)^2}-\frac{\tau^N_{d_0d_1}(\widehat{\tau}_{d_0d_1}^D-\tau_{d_0d_1}^D)^2}{(\widehat{\tau}_{d_0d_1}^D)^2\tau_{d_0d_1}^D}\right)^2|\mathcal{F}_s^c\right\}\\
     \lessp&E\left\{\left(\widehat{\varphi}^N_{d_0d_1}-\varphi^N_{d_0d_1}\right)^2+\left(\widehat{\varphi}^D_{d_0d_1}-\varphi^D_{d_0d_1}\right)^2|\mathcal{F}_s^c\right\}+\\
    &E\left\{(\widehat{\tau}^N_{d_0d_1}-\tau^N_{d_0d_1})^2(\widehat{\tau}_{d_0d_1}^D-\tau_{d_0d_1}^D)^2+(\widehat{\tau}_{d_0d_1}^D-\tau_{d_0d_1}^D)^4|\mathcal{F}_s^c\right\}\\   
     \lessp&     \eta(\tau_{d_0d_1}^D)^2+     \eta(\tau_{d_0d_1}^N)^2+\eta(\pi)^2+\sum_{z=0,1}(\eta(p_z)^2+\eta(m_{zd_z})^2),
\end{align*}
where the first inequality follows from the quadratic-mean arithmetic-mean inequality, the third inequality follows from the quadratic-mean arithmetic-mean inequality, Assumption \ref{asp:regularityL2}(d), and Jensen's inequality, and last inequality follows from Assumption \ref{asp:regularityL2}(d) and the proof of Theorem 3 in \cite{tong2025semiparametric}. To simplify bounding the term $(\ast_7)$, we further define
\begin{align*}
    \widetilde{Y}^\ast:=\tau_{d_0d_1}+\frac{\widehat{\varphi}^N_{d_0d_1}-\tau_{d_0d_1}\widehat{\varphi}^D_{d_0d_1}}{\widehat{\tau}^D_{d_0d_1}}.
\end{align*}
Then we have that
\begin{align*}
    &(\ast_7)\\
    \lessp&E\{[(\widehat{Y}^\ast-E\{\widehat{Y}^\ast|\bX,\widehat{\bGamma}_s\})-(\widetilde{Y}^\ast-E\{\widetilde{Y}^\ast|\bX,\widehat{\bGamma}_s\})]^2|\widehat{\bGamma}_s\}+\\
    &E\{[(\widetilde{Y}^\ast-E\{\widetilde{Y}^\ast|\bX,\widehat{\bGamma}_s\})-(Y^\ast(\bGamma_1,\widehat{\bGamma}_2)-E\{Y^\ast(\bGamma_1,\widehat{\bGamma}_2)|\bX,\widehat{\bGamma}_s\})]^2|\widehat{\bGamma}_s\}+\\
    &E\{[(Y^\ast(\bGamma_1,\widehat{\bGamma}_2)-E\{Y^\ast(\bGamma_1,\widehat{\bGamma}_2)|\bX,\widehat{\bGamma}_s\})-(Y^\ast-E\{Y^\ast|\bX,\widehat{\bGamma}_s\})]^2|\widehat{\bGamma}_s\}\\
    =&E\left\{\left(\frac{\widehat{\tau}_{d_0d_1}^N-\widehat{\tau}_{d_0d_1}^N}{(\widehat{\tau}_{d_0d_1}^D)^2}-\frac{{\tau}_{d_0d_1}^N(\widehat{\tau}_{d_0d_1}^D-{\tau}_{d_0d_1}^D)}{(\widehat{\tau}_{d_0d_1}^D)^2{\tau}_{d_0d_1}^D}\right)^2\left(\widehat{\varphi}^D_{d_0d_1}-E\{\widehat{\varphi}^D_{d_0d_1}|\bX,\widehat{\bGamma}_s\}\right)^2|\widehat{\bGamma}_s\right\}+\\
    &E\left\{\left(\frac{\widehat{\tau}^D_{d_0d_1}-\tau_{d_0d_1}^D}{\widehat{\tau}^D_{d_0d_1}\tau_{d_0d_1}^D}\right)^2\text{Var}(\widehat{\varphi}_{d_0d_1}^N-\tau_{d_0d_1}\widehat{\varphi}_{d_0d_1}^D|\bX,\widehat{\bGamma}_s)|\widehat{\bGamma}_s\right\}+\\
    &\small{E\left\{\frac{\left[\widehat{\varphi}^N_{d_0d_1}-\varphi^N_{d_0d_1}-E\{\widehat{\varphi}^N_{d_0d_1}-\varphi^N_{d_0d_1}|\bX,\widehat{\bGamma}_s\}-\tau_{d_0d_1}(\widehat{\varphi}^D_{d_0d_1}-\varphi^D_{d_0d_1}-E\{\widehat{\varphi}^D_{d_0d_1}-\varphi^D_{d_0d_1}|\bX,\widehat{\bGamma}_s\})\right]^2}{(\tau_{d_0d_1}^D)^2}|\widehat{\bGamma}_s\right\}}\\
    \lessp&E\left\{\left(\frac{\widehat{\tau}_{d_0d_1}^N-\widehat{\tau}_{d_0d_1}^N}{(\widehat{\tau}_{d_0d_1}^D)^2}-\frac{{\tau}_{d_0d_1}^N(\widehat{\tau}_{d_0d_1}^D-{\tau}_{d_0d_1}^D)}{(\widehat{\tau}_{d_0d_1}^D)^2{\tau}_{d_0d_1}^D}\right)^2\text{Var}(\widehat{\varphi}^D_{d_0d_1}|\bX,\widehat{\bGamma}_s)|\widehat{\bGamma}_s\right\}+\\
    &E\left\{\left(\widehat{\tau}^D_{d_0d_1}-\tau_{d_0d_1}^D\right)^2|\widehat{\bGamma}_s\right\}+\\
    &E\{\text{Var}(\widehat{\varphi}^N_{d_0d_1}-\varphi^N_{d_0d_1}|\bX,\widehat{\bGamma}_s)|\widehat{\bGamma}_s\}+E\{\text{Var}(\widehat{\varphi}^D_{d_0d_1}-\varphi^D_{d_0d_1}|\bX,\widehat{\bGamma}_s)|\widehat{\bGamma}_s\}\\
    \lessp&\eta(\tau_{d_0d_1}^D)^2+     \eta(\tau_{d_0d_1}^N)^2+E\{(\widehat{\varphi}^N_{d_0d_1}-\varphi^N_{d_0d_1})^2|\widehat{\bGamma}_s\}+E\{(\widehat{\varphi}^D_{d_0d_1}-\varphi^D_{d_0d_1})^2|\widehat{\bGamma}_s\}\\
    \lessp&\eta(\tau_{d_0d_1}^D)^2+     \eta(\tau_{d_0d_1}^N)^2+\eta(\pi)^2+\sum_{z=0,1}(\eta(p_z)^2+\eta(m_{zd_z})^2),
\end{align*}
where the first inequality follows from the quadratic-mean–arithmetic-mean inequality, the first equality from the LOTE, the second inequality from Assumption \ref{asp:regularityL2}(d) and the proof of Theorem 3 in \cite{tong2025semiparametric}, and the final two inequalities from reasoning analogous to that used for bounding $(\ast_6)$.

To bound the term $(\ast_2)$, we note that
\begin{align*}
    (E\{(\ast_2)^2\})^{1/2}\leq &(E\{(Y^\ast-\tau_{d_0d_1}(\bX))^2\bb(\bX)^\top\bb(\bX)/n\})^{1/2}\\
    \leq &\epsilon_2 \sqrt{K/n},
\end{align*}
where the last inequality follows from Assumption \ref{asp:regularityL2}(d). By Markov's inequality, $(\ast_2)=O_{\mP}(\sqrt{K/n})$. To bound the term $(\ast_3)$, we note that 
\begin{align*}
    (E\{(\ast_3)^2\})^{1/2}\leq &(E\{r_\tau^2\bb(\bX)^\top\bb(\bX)/n\})^{1/2}\\
    \leq& l_Kc_K(E\{\bb(\bX)^\top\bb(\bX)/n\})^{1/2}~~(\text{Assumption \ref{asp:regularityL2}(b)})\\
    =&l_Kc_K\sqrt{\text{trace}(\bH)/n}\\
    =&l_Kc_K\sqrt{\frac{K}{n}},
\end{align*}
which implies that $(\ast_3)=O_{\mP}(l_Kc_K\sqrt{{K}/{n}})$. Alternatively, to bound the term $(\ast_3)$, we note that 
\begin{align*}
    (E\{(\ast_3)^2\})^{1/2}\leq &(E\{r_\tau^2\|\bb(\bX)\|_2^2/n\})^{1/2}\\
    \leq& \xi_K(E\{r_\tau^2/n\})^{1/2})\\
    =&\xi_Kc_K/\sqrt{n},
\end{align*}
which implies that $(\ast_3)=O_{\mP}(\xi_Kc_K/\sqrt{n})$. Therefore, $(\ast_3)\lessp\min\{l_Kc_K\sqrt{{K}/{n}},\xi_Kc_K/\sqrt{n}\}=m_{0n}$. Eventually, we combine all the above results and conclude that 
\begin{align*}
    \| \bb(\bX)^\top (\widehat{\bbeta}-\bbeta_0)\|_{\mP,2}\lessp\sqrt{\frac{K}{n}}+ m_{0n}+m_{1n}+m_{2n}.
\end{align*}
Finally, we bound the regularization bias $\lambda\| \bP\bbeta_0\|_2$. 
By the Pythagorean theorem, it follows that 
\begin{align}
    E\lb \tau_{d_0d_1}(\bX)^2\rb=&E\left[\lb \bb(\bX)^\top\bbeta_0\rb^2\right]+E\left[r_\tau^2\right]\nonumber\\
    \geq&E\left[\lb \bb(\bX)^\top\bbeta_0\rb^2\right]\nonumber\\
    =&\bbeta_0^\top\bH\bbeta_0=\|\bbeta_0\|_2^2,\nonumber
\end{align}
which further implies 
\begin{align*}
    \lambda\| \bP\bbeta_0\|_2\leq&\lambda\| \bP\|_{\text{op}}\|\bbeta_0\|_2\\
    \leq&\lambda\| \bP\|_{\text{op}} \left[E\{\tau_{d_0d_1}(\bX)^2\}\right]^{1/2}\\
    \lessp& \lambda \| \bP\|_{\text{op}}~~~\text{(Assumption \ref{asp:regularityL2}(d))}\\
    \leq& \lambda \omega_K(\bP).
\end{align*}
\end{proof}

\subsection{Some useful lemmas and their proofs}
We state two key lemmas, largely following \cite{belloni2015some}, that are useful for the subsequent proofs. For clarity of exposition, we define 
$\mG_n\{V_i\} = \sqrt{n}\mP_n\{V - \mP\{V\}\}$ as the $\sqrt{n}$-scaled and mean-centered empirical measure. The first lemma establishes the pointwise linearization of the estimator $\widehat{\bbeta}$.
\begin{lemma}[\emph{Pointwise linearization}]\label{lemma:pointwise-linearization}
Under Assumption \ref{asp:regularityL2}(a)–(d), for any $\widetilde{\bb}$ in the unit ball $\{\widetilde{\bb}\in\mathbb{R}^K:\| \widetilde{\bb}\|_2=1\}$, the estimator $\widehat{\bbeta}$ is asymptotically linear with 
\begin{align*}
    &\sqrt{n}\widetilde{\bb}^\top(\widehat{\bbeta}-\bbeta_0)=\widetilde{\bb}^\top\bH^{-1}\mG_n\{\bb(\bX)(Y^\ast-\bb(\bX)^\top\bbeta_0)\}+\text{Rem}_{1n}(\widetilde{\bb}),    
\end{align*}
where the remainder term $\text{Rem}_{1n}$ is bounded as follows 
\begin{align*}
    \text{Rem}_{1n}\lessp&\frac{1}{1+\lambda\omega_1(\bP)}\left\{{\sqrt{n}(m_{1n}+m_{2n})}+\sqrt{\frac{\xi_K^2\log K}{n(1+2\lambda\omega_1(\bP))^2}}\left(1+{\sqrt{n}\sum_{j=1}^2m_{jn}}+l_Kc_K\sqrt{K}+\lambda\omega_K(\bP)\right)\right. \\
&\left. + \lambda\omega_K(\bP)\left(\sqrt{\frac{K}{n}}+m_{0n}+1\right)\right\}.
\end{align*}
\end{lemma}
\begin{proof}
Consider the following decomposition:
\begin{align*}
    &\sqrt{n}\widetilde{\bb}^\top(\widehat{\bbeta}-\bbeta_0)\\
    =&\sqrt{n}\widetilde{\bb}^\top(\widehat{\bH}+\lambda \bP)^{-1}\left\{\mP_n\{\bb(\bX)\widehat{Y}^\ast\}-(\widehat{\bH}+\lambda \bP)\bbeta_0\right\}\\
    =&\widetilde{\bb}^\top(\widehat{\bH}+\lambda \bP)^{-1}\left\{\sqrt{n}\mP_n\{\bb(\bX)(\widehat{Y}^\ast-Y^\ast)\}+\mG_n\{\bb(\bX)(Y^\ast-\tau_{d_0d_1}(\bX))\}+ \mG_n\{\bb(\bX)r_\tau(\bX)\}-\lambda\bP \bbeta_0\right\}\\
    =&\widetilde{\bb}^\top\bH^{-1}\mG_n\{\bb(\bX)(Y^\ast-\bb(\bX)^\top\bbeta_0)\}+\\
    &\underbrace{\widetilde{\bb}^\top[(\widehat{\bH}+\lambda \bP)^{-1}-(\bH+\lambda \bP)^{-1}]\mG_n\{\bb(\bX)(Y^\ast-\tau_{d_0d_1}(\bX))\}}_{(\ast_1)}+\\
    &\underbrace{\widetilde{\bb}^\top[(\widehat{\bH}+\lambda \bP)^{-1}-(\bH+\lambda \bP)^{-1}]\mG_n\{\bb(\bX)r_\tau(\bX)\}}_{(\ast_2)}+\\    
    &\underbrace{\sqrt{n}\widetilde{\bb}^\top(\bH+\lambda\bP)^{-1}\mP_n\{\bb(\bX)(\widehat{Y}^\ast-Y^\ast)\}}_{(\ast_3)}+\\
    &\underbrace{\sqrt{n}\widetilde{\bb}^\top[(\widehat{\bH}+\lambda \bP)^{-1}-(\bH+\lambda \bP)^{-1}]\mP_n\{\bb(\bX)(\widehat{Y}^\ast-Y^\ast)\}}_{(\ast_4)}-\\
    &\underbrace{\lambda\widetilde{\bb}^\top[(\widehat{\bH}+\lambda \bP)^{-1}-(\bH+\lambda \bP)^{-1}]\bP\bbeta_0}_{(\ast_5)}-\underbrace{\lambda\widetilde{\bb}^\top(\bH+\lambda \bP)^{-1}\bP\bbeta_0}_{(\ast_6)}+\\
    &\underbrace{\widetilde{\bb}^\top[(\bH+\lambda \bP)^{-1}-\bH^{-1}]\mG_n\{\bb(\bX)(Y^\ast-\bb(\bX)^\top\bbeta_0)\}}_{(\ast_7)},
\end{align*}
where the second equality follows from that 
\begin{align}
    E\{\bb(\bX)(Y^\ast-\tau_{d_0d_1}(\bX))\}=&E\{\bb(\bX)(E\{Y^\ast|\bX\}-\tau_{d_0d_1}(\bX))\}=0,\label{eq:mean-zero-Y-tau}\\
    E\{\bb(\bX)r_\tau(\bX)\}=&E\{\bb(\bX)\tau_{d_0d_1}(\bX)\}-E\{\bb(\bX)\bb(\bX)^\top\}\bbeta_0\nonumber\\
    =&E\{\bb(\bX)\tau_{d_0d_1}(\bX)\}-\bH\bbeta_0\nonumber\\
    =&E\{\bb(\bX)\tau_{d_0d_1}(\bX)\}-\bH\bH^{-1}E\{\bb(\bX)\tau_{d_0d_1}(\bX)\}=0.\nonumber
\end{align}
Therefore, the remainder term is the sum of seven components, i.e., $\text{Rem}_{1n}(\widetilde{\bb})=\sum_{j=1}^7(\ast_j)$. Thus, it suffices to analyze and bound each component. We repeatedly use the following results: 
\begin{align}
    &\|(\widehat{\bH}+\lambda \bP)^{-1}-(\bH+\lambda \bP)^{-1}\|_\text{op}\nonumber\\
    =&\| (\widehat{\bH}+\lambda\bP)^{-1}(\bH-\widehat{\bH})(\bH+\lambda\bP)^{-1}\|_\text{op}\nonumber\\
    =&\| (\widehat{\bH}+\lambda\bP)^{-1}\|_\text{op}\|\bH-\widehat{\bH}\|_\text{op}\|(\bH+\lambda\bP)^{-1}\|_\text{op}\nonumber\\
    \leq&\frac{1}{1+\lambda\omega_1(
    \bP
    )}\| (\widehat{\bH}+\lambda\bP)^{-1}\|_\text{op}\|\bH-\widehat{\bH}\|_\text{op}\nonumber\\
    \lessp&\frac{2}{(1+2\lambda\omega_1(\bP))(1+\lambda\omega_1(
    \bP
    ))}\|\widehat{\bH}-\bH\|_\text{op}\label{eq:regularized-norm-differ-bound-2}\\
    \lessp&\frac{\sqrt{\xi_K^2\log K/n}}{(1+2\lambda\omega_1(\bP))(1+\lambda\omega_1(
    \bP
    ))}\label{eq:regularized-norm-differ-bound-2}&,
\end{align}
where the first equality is due to the matrix identity $\boldsymbol{A}^{-1}-\boldsymbol{B}^{-1}=\boldsymbol{A}^{-1}(\boldsymbol{B}-\boldsymbol{A})\boldsymbol{B}^{-1}$ for any invertible matrices $\boldsymbol{A}$ and $\boldsymbol{B}$, and the next to last inequality follows from $\| (\widehat{\bH}+\lambda\bP)^{-1}\|_\text{op}=1/\omega_1(\widehat{\bH}+\lambda\bP)\leq 1/(\omega_1(\widehat{\bH})+\lambda\omega_1(\bP))\lessp 2/(1+2\lambda\omega_1(\bP))$ ($\omega_1(\widehat{\bH})>1/2$ $\text{wp}\to1$), and the last inequality follows from the Rudelson’s LLN. To bound the first term, we note that
\begin{align}
    &E\{(\ast_1)^2|\{\bX_i\}_1^n\}\nonumber\\
    =&\text{Var}\{(\ast_1)|\{\bX_i\}_1^n\}\nonumber\\
    =&\widetilde{\bb}^\top [(\widehat{\bH}+\lambda \bP)^{-1}-(\bH +\lambda \bP)^{-1}]\left\{1/n\sum_{i=1}^n\text{Var}\{\bb(\bX)(Y^\ast-\tau_{d_0d_1}(\bX))|\{\bX_i\}_1^n\}\right\} \nonumber\\
    &[(\widehat{\bH}+\lambda \bP)^{-1}-(\bH +\lambda \bP)^{-1}]\widetilde{\bb}\nonumber\\
    =&\widetilde{\bb}^\top [(\widehat{\bH}+\lambda \bP)^{-1}-(\bH +\lambda \bP)^{-1}]\left\{1/n\sum_{i=1}^n\bb(\bX)\bb(\bX)^\top E\{(Y^\ast-\tau_{d_0d_1}(\bX))^2|\bX\}\right\} \nonumber\\
    &[(\widehat{\bH}+\lambda \bP)^{-1}-(\bH +\lambda \bP)^{-1}]\widetilde{\bb}\nonumber\\
    \leq&\epsilon_2^2\widetilde{\bb}^\top [(\widehat{\bH}+\lambda \bP)^{-1}-(\bH +\lambda \bP)^{-1}]\widehat{\bH}[(\widehat{\bH}+\lambda \bP)^{-1}-(\bH +\lambda \bP)^{-1}]\widetilde{\bb}\nonumber\\
    \lessp&\frac{\| \widehat{\bH}\|_{\text{op}}\| \widehat{\bH}-\bH\|_{\text{op}}^2}{(1+2\lambda\omega_1(\bP))^2(1+\lambda\omega_1(\bP))^2}\nonumber\\
    \lessp&\frac{\xi_K^2\log K}{n(1+\lambda\omega_1(\bP))^2(1+2\lambda\omega_1(\bP))^2},\nonumber
\end{align}
where $\{\bX_i\}_1^n=\{\bX_1,\ldots,\bX_n\}$, the first equality follows from that $E\{(\ast_1)|\{\bX_i\}_1^n\}=0$ by Equation \eqref{eq:mean-zero-Y-tau}, the first inequality follows from Assumption \ref{asp:regularityL2}(d), the second to last inequality follows from Equation \eqref{eq:regularized-norm-differ-bound-2}, and the last inequality follows from Rudelson’s LLN such that $\|\widehat{\bH}-\bH\|_{\text{op}}\lessp \sqrt{\xi_K^2\log K/n}$ under Assumption \ref{asp:regularityL2}(c). By Markov's inequality, it follows that $(\ast_1)=O_{\mP}(\sqrt{\xi_K^2\log K/n}/(1+\lambda\omega_1(\bP))/(1+2\lambda\omega_1(\bP)))$. For the second term, Step 2 of the proof of Lemma 4.1 in \cite{belloni2015some} implies
\begin{align*}
    |(\ast_2)|\lessp &\| (\widehat{\bH}+\lambda \bP)^{-1}-(\bH+\lambda \bP)^{-1}\|_{\text{op}} l_Kc_K\sqrt{K}\\
    \lessp&\frac{\sqrt{\xi_K^2\log K/n}l_Kc_K\sqrt{K}}{(1+2\lambda\omega_1(\bP))(1+\lambda\omega_1(\bP))},
\end{align*}
where the second inequality follows from Equation \eqref{eq:regularized-norm-differ-bound-2}. 
To bound the third term, it follows that
\begin{align*}
    (\ast_3)=&\sqrt{n}\widetilde{\bb}^\top(\bH+\lambda\bP)^{-1}\mP_n\{\bb(\bX)(\widehat{Y}^\ast-Y^\ast)\}\\\leq &\sqrt{n}\| \widetilde{\bb} \|_2\| (\bH+\lambda\bP)^{-1}\|_{\text{op}}\| \mP_n\{\bb(\bX)(\widehat{Y}^\ast-Y^\ast)\}\|_2\\
    \lesssim&\frac{\sqrt{n}\| \mP_n\{\bb(\bX)(\widehat{Y}^\ast-Y^\ast)\}\|_2}{(1+\lambda\omega_1(\bP))}\\
    \lessp&\frac{\sqrt{n}(m_{1n}+m_{2n})}{(1+\lambda\omega_1(\bP))},
\end{align*}
where the first inequality follows from Cauchy-Schwarz, the second inequality follows from that $\| \widetilde{\bb} \|_2=1$ and that $\| (\bH+\lambda\bP)^{-1}\|_{\text{op}}=1/\omega_1(\bH+\lambda\bP)\leq1/(1+\lambda\omega_1(\bP))$, and the third inequality follows from the proof of Theorem \ref{thm:least-squares-series-density-ratio}. Similarly, for the term $(\ast_4)$, we obtatin that
\begin{align*}
    (\ast_4)=    &\sqrt{n}\widetilde{\bb}^\top[(\widehat{\bH}+\lambda \bP)^{-1}-(\bH+\lambda \bP)^{-1}]\mP_n\{\bb(\bX)(\widehat{Y}^\ast-Y^\ast)\}\\
    \lessp &\frac{\sqrt{n}\| \widetilde{\bb} \|_2\| \widehat{\bH}^{-1}-\bH^{-1}\|_{\text{op}}\| \mP_n\{\bb(\bX)(\widehat{Y}^\ast-Y^\ast)\}\|_2}{(1+2\lambda\omega_1(\bP))(1+\lambda\omega_1(
    \bP
    ))}\\
    \lessp&\sqrt{\xi_K^2\log K/n}\sqrt{n}(m_{1n}+m_{2n})/(1+2\lambda\omega_1(\bP))/(1+\lambda\omega_1(
    \bP
    )).
\end{align*}
where the first inequality follows from Cauchy-Schwarz and Equation \eqref{eq:regularized-norm-differ-bound-2}, and the second inequality follows from the proof of Theorem \ref{thm:least-squares-series-density-ratio} and the Rudelson’s LLN for Matrices under Assumption \ref{asp:regularityL2}(c). For the fifth term, we have that 
\begin{align*}
    (\ast_5)\leq&\lambda\| \widetilde{\bb} \|_2\| [(\widehat{\bH}+\lambda \bP)^{-1}-(\bH+\lambda \bP)^{-1}]\|_{\text{op}}\|\bP\|_{\text{op}}\| \bbeta_0\|_2\\
    \lessp&\lambda \| \bbeta_0\|_2O_{\mP}(\omega_K(\bP)\sqrt{\xi_K^2 \log K/n}/(1+2\lambda\omega_1(\bP))/(1+\lambda\omega_1(\bP)))\\
    \lessp& \lambda\omega_K(\bP)\sqrt{\xi_K^2 \log K/n}/(1+2\lambda\omega_1(\bP))/(1+\lambda\omega_1(\bP)),
\end{align*}
where the Second inequality follows from Equation \eqref{eq:regularized-norm-differ-bound-2} and the last inequality follows from Assumption \ref{asp:regularityL2}(c). Finally, for the sixth term, we have that
\begin{align*}
    (\ast_6)\leq&\frac{\lambda\omega_K(\bP)}{1+\lambda\omega_1(\bP)}\|\bbeta_0\|_2\lessp \frac{\lambda\omega_K(\bP)}{(1+\lambda\omega_1(\bP))}.
\end{align*}
Finally, for the seventh term, it follows that
\begin{align*}
        (\ast_7)\leq&\| \widetilde{\bb}\|_2\|(\bH+\lambda \bP)^{-1}-\bH^{-1}\|_{\text{op}}\|\mG_n\{\bb(\bX)(Y^\ast-\bb(\bX)^\top\bbeta_0)\}\|_2\\
        \leq&\lambda\|(\bH+\lambda \bP)^{-1}\|_{\text{op}}\| \bP\|_{\text{op}}\|\mG_n\{\bb(\bX)(Y^\ast-\bb(\bX)^\top\bbeta_0)\}\|_2\\
        \leq&\frac{\lambda\omega_K(\bP)}{1+\lambda\omega_1(\bP)}\|\mP_n\{\bb(\bX)(Y^\ast-\tau_{d_0d_1}(\bX))+\mP_n\{\bb(\bX)r_\tau(\bX)\}\|_2\\
        \lessp&\frac{\lambda\omega_K(\bP)}{1+\lambda\omega_1(\bP)}\left(\sqrt{\frac{K}{n}}+m_{0n}\right),
\end{align*}
where the first inequality is due to Cauchy-Schwarz and the second inequality follows from $(\bH+\lambda \bP)^{-1}-\bH^{-1}=(\bH+\lambda \bP)^{-1}(-\lambda\bP)\bH^{-1}$, and the last inequality follows from the proof of Theorem \ref{thm:least-squares-series-density-ratio}.
\end{proof}

The second lemma establishes the uniform linearization of the estimator $\widehat{\bbeta}$, which facilitates the derivation of its uniform convergence rate.

\begin{lemma}[\emph{Uniform linearization}]\label{lemma:uniform-linearization}
Under Assumptions \ref{asp:regularityL2}–\ref{assump:uniform-further-boundedness}, the estimator $\widehat{\bbeta}$ is asymptotically linear with 
\begin{align*}
     &\sqrt{n}\widetilde{\bb}(\bx)^\top(\widehat{\bbeta}-\bbeta_0)=\widetilde{\bb}(\bx)^\top(\bH+\lambda\bP)^{-1}\mG_n\{\bb(\bX)(Y^\ast-\bb(\bX)^\top\bbeta_0)\}+\text{Rem}_{1n}(\widetilde{\bb}(\bx)),\\
    &\sqrt{n}\widetilde{\bb}(\bx)^\top(\widehat{\bbeta}-\bbeta_0)=\widetilde{\bb}(\bx)^\top(\bH+\lambda\bP)^{-1}\mG_n\{\bb(\bX)(Y^\ast-\tau_{d_0d_1}(\bX))\}+\text{Rem}_{1n}(\widetilde{\bb}(\bx))+\text{Rem}_{2n}(\widetilde{\bb}(\bx)),   
\end{align*}
where the remainder terms $\text{Rem}_{1n}(\widetilde{\bb}(\bx))$ and $\text{Rem}_{2n}(\widetilde{\bb}(\bx))$ capture the impact of the unknown design and first-stage nuisance estimation errors, and the impact of misspecification error, respectively, and are bounded as follows. Define $m_{3n}=\sqrt{\xi_K^2\log K/n}(n^{1/\nu}\sqrt{\log K}+\sqrt{K}l_Kc_K)$, then
\begin{align*}
    &\sup_{\bx}|\text{Rem}_{1n}(\widetilde{\bb}(\bx))|\\
    \lessp&\frac{1}{1+\lambda\omega_1(\bP)}\left\{\left(\frac{\xi_K\sqrt{\log K}}{1+2\lambda\omega_1(\bP)}+\sqrt{n}\right)(\lambda\omega_K(\bP)/\sqrt{n} +m_{1n}+m_{2n})+\frac{m_{3n}}{1+2\lambda\omega_1(\bP)}\right\},\\
   &\sup_{\bx}|\text{Rem}_{2n}(\widetilde{\bb}(\bx))|\lessp\frac{\sqrt{\log K}l_Kc_K}{1+\lambda\omega_1(\bP)}.
\end{align*}
\end{lemma}
\begin{proof}
We follow the decomposition from the proof of Lemma \ref{lemma:pointwise-linearization} with $\text{Rem}_{1n}(\widetilde{\bb}(\bx))=\sum_{j=1}^6(\ast_j)$, where the seventh term is absorbed into the linear representation. For the first term, we adapt Step 1 of the proof of Lemma 4.2 in \cite{belloni2015some} as follows:
\begin{align*}
    \sup_{\bx}|(\ast_1)|\lessp&n^{1/\nu}\| (\widehat{\bH}+\lambda \bP)^{-1}-(\bH+\lambda \bP)^{-1}\|_{\text{op}}\| \widehat{\bH}\|_{\text{op}}^{1/2}\sqrt{\log K}\\
    \lessp&n^{1/\nu}\| (\widehat{\bH}+\lambda \bP)^{-1}-(\bH+\lambda \bP)^{-1}\|_{\text{op}}\sqrt{\log K}\\
    \lessp&n^{1/\nu}\frac{\sqrt{{\xi_K^2\log K}/{n}}\sqrt{\log K}}{(1+2\lambda\omega_1(\bP))(1+\lambda\omega_1(\bP))},
\end{align*}
where the second inequality follows from that $\| \widehat{\bH}\|_{\text{op}}=1/\omega_K(\widehat{\bH})=O_{\mP}(1)$ and the third inequality follows from Equation \eqref{eq:regularized-norm-differ-bound-2}. For the second term, we adapt Step 2 of the proof of Lemma 4.2 in \cite{belloni2015some} as follows:
\begin{align*}
    \sup_{\bx}|(\ast_2)|\lessp&\| (\widehat{\bH}+\lambda \bP)^{-1}-(\bH+\lambda \bP)^{-1}\|_{\text{op}}\| \mG_n\{\bb(\bX)r_\tau\}\|_2\\
    \lessp& \sqrt{\frac{\xi_K^2\log K}{n}}\frac{l_Kc_K\sqrt{K}}{(1+2\lambda\omega_1(\bP))(1+\lambda\omega_1(\bP))}.
\end{align*}
Therefore, we obtain that $\sup_{\bx}|\sum_{j=1}^2(\ast_j)|\lessp m_{3n}/(1+2\lambda\omega_1(\bP))/(1+\lambda\omega_1(\bP))$. The third term can be bounded as
\begin{align}
    \sup_{\bx}|(\ast_3)|=&\sup_{\bx}|\sqrt{n}\widetilde{\bb}(\bx)^\top(\bH+\lambda\bP)^{-1}\mP_n\{\bb(\bX)(\widehat{Y}^\ast-Y^\ast)\}|\nonumber\\
    \leq&\sqrt{n}\sup_{\bx}\| \widetilde{\bb}(\bx)\|_2\|(\bH+\lambda\bP)^{-1}\|_{\text{op}}\|\mP_n\{\bb(\bX)(\widehat{Y}^\ast-Y^\ast)\}\|_2\nonumber\\
    \leq&\frac{\sqrt{n}\sup_{\bx}\| \widetilde{\bb}(\bx)\|_2}{1+\lambda \omega_1(\bP)}\|\mP_n\{\bb(\bX)(\widehat{Y}^\ast-Y^\ast)\}\|_2\nonumber\\
    \lessp&\frac{\sqrt{n}(m_{1n}+m_{2n})}{1+\lambda \omega_1(\bP)}\nonumber.
\end{align}
The fourth term can be bounded as
\begin{align}
    \sup_{\bx}|(\ast_4)|=&\sup_{\bx}|\sqrt{n}\widetilde{\bb}^\top[(\widehat{\bH}+\lambda\bP)^{-1}-(\bH+\lambda\bP)^{-1}]\mP_n\{\bb(\bX)(\widehat{Y}^\ast-Y^\ast)\}|\nonumber\\
    \leq&\sqrt{n}\sup_{\bx}\| \widetilde{\bb}(\bx)\|_2\|(\widehat{\bH}+\lambda\bP)^{-1}-(\bH+\lambda\bP)^{-1}\|_{\text{op}}\|\mP_n\{\bb(\bX)(\widehat{Y}^\ast-Y^\ast)\}\|_2\nonumber\\
    \lessp&\frac{\sqrt{\xi_K^2\log K/n}\sqrt{n}(m_{1n}+m_{2n})}{(1+2\lambda\omega_1(\bP))(1+\lambda\omega_1(\bP))},\nonumber
\end{align}
where the last inequality follows from Equation \eqref{eq:regularized-norm-differ-bound-2} and the proof of Theorem \ref{thm:least-squares-series-density-ratio}. For the fifth term, we have that 
\begin{align*}
    \sup_{\bx}|(\ast_5)|\leq&\lambda\sup_{\bx}\| \widetilde{\bb}(\bx) \|_2\| [(\widehat{\bH}+\lambda \bP)^{-1}-(\bH+\lambda \bP)^{-1}]\|_{\text{op}}\|\bP\|_{\text{op}}\| \bbeta_0\|_2\\
    \lessp&\lambda \omega_K(\bP)\| \bbeta_0\|_2\sqrt{\xi_K^2 \log K/n}/(1+2\lambda\omega_1(\bP))/(1+\lambda\omega_1(\bP)\\
    \lessp& \frac{\lambda\omega_K(\bP)\sqrt{\xi_K^2 \log K/n}}{(1+2\lambda\omega_1(\bP))(1+\lambda\omega_1(\bP))}.
\end{align*}
For the sixth term, it follows that
\begin{align*}
    (\ast_6)\leq&\frac{\lambda\sup_{\bx}\|\widetilde{\bb}(\bx)\|_2\|\bP\|_{\text{op}}\|\bbeta_0\|_2}{1+\lambda\omega_1(\bP)}\\
    \lessp &\frac{\lambda\omega_K(\bP)}{1+\lambda\omega_1(\bP)}.
\end{align*}
Finally, the term $\sup_{\bx}|\text{Rem}_{2n}(\widetilde{\bb}(\bx))|$ satisfies
\begin{align*}
    \sup_{\bx}|\text{Rem}_{2n}(\widetilde{\bb}(\bx))|=&\sup_{\bx}|\widetilde{\bb}(\bx)^\top(\bH+\lambda\bP)^{-1}\mG_n\{\bb(\bX)r_\tau\}|\\
    \leq&\sup_{\bx} \| \widetilde{\bb}(\bx)\|_2\| (\bH+\lambda\bP)^{-1}\|_{\text{op}}\| \mG_n\{\bb(\bX)r_\tau\}\|_2\\
    \lessp&1/(1+\lambda\omega_1(\bP))l_Kc_K\sqrt{\log K},
\end{align*}
where $\| \mG_n\{\bb(\bX)r_\tau\}\|_2=O_{\mP}(l_Kc_K\sqrt{\log K})$ follows from Step 3 of the proof of Lemma 4.2 in \cite{belloni2015some}.

\end{proof}

\subsection{Proof of Theorem \ref{thm:uniform-rate}}
\begin{proof}
By the Equation (4.23) in \cite{belloni2015some}, we obtain that 
\begin{align*}
    \sup_{\bx} |\widetilde{\bb}(\bx)^\top(\bH+\lambda\bP)^{-1}\mG_n\{\bb(\bX)(Y^\ast-\tau_{d_0d_1}(\bX))\}|\lessp \frac{\sqrt{\log K}}{1+\lambda\omega_1(\bP)}.
\end{align*}
and  
\begin{align*}
    &\sup_{\bx} |\bb(\bX)^\top(\widehat{\bbeta}-\bbeta_0)|\\
    =&\sup_{\bx} \|\bb(\bX)\|_2|\widetilde{\bb}(\bX)^\top(\widehat{\bbeta}-\bbeta_0)|\\
    \leq&\xi_K\sup_{\bx} |\widetilde{\bb}(\bX)^\top(\widehat{\bbeta}-\bbeta_0)|\\
    \leq&\xi_K/\sqrt{n}\sup_{\bx} |\sqrt{n}\widetilde{\bb}(\bX)^\top(\widehat{\bbeta}-\bbeta_0)|\\
    \leq&\xi_K/\sqrt{n}\left\{\sup_{\bx} |\widetilde{\bb}(\bx)^\top(\bH+\lambda\bP)^{-1}\mG_n\{\bb(\bX)(Y-\tau_{d_0d_1}(\bX))\}|+\sup_{\bx}\sum_{j=1}^2|\text{Rem}_{jn}(\widetilde{\bb}(\bx))|\right\}\\
\lessp &\frac{\xi_K}{(1+\lambda\omega_1(\bP))\sqrt{n}}\Biggl\{\left(\frac{\xi_K\sqrt{\log K}}{1+2\lambda\omega_1(\bP)}+\sqrt{n}\right)(\lambda\omega_K(\bP) /\sqrt{n}+m_{1n}+m_{2n})+\\
&\frac{m_{3n}}{1+2\lambda\omega_1(\bP)}+\sqrt{\log K}(1+l_Kc_K)\Biggr\}.
\end{align*}
Finally, we obtain the uniform rate of convergence for the regularized estimator as follows
Thus, it follows that
\begin{align*}
&\sup_{\bx} |\widehat{\tau}_{d_0d_1}-\tau_{d_0d_1}|\\
\leq&\sup_{\bx} |\bb(\bX)^\top(\widehat{\bbeta}-\bbeta_0
)+r_\tau|\\
\leq&\sup_{\bx} |\bb(\bX)^\top(\widehat{\bbeta}-\bbeta_0
)|+\sup_{\bx}|r_\tau|\\
\leq&\sup_{\bx} |\bb(\bX)^\top(\widehat{\bbeta}-\bbeta_0)|+l_Kc_K\\
\lessp &\frac{\xi_K}{(1+\lambda\omega_1(\bP))\sqrt{n}}\Biggl\{\left(\frac{\xi_K\sqrt{\log K}}{1+2\lambda\omega_1(\bP)}+\sqrt{n}\right)(\lambda\omega_K(\bP) /\sqrt{n}+m_{1n}+m_{2n})+\\
&\frac{m_{3n}}{1+2\lambda\omega_1(\bP)}+\sqrt{\log K}(1+l_Kc_K)\Biggr\}+l_Kc_K.
\end{align*}
where the third inequality follows from Assumption \ref{asp:regularityL2}(b). 
\end{proof}

\subsection{Proof of Theorem \ref{thm:pointwise-normal}}
\begin{proof}
We note that under regularity conditions outlined in Theorem \ref{thm:pointwise-normal}, Lemma \ref{lemma:pointwise-linearization} implies that
\begin{align*}
    &\text{Rem}_{1n}(\widetilde{\bb})\lessp\sqrt{\frac{\xi_K^2\log K}{n}}(1+\sqrt{K}l_Kc_K),
\end{align*}
which enables us to follow a similar proof strategy as in Theorem 4.2 of \cite{belloni2015some}.
\end{proof}

\subsection{Proof of Theorem \ref{thm:uniform-strong-gaussian}}
\begin{proof}
To enables us to follow a similar proof strategy as in Theorem 4.4 of \cite{belloni2015some}, we refine Lemma \ref{lemma:uniform-linearization}. We consider the following decomposition
\begin{align*}
     \sqrt{n}\widetilde{\bb}(\bx)^\top(\widehat{\bbeta}-\bbeta_0)=&\widetilde{\bb}(\bx)^\top(\bH+\lambda\bP)^{-1}\mG_n\{\bb(\bX)(Y^\ast-\bb(\bX)^\top\bbeta_0)\}+\text{Rem}_{1n}(\widetilde{\bb}(\bx))\\
     =&\widetilde{\bb}(\bx)^\top\bH^{-1}\mG_n\{\bb(\bX)(Y^\ast-\bb(\bX)^\top\bbeta_0)\}+\text{Rem}_{1n}(\widetilde{\bb}(\bx))+\text{Rem}_{0n}(\widetilde{\bb}(\bx)),
\end{align*}
where the new remainder term $\text{Rem}_{0n}(\widetilde{\bb}(\bx))$ is given by
\begin{align*}
    \text{Rem}_{0n}(\widetilde{\bb}(\bx))=&\widetilde{\bb}^\top[(\bH+\lambda \bP)^{-1}-\bH^{-1}]\left\{\mG_n\{\bb(\bX)(Y^\ast-\bb(\bX)^\top\bbeta_0)\}\right\}.
\end{align*}
We bound this new term as follows:
\begin{align*}
    \sup_{\bx}|\text{Rem}_{0n}(\widetilde{\bb}(\bx))|\leq&\sup_{\bx} |\widetilde{\bb}^\top(\bx)[(\bH+\lambda \bP)^{-1}-\bH^{-1}]\mG_n\{\bb(\bX)(Y^\ast-\bb(\bX)^\top\bbeta_0)\}|\\
    \leq&\sup_{\bx} \|\widetilde{\bb}(\bx)\|_2\|(\bH+\lambda \bP)^{-1}-\bH^{-1}\|_{\text{op}}\|\mG_n\{\bb(\bX)(Y^\ast-\bb(\bX)^\top\bbeta_0)\}\|_2\\
    \lessp&\frac{\lambda\omega_K(\bP)}{1+\lambda\omega_1(\bP)}\left(\sqrt{\frac{K}{n}}+m_{0n}\right),
\end{align*}
where the last inequality follows from the proof for Lemma \ref{lemma:pointwise-linearization}. Therefore, under regularity conditions outlined in Theorem \ref{thm:uniform-strong-gaussian}, Lemma \ref{lemma:uniform-linearization} together with the above result imply that
\begin{align*}
    &\text{Rem}_{1n}(\widetilde{\bb}(\bx))=O_{\mP}\left(m_{3n}\right),
\end{align*}

\end{proof}

\subsection{Proof of Theorem \ref{thm:cov-matrix-rate}}
\begin{proof}
We consider the following decomposition
\begin{align*}
  \widehat{\mathbb{V}}-\mathbb{V}=&\underbrace{\mP_n\left\{\bb(\bX)\bb(\bX)^\top\left[(\widehat{Y}^\ast-\bb(\bX)^\top\widehat{\bbeta})^2-(Y^\ast-\bb(\bX)^\top\bbeta_0)^2\right]\right\}}_{(\ast_1)}+\\
  &\underbrace{\mP_n\left\{\bb(\bX)\bb(\bX)^\top\right(Y^\ast-\bb(\bX)^\top\bbeta_0)^2\}-\mathbb{V}}_{(\ast_2)}.
\end{align*}
For the second term, the argument in the proof of Theorem 4.6 of \cite{belloni2015some} (see analysis for the second term of Equation (A.56)) implies that
\begin{align*}
   (\ast_2)\lessp (n^{1/\nu}+l_Kc_K)\sqrt{\frac{\xi_K^2\log K}{n}}. 
\end{align*}
To analyze the first term, we further decompose it as
\begin{align*}
    (\ast_1)=&2\mP_n\left\{\bb(\bX)\bb(\bX)^\top(Y^\ast-\bb(\bX)^\top\bbeta_0)\left[\bb(\bX)^\top(\bbeta_0-\widehat{\bbeta})+(\widehat{Y}^\ast-Y^\ast)\right]\right\}+\\
    &\mP_n\left\{\bb(\bX)\bb(\bX)^\top\left[\bb(\bX)^\top(\bbeta_0-\widehat{\bbeta})+(\widehat{Y}^\ast-Y^\ast)\right]^2\right\}\\
    \leq&2\underbrace{\mP_n\left\{\bb(\bX)\bb(\bX)^\top|(Y^\ast-\bb(\bX)^\top\bbeta_0)\bb(\bX)^\top(\bbeta_0-\widehat{\bbeta})|\right\}}_{(\ast_3)}+\\
    &2\underbrace{\mP_n\left\{\bb(\bX)\bb(\bX)^\top|(Y^\ast-\bb(\bX)^\top\bbeta_0)(\widehat{Y}^\ast-Y^\ast)|\right\}}_{(\ast_4)}+\\
    &2\underbrace{\mP_n\left\{\bb(\bX)\bb(\bX)^\top\left[\bb(\bX)^\top(\bbeta_0-\widehat{\bbeta})\right]^2\right\}}_{(\ast_5)}+2\underbrace{\mP_n\left\{\bb(\bX)\bb(\bX)^\top(\widehat{Y}^\ast-Y^\ast)^2\right\}}_{(\ast_6)}.
\end{align*}
By Chebyshev's Inequality and assumption that the support of distribution of $|\widehat{Y}^\ast-Y^\ast|$ conditional on the training sample is finite, it follows that, for all $s\in[S]$,
\begin{align*}
    \max_{i\in\mathcal{F}_s}|\widehat{Y}^\ast-Y^\ast|=&E\{\max_{i\in\mathcal{F}_s}|\widehat{Y}^\ast-Y^\ast||\widehat{\bGamma}_s\}+o_{\mP^\ast}(1)\\
    =&E\{\max_{i\in\mathcal{F}_s}|\widehat{Y}^\ast-Y^\ast||\widehat{\bGamma}_s\}+o_{\mP}(1),
\end{align*}
where the probability measure $\mP^\ast$ is conditional on the training sample $\mathcal{F}^c_s$. Thus, we obtain 
\begin{align}\label{eq:bounded-abs-1-n-hatY-Y}
    \max_{1\leq i\leq n}|\widehat{Y}^\ast-Y^\ast|\leq\sum_{s}\max_{i\in\mathcal{F}_s}|\widehat{Y}^\ast-Y^\ast|=S\times E\{\max_{i\in\mathcal{F}_s}|\widehat{Y}^\ast-Y^\ast||\widehat{\bGamma}_s\}+o_{\mP}(1).
\end{align}
Similarly, it follows that 
\begin{align}\label{eq:bounded-square-1-n-hatY-Y}
    \max_{1\leq i\leq n}(\widehat{Y}^\ast-Y^\ast)^2=S\times E\{\max_{i\in\mathcal{F}_s}(\widehat{Y}^\ast-Y^\ast)^2|\widehat{\bGamma}_s\}+o_{\mP}(1).
\end{align}
For the term $(\ast_4)$, it follows that 
\begin{align*}
    \| (\ast)_4\|_{\text{op}}\leq &\max_{1\leq i\leq n}|\widehat{Y}^\ast-Y^\ast|\max_{1\leq i\leq n}\left(|Y^\ast-\bb(\bX)\bbeta_0|\right)\| \mP_n\{\bb(\bX)\bb(\bX)^\top\}\|_{\text{op}}\\
    \leq&\left(SE\left\{\max_{1\leq i\leq n}|\widehat{Y}^\ast-Y^\ast||\widehat{\bGamma}_s\right\}+o_{\mP}(1)\right)\times\\
    &\max_{1\leq i\leq n}\left(|Y^\ast-\tau_{d_0d_1}(\bX)|+|r_\tau(\bX)|\right)\| \mP_n\{\bb(\bX)\bb(\bX)^\top\}\|_{\text{op}}\\
    \lessp&E\left\{\max_{1\leq i\leq n}|\widehat{Y}^\ast-Y^\ast||\widehat{\bGamma}_s\right\}(n^{1/\nu}+l_Kc_K),
\end{align*}
where the second inequality follows from Equation \eqref{eq:bounded-abs-1-n-hatY-Y}, the third inequality holds because $\| \widehat{\bH}\|_{\text{op}}\lessp 1$, $\max_{1\leq i\leq n}|r_\tau(\bX)|\leq l_Kc_K$ by Assumption \ref{asp:regularityL2}(b), and $\max_{1\leq i\leq n}|Y^\ast-\tau_{d_0d_1}(\bX)|=O_{\mP}(n^{1/\nu})$ by the Step 1 of proof of Lemma 4.2 in \cite{belloni2015some}. For the term $(\ast_3)$, we obtain that
\begin{align*}
    \| (\ast)_3\|_{\text{op}}\leq &\| \widehat{\bH}\|_{\text{op}}\sup_{\bX}|Y^\ast-\bb(\bX)^\top\bbeta_0|\sup_{\bX}|\bb(\bX)^\top(\bbeta_0-\widehat{\bbeta})|\\
    \lessp&n^{1/\nu}\frac{\xi_K}{(1+\lambda\omega_1(\bP))\sqrt{n}}\Biggl\{\left(\frac{\xi_K\sqrt{\log K}}{1+2\lambda\omega_1(\bP)}+\sqrt{n}\right)(\lambda\omega_K(\bP)/\sqrt{n}+m_{1n}+m_{2n})+\\
&\frac{m_{3n}}{1+2\lambda\omega_1(\bP)}+\sqrt{\log K}(1+l_Kc_K)\Biggr\},
\end{align*}
where the last inequality follows from Lemma \ref{lemma:uniform-linearization} and that $\| \widehat{\bH}\|_{\text{op}}\lessp 1$. For the term $(\ast_5)$, we have that
\begin{align*}
    \| (\ast)_5\|_{\text{op}}\lessp &\frac{\xi_K^2}{(1+\lambda\omega_1(\bP))^2n}\Biggl\{\left(\frac{\xi_K\sqrt{\log K}}{1+2\lambda\omega_1(\bP)}+\sqrt{n}\right)(\lambda\omega_K(\bP)/\sqrt{n}+m_{1n}+m_{2n})+\\
&\frac{m_{3n}}{1+2\lambda\omega_1(\bP)}+\sqrt{\log K}(1+l_Kc_K)\Biggr\}^2,
\end{align*}
where the inequality follows from Lemma \ref{lemma:uniform-linearization} and that $\| \widehat{\bH}\|_{\text{op}}\lessp 1$. For the term $(\ast_6)$, we have that 
\begin{align*}
    \| (\ast)_6\|_{\text{op}}\leq &\max_{1\leq i\leq n}(\widehat{Y}^\ast_i-Y_i^\ast)^2\| \mP_n\{\bb(\bX)\bb(\bX)^\top\}\|_{\text{op}}\\
    \lessp&\left(E\left\{\max_{1\leq i\leq n}(\widehat{Y}_i^\ast-Y^\ast_i)^2|\widehat{\bGamma}_s\right\}+o_{\mP}(1)\right) 
    \lessp E\left\{\max_{1\leq i\leq n}(\widehat{Y}_i^\ast-Y^\ast_i)^2|\widehat{\bGamma}_s\right\},
\end{align*}
where the second inequality follows from Equation \eqref{eq:bounded-square-1-n-hatY-Y}. Finally, under the conditions as those outlined in Theorem \ref{thm:cov-matrix-rate}, the bounds simplify to
\begin{align*}
    \| \widehat{\mathbb{V}}-\mathbb{V}\|_{\text{op}}\lessp\left(E\left\{\max_{1\leq i\leq n}|\widehat{Y}^\ast-Y^\ast||\widehat{\bGamma}_s\right\}+\sqrt{\frac{\xi_K^2\log K}{n}}\right)(n^{1/\nu}+l_Kc_K)+E\left\{\max_{1\leq i\leq n}(\widehat{Y}_i^\ast-Y^\ast_i)^2|\widehat{\bGamma}_s\right\}.
\end{align*}
\end{proof}

\subsection{Proof of Lemma \ref{lemma:uniform-linearization-boot}}
\begin{proof}
   By the proof of Theorem 4.5 in \cite{belloni2015some}, (i) the weights satisfy the following properties
\begin{align}
    \max_{1\leq i\leq n}\mathfrak{w}_i\lessp \log n,\nonumber\\
    E\{\mathfrak{w}\}=E\{\mathfrak{w}^2\}=1,\nonumber\\
    E\{\mathfrak{w}^{\nu/2}\}\lesssim1,\label{eq:property-exp-weights}
\end{align}
and (ii) that
\begin{align}
    \text{log}K\lesssim \log n.\label{eq:property-logK<=logn}
\end{align}

The bootstrap estimator $\widehat{\bbeta}^b$ is given by $\widehat{\bbeta}^b=(\widehat{\bH}^b_\tau+\lambda\bP)^{-1}\widehat{\bh}$, where
\begin{align*}
    &\widehat{\bH}^b_\tau=\mP_n\{\sqrt{\mathfrak{w}}_i\bb(\bX_i)(\sqrt{\mathfrak{w}}_i\bb(\bX_i))^\top\}\\
    &\widehat{\bh}=\mP_n\{\sqrt{\mathfrak{w}}\bb(\bX)\sqrt{\mathfrak{w}}\widehat{Y}^\ast\}.
\end{align*}
In addition, the projection $\bbeta_0$ becomes
\begin{align*}
    \bbeta_0=&{\bH^b}^{-1}\bh^b\\
    =&\bH^{-1}E\{\mathfrak{w}\bb(\bX)\tau_{d_0d_1}(\bX)\}\\
    =&\bH^{-1}E\{\bb(\bX)\tau_{d_0d_1}(\bX)\}=\bH^{-1}\bh,
\end{align*}
where the second and third equalities follow from $E\{\mathfrak{W}\}=1$. Now, we consider the decomposition
\begin{align*}
    &\sqrt{n}\widetilde{\bb}^\top(\widehat{\bbeta}^b-\widehat{\bbeta})=\sqrt{n}\widetilde{\bb}^\top(\widehat{\bbeta}^b-\bbeta_0)-\sqrt{n}\widetilde{\bb}^\top(\widehat{\bbeta}-\bbeta_0),
\end{align*}
where the second term $\sqrt{n}\widetilde{\bb}^\top(\widehat{\bbeta}-\bbeta_0)$ has already been analyzed in Lemma \ref{lemma:uniform-linearization}. It suffices to analyze the first term $\sqrt{n}\widetilde{\bb}^\top(\widehat{\bbeta}^b-\bbeta_0)$. By the proof of Lemma \ref{lemma:uniform-linearization}, we obtain 
\begin{align*}
    &\sqrt{n}\widetilde{\bb}^\top(\widehat{\bbeta}^b-\bbeta_0)\\
    =&\widetilde{\bb}^\top({\bH+\lambda\bP})^{-1}\mG_n\{\mathfrak{w}\bb(\bX)(Y^\ast-\bb(\bX)^\top\bbeta_0)\}+\\
    &\underbrace{\widetilde{\bb}^\top[(\widehat{\bH}^b+\lambda \bP)^{-1}-(\bH^b+\lambda \bP)^{-1}]\mG_n\{\mathfrak{w}\bb(\bX)(Y^\ast-\tau_{d_0d_1}(\bX))\}}_{(\ast_1)}+\\
    &\underbrace{\widetilde{\bb}^\top[(\widehat{\bH}^b+\lambda \bP)^{-1}-(\bH^b+\lambda \bP)^{-1}]\mG_n\{\mathfrak{w}\bb(\bX)r_\tau(\bX)\}}_{(\ast_2)}+\\    
    &\underbrace{\sqrt{n}\widetilde{\bb}^\top(\bH^b+\lambda\bP)^{-1}\mP_n\{\mathfrak{w}\bb(\bX)(\widehat{Y}^\ast-Y^\ast)\}}_{(\ast_3)}+\\
    &\underbrace{\sqrt{n}\widetilde{\bb}^\top[(\widehat{\bH}^b+\lambda \bP)^{-1}-(\bH^b+\lambda \bP)^{-1}]\mP_n\{\mathfrak{w}\bb(\bX)(\widehat{Y}^\ast-Y^\ast)\}}_{(\ast_4)}-\\
    &\underbrace{\lambda\widetilde{\bb}^\top[(\widehat{\bH}^b+\lambda \bP)^{-1}-(\bH^b+\lambda \bP)^{-1}]\bP\bbeta_0}_{(\ast_5)}-\underbrace{\lambda\widetilde{\bb}^\top(\bH^b+\lambda \bP)^{-1}\bP\bbeta_0}_{(\ast_6)}.
\end{align*}
We analyze these terms using techniques similar to those employed in the proof of Lemma~\ref{lemma:uniform-linearization}. The first term obeys
\begin{align*}
    \sup_{\bx}|(\ast_1)|\lessp&\sqrt{\log n}n^{1/\nu}\| (\widehat{\bH}^b+\lambda \bP)^{-1}-(\bH^b+\lambda \bP)^{-1}\|_{\text{op}}\| \widehat{\bH}^b\|_{\text{op}}^{1/2}\sqrt{\log K}\\
    \lessp&\sqrt{\log n}n^{1/\nu}\| (\widehat{\bH}^b_\tau+\lambda \bP)^{-1}-(\bH^b+\lambda \bP)^{-1}\|_{\text{op}}\sqrt{\log K}\\
    \lessp&\frac{n^{1/\nu}\sqrt{\log n}\sqrt{{\xi_K^2\log K}/{n}}\sqrt{\log K}}{(1+2\lambda\omega_1(\bP))(1+\lambda\omega_1(\bP))}\\
    \lessp&\frac{n^{1/\nu}\sqrt{\xi_K^2(\log n)^3/n}}{(1+2\lambda\omega_1(\bP))(1+\lambda\omega_1(\bP))},
\end{align*}
where the last inequality follows from Equation \eqref{eq:property-logK<=logn}. The second term obeys 

\begin{align*}
    \sup_{\bx}|(\ast_2)|\lessp&\| (\widehat{\bH}^b_\tau+\lambda \bP)^{-1}-(\bH^b+\lambda \bP)^{-1}\|_{\text{op}}\| \mG_n\{\mathfrak{w}\bb(\bX)r_\tau\}\|_2\\
    \lessp& \sqrt{\frac{\xi_K^2(\log n)^3}{n}}\frac{l_Kc_K\sqrt{K}}{(1+2\lambda\omega_1(\bP))(1+\lambda\omega_1(\bP))}.
\end{align*}
The third term obeys
\begin{align}
    \sup_{\bx}|(\ast_3)|=&\sup_{\bx}|\sqrt{n}\widetilde{\bb}(\bx)^\top(\bH^b+\lambda\bP)^{-1}\mP_n\{\mathfrak{w}\bb(\bX)(\widehat{Y}^\ast-Y^\ast)\}|\nonumber\\
    \leq&\sqrt{n}\sup_{\bx}\| \widetilde{\bb}(\bx)\|_2\|(\bH^b+\lambda\bP)^{-1}\|_{\text{op}}\|\mP_n\{\mathfrak{w}\bb(\bX\bb(\bX)(\widehat{Y}^\ast-Y^\ast)\}\|_2\nonumber\\
    \leq&\frac{\sqrt{n}\sup_{\bx}\| \widetilde{\bb}(\bx)\|_2}{1+\lambda \omega_1(\bP)}\|\mP_n\{\mathfrak{w}\bb(\bX)(\widehat{Y}^\ast-Y^\ast)\}\|_2\nonumber\\
    \lessp&\frac{\sqrt{n}\log n(m_{1n}+m_{2n})}{1+\lambda \omega_1(\bP)}\nonumber.
\end{align}
The fourth term obeys
\begin{align}
    \sup_{\bx}|(\ast_4)|=&\sup_{\bx}|\sqrt{n}\widetilde{\bb}^\top[(\widehat{\bH}^b_\tau+\lambda\bP)^{-1}-(\bH^b+\lambda\bP)^{-1}]\mP_n\{\mathfrak{w}\bb(\bX)(\widehat{Y}^\ast-Y^\ast)\}|\nonumber\\
    \leq&\sqrt{n}\sup_{\bx}\| \widetilde{\bb}(\bx)\|_2\|(\widehat{\bH}^b+\lambda\bP)^{-1}-(\bH^b+\lambda\bP)^{-1}\|_{\text{op}}\|\mP_n\{\mathfrak{w}\bb(\bX)(\widehat{Y}^\ast-Y^\ast)\}\|_2\nonumber\\
    \lessp&\frac{\sqrt{\xi_K^2(\log n)^3/n}\sqrt{n}(m_{1n}+m_{2n})}{(1+2\lambda\omega_1(\bP))(1+\lambda\omega_1(\bP))},\nonumber
\end{align}
The bounds for the fifth and the sixth terms remain identical with 
\begin{align*}
  &(\ast_5)\lessp \frac{\lambda\omega_K(\bP)\sqrt{\xi_K^2 \log n/n}}{(1+2\lambda\omega_1(\bP))(1+\lambda\omega_1(\bP))},\\
    &(\ast_6)\lessp \frac{\lambda\omega_K(\bP)}{1+\lambda\omega_1(\bP)}.
\end{align*}
\end{proof}

\subsection{Proof of Theorem \ref{thm:uniform-strong-gaussian-boot}}
\begin{proof}
From the proof of Lemma~\ref{lemma:uniform-linearization-boot}, the regularization bias is not affected by the additional bootstrap weights, which implies that the reasoning in the proof of Theorem~\ref{thm:uniform-strong-gaussian} 
applies. 
\end{proof}

\subsection{Proof of Theorem \ref{thm:strong-gaussian-appro-suprema}}
\begin{proof}
The proof follows the similar steps as those in the proofs of Theorems 5.4 and 5.5 in \cite{belloni2015some}.
\end{proof}

\subsection{Proof of Theorem \ref{thm:honest-cov-gauss-boot-ci}}
\begin{proof}
The proof follows from Theorems~\ref{thm:cov-matrix-rate} and \ref{thm:strong-gaussian-appro-suprema}, and the similar steps as those in the proof of Theorem 5.6 of \cite{belloni2015some}, by replacing the corresponding preliminary results in Lemma 5.1 and Theorem 5.5 with those in Theorems~\ref{thm:cov-matrix-rate} and \ref{thm:strong-gaussian-appro-suprema}, respectively.
\end{proof}

\section{Supporting information under monotonicity or counterfactual intermediate independence}\label{sec:mono-supp-simulation}
For brevity, results assuming monotonicity or counterfactual intermediate independence are presented without proof, with the exception of the derivations for the critical rate modifications of $m_{2n}$. The omitted derivations are analogous to those for the general odds ratio parameterization. We begin by formally stating the assumptions of monotonicity and counterfactual intermediate independence below.
\begin{assumption}[\emph{Monotonicity}]\label{assumption:mono-supp}
 $D(1)\geq D(0)$.   
\end{assumption}

\begin{assumption}[\emph{Counterfactual intermediate independence}]\label{assumption:CII-supp}
 $D(0) \perp D(1)|\bC$.   
\end{assumption}

Under the monotonicity assumption, the identification formulas for the principal scores simplify to the following expressions \citep{DingandLu2016}:
\begin{align*}
e_{d_0d_1}(\bC)=(2d_0-d_1)p_0(\bC)+(2d_1-d_0-1)p_1(\bC)+1-d_1.    
\end{align*}

Under counterfactual intermediate independence, the identification formulas for the principal scores simplify to the following expressions \citep{tong2025semiparametric}:
\begin{align*}
e_{d_0d_1}(\bC)=\prod_{z=0,1} p_{z}(\bC)^{d_{z}}\{1-p_{z}(\bC)\}^{1-d_{z}}.    
\end{align*}

\subsection{EIFs and pseudo-outcomes}
The EIFs and the implied pseudo-outcomes under monotonicity or counterfactual intermediate independence can be derived by substituting the identification formulas for the principal scores and the corresponding pathwise derivatives, $\phi_{d_0d_1}^e$, as detailed below. Following \cite{JiangJRSSB2022,tong2025semiparametric}, the expression for $\phi_{d_0d_1}^e$ is given by:
\begin{equation*}
   \phi_{d_0d_1}^e(D,Z,\bC)=\begin{cases}
       (2d_0-d_1)\{\psi_{D,0}-p_0(\bC)\}+(2d_1-d_0-1)\{\psi_{D,1}-p_1(\bC)\},&(\text{Assumption \ref{assumption:mono-supp}})\\
       \sum_{z=0,1}(-1)^{d_z+1}\{\psi_{D,z}-p_z(\bC)\}p_{1-z}(\bC)^{d_{1-z}}\{1-p_{1-z}(\bC)\}^{1-d_{1-z}},&(\text{Assumption \ref{assumption:CII-supp}})
   \end{cases}. 
\end{equation*}

\subsection{Rates summarizing nuisance estimation errors}

Given the hierarchical structure of the asymptotic results in Section \ref{sec:large-sample-theory-inference-bands}, to ensure the validity of the large-sample theory under Assumption \ref{assumption:mono-supp} or Assumption \ref{assumption:CII-supp}, it suffices to remove the constraints on the odds ratio parameter in the main manuscript Assumption \ref{asp:regularityL2}(d) and refine the form of $m_{2n}$, which summarizes the impact of nuisance estimation error, in the main manuscript Theorem \ref{thm:least-squares-series-density-ratio}. With these modifications, the asymptotic theory presented in the main manuscript remains valid under the Assumption \ref{assumption:mono-supp} or Assumption \ref{assumption:CII-supp}

We adopt the proof provided in Section \ref{sec:proofs-all;ss:proof-thm1}. To derive the form of $m_{2n}$, we can decompose the expected bias as follows:
\begin{align*}
&\mP\{\bb(\bX)(\widehat{Y}^\ast-Y^\ast)|\widehat{\bGamma}_s\}\\
=&\underbrace{E\left\{\bb(\bX)\frac{E\{\widehat{\varphi}^N_{d_0d_1}-\varphi^N_{d_0d_1}|\bX,\mathcal{F}_s^c\}-\widehat{\tau}_{d_0d_1}E\{\widehat{\varphi}^D_{d_0d_1}-\varphi^D_{d_0d_1}|\bX,\mathcal{F}_s^c\}}{\widehat{\tau}_{d_0d_1}^D}|\mathcal{F}_s^c\right\}}_{(\ast_1)}+\\
&\underbrace{E\left\{\bb(\bX)\frac{(\widehat{\tau}^N_{d_0d_1}-\tau^N_{d_0d_1})(\widehat{\tau}_{d_0d_1}^D-\tau_{d_0d_1}^D)}{(\widehat{\tau}_{d_0d_1}^D)^2}|\mathcal{F}_s^c\right\}-E\left\{\bb(\bX)\frac{\tau^N_{d_0d_1}(\widehat{\tau}_{d_0d_1}^D-\tau_{d_0d_1}^D)^2}{(\widehat{\tau}_{d_0d_1}^D)^2\tau_{d_0d_1}^D}|\mathcal{F}_s^c\right\}}_{(\ast_2)}.
\end{align*}
For the term $(\ast_1)$, it follows that, under monotonicity,
\begin{align*}
      &\| (\ast_1)\|_2^2\\
    \lessp&\sum_{k=1}^K[(E\{(\widehat{\varphi}^N_{d_0d_1}-\varphi^N_{d_0d_1})|\mathcal{F}_s^c\})^2+(E\{(\widehat{\varphi}^D_{d_0d_1}-\varphi^D_{d_0d_1})|\mathcal{F}_s^c\})^2]\\
    \lessp&K[\sum_{z=0,1}\eta(m_{zd_z})^2(\eta(\pi)^2+\sum_{j=0,1}\eta(p_j)^2)+\eta(\pi)^2\sum_{j=0,1}\eta(p_j)^2],
\end{align*}
where 
the final inequality follows from the Cauchy–Schwarz inequality and the proof provided in Section S7.2 of the Supplementary Material to \cite{JiangJRSSB2022}. For the term $(\ast_2)$, we have that, under monotonicity,
\begin{align*}
    &\| (\ast_2)\|_2^2
    \lessp K\left(\eta(\tau_{d_0d_1}^N)^2\eta(\tau_{d_0d_1}^D)^2+\eta(\tau_{d_0d_1}^D)^4\right).
\end{align*}
Thus, it follows that 
\begin{align*}
   &\mP\{\bb(\bX)(\widehat{Y}^\ast-Y^\ast)|\widehat{\bGamma}_s\}\\
   \lessp&\underbrace{\sqrt{K} \left[\sum_{z=0,1}\eta(m_{zd_z})\left(\eta(\pi)+\sum_{j=0,1}\eta(p_j)\right)+\eta(\pi)\sum_{j=0,1}\eta(p_j)\right]}_{m_{2n}^{(1)}}+\\
   &\underbrace{\sqrt{K}\left(\eta(\tau_{d_0d_1}^N)\eta(\tau_{d_0d_1}^D)+\eta(\tau_{d_0d_1}^D)^2\right)}_{m_{2n}^{(2)}}\\
   =&m_{2n}.
\end{align*}
That is, under monotonicity, the rate $m_{2n}$ is improved by eliminating the dependence on the product of the estimation errors for $p_0(\bC)$ and $p_1(\bC)$ in the $\mathcal{L}^2(\mP)$ norm. Similarly, following Section S8.3 of the Supplementary Material in \cite{tong2025semiparametric}, we obtain that: under Assumption \ref{assumption:CII-supp},
\begin{align*}
   &\mP\{\bb(\bX)(\widehat{Y}^\ast-Y^\ast)|\widehat{\bGamma}_s\}\\
   \lessp&\underbrace{\sqrt{K} \left[\sum_{z=0,1}\eta(m_{zd_z})\left\{\eta(\pi)+\eta(p_{1-z})\right\}+\eta(p_0)\eta(p_1)+\eta(\pi)\sum_{j=0,1}\eta(p_j)\right]}_{m_{2n}^{(1)}}+\\
   &\underbrace{\sqrt{K}\left(\eta(\tau_{d_0d_1}^N)\eta(\tau_{d_0d_1}^D)+\eta(\tau_{d_0d_1}^D)^2\right)}_{m_{2n}^{(2)}}\\
   =&m_{2n}.
\end{align*}
That is, under Assumption \ref{assumption:CII-supp}, the rate $m_{2n}$ is improved by eliminating the dependence on $\sum_{z=0,1}\eta(m_{zd_z})\eta(p_z)$.

\section{Supporting information for the simulation experiments}\label{sec:supp-simulation}


\subsection{Design}\label{sec:sim;ss:design}
Throughout the simulation experiments, all uniform and pointwise performance metrics are computed over $1000$ replications.
\subsubsection{Non-monotonicity data-generating process}
We generate a sample of $n = 3000$ individuals with two continuous covariates, $X\sim\text{Uniform}[-1,1]$ and $\widetilde{X}_1\sim\mathcal{N}(0,1)$, and one binary covariate, $\widetilde{X}_2 \sim \text{Bern}(0.5)$. The treatment assignment is then simulated using logistic regression, with $\Pr(Z = 1|\bX) = \text{expit}(0.4X-0.3\widetilde{X}_1+0.4\widetilde{X}_2)$, where $\text{expit}(x)=(1+e^{-x})^{-1}$. The potential intermediate outcomes are simulated from 
\begin{align*}
    &\Pr(D(0) = 1|\bC) = \text{expit}(0.1X^2+\sin(X)/3+0.4X-0.2\widetilde{X}_1+0.3\widetilde{X}_2),\\
    &\Pr(D(1) = 1|\bC) = \text{expit}(0.2X^2-0.3X+0.3\widetilde{X}_1+0.2\widetilde{X}_2).
\end{align*}
We assume that the odds ratio function satisfies $\theta(\bX) = |X-\widetilde{X}_1/2-\widetilde{X}_2/3| + 0.3$. The joint distribution of $D(0)$ and $D(1)$ can be computed following \cite{tong2025semiparametric}. Finally, the potential final outcomes follow 
\begin{align*}
    &Y(1)|G,\bC\sim \mathcal{N}(-1+D(1)+0.5X^2-0.3X-0.3\widetilde{X}_1+0.4\widetilde{X}_2, 1),\\
    &Y(0)|G,\bC\sim \mathcal{N}(3-D(0)+\cos(X)/3-0.7X^2-0.2X+0.7\widetilde{X}_1-0.3\widetilde{X}_2, 1).
\end{align*}
It is straightforward to verify that all structural assumptions are satisfied in the above data-generating process.

\subsubsection{Monotonicity data-generating process}
We modify the data-generating process for the potential intermediate outcomes and potential final outcomes, while all other aspects remain identical to the data-generating process under non-monotonicity. Specifically, the potential intermediate outcomes are simulated from the multinomial logistic model as follows:
\begin{align*}
    &\Pr(G=11|\bC)=\left\{1+\exp(-0.3X+0.4\widetilde{X}_1+0.1\widetilde{X}_2)+\exp(0.3X-0.4\widetilde{X}_1+0.1\widetilde{X}_2)\right\}^{-1},\\
    &\Pr(G=01|\bC)=\Pr(G=11|\bC)\exp(-0.3X+0.4\widetilde{X}_1+0.1\widetilde{X}_2),\\
    &\Pr(G=00|\bC)=\Pr(G=11|\bC)\exp(0.3X-0.4\widetilde{X}_1+0.1\widetilde{X}_2).
\end{align*}
The potential final outcomes follow 
\begin{align*}
    &Y(1)|G,\bC\sim \mathcal{N}(-1+D(1)-0.3X-0.3\widetilde{X}_1+0.4\widetilde{X}_2, 1),\\
    &Y(0)|G,\bC\sim \mathcal{N}(3-D(0)-0.2X+0.7\widetilde{X}_1-0.3\widetilde{X}_2, 1).
\end{align*}
To assess the impact of odds ratio misspecification, we also compute estimates fitted under non-monotonicity with $\theta(\mathbf{C})=5$.

\subsection{Results}
In the main manuscript, we established uniform and pointwise inference results for the complier stratum in the presence of non-monotonicity. Tables \ref{tab:pointwise-11-compare}--\ref{tab:pointwise-10-compare} summarize the pointwise inference performance for the remaining strata, while Figures \ref{fig:sim-results-11}--\ref{fig:sim-results-10} illustrate the corresponding uniform inference results. These findings exhibit qualitatively similar patterns to those observed for the complier stratum. To expand, we now discuss 
the performance under monotonicity, the consequences of sensitivity parameter misspecification, and the empirical comparison with existing instrumental variable (IV) method. These results provide practical guidance for implementing our methods.


\subsubsection{Performance under monotonicity}
Figures \ref{fig:sim-mono-fitmono-11}--\ref{fig:sim-mono-fitmono-00} summarize the uniform inference results under monotonicity, while Tables \ref{tab:pointwise-11-mono-fitmono}--\ref{tab:pointwise-00-mono-fitmono} summarize the pointwise inference results. Generally, the observed patterns and practical recommendations are qualitatively similar to those under non-monotonicity, indicating that the proposed methods perform well when monotonicity holds.

\subsubsection{Sensitivity parameter misspecifications}
We first discuss the scenario where non-monotonicity holds, but the estimators are fitted assuming monotonicity. In this case, Figures \ref{fig:sim-results-nonmono-fitmono-11}--\ref{fig:sim-results-nonmono-fitmono-00} summarize the uniform inference results, while Tables \ref{tab:pointwise-11-nonmono-fitmono}--\ref{tab:pointwise-00-nonmono-fitmono} summarize the pointwise inference results. Consistent with the simulation findings in \cite{tong2025semiparametric}, the misspecification of non-monotonicity as monotonicity can lead to both biased uniform and pointwise inference for CPCEs across all strata. Moreover, the estimates for CPCEs among compliers become extremely unstable and invalid, as evidenced by extremely large mean integrated squared errors (MISEs), uniform band width, bias, and pointwise AESEs. This occurs because the principal score derived under monotonicity, $\whp_1(\mathbf{C})-\whp_0(\mathbf{C})$, can become extremely small or even negative when the data exhibit non-monotonicity.

Second, we discuss the scenario where monotonicity holds, but the estimators are fitted assuming non-monotonicity. In this case, Figures \ref{fig:sim-mono-fit5-11}--\ref{fig:sim-mono-fit5-00} summarize the uniform inference results, while Tables \ref{tab:pointwise-11-mono-fit5}--\ref{tab:pointwise-00-mono-fit5} summarize the pointwise inference results. Consistent with \cite{tong2025semiparametric}, when monotonicity is misspecified as a large constant odds ratio, both pointwise and uniform inference remain approximately valid, and the impact of such misspecification is not significant. This follows from the fact that monotonicity is equivalent to setting $\theta\to\infty$, provided that $p_1(\mathbf{C})>p_0(\mathbf{C})$.

\subsubsection{Comparison with existing IV method}
We fit the ratio estimator proposed by \cite{takatsu2025doubly} under the IV framework. Under a data-generating process satisfying non-monotonicity, Figure \ref{fig:sim-results-JRSSA-uniform} summarizes the root-MISE for the CLATE, while Figure \ref{fig:sim-results-JRSSA-pointwise} presents the pointwise error box plots at five representative points. Under a data-generating process satisfying monotonicity, Figure \ref{fig:sim-results-JRSSA-uniform-mono} summarizes the root-MISE for the CLATE, while Figure \ref{fig:sim-results-JRSSA-pointwise-mono} presents the pointwise inference performance at four representative points. In general, violations of structural assumptions render inference derived from the naive application of the IV estimator biased; notably, the resulting estimates become unstable when the monotonicity assumption is violated.

Since the method in \cite{takatsu2025doubly} is designed for the IV framework under the assumptions of monotonicity and the exclusion restriction (ER), it may exhibit bias and instability (especially when monotonicity fails) when these assumptions are violated. Therefore, it is recommended to assess the plausibility of structural causal assumptions when considering the application of the method proposed in this manuscript or those designed for the IV framework, such as \cite{takatsu2025doubly}.

\clearpage

\section{Additional Tables and Figures}\label{sec:additional-table-figures}

\begin{table}[ht!]
\centering
\caption{Comparison with the concurrent work of \cite{zhang2026estimation}.}
\label{supp;tab:comparison-zhang}
\renewcommand{\arraystretch}{1.0}
\setlength{\tabcolsep}{3pt}
\begin{tabularx}{\linewidth}{
>{\raggedright\arraybackslash}p{0.20\linewidth}
>{\raggedright\arraybackslash}X
>{\raggedright\arraybackslash}X}
\toprule
\textbf{Aspect} & \textbf{This paper} & \textbf{\cite{zhang2026estimation}} \\
\midrule
Principal-stratum setting
&
Allows all principal strata and does not impose monotonicity as a identifying restriction.
&
Works under monotonicity and studies always-takers, compliers, and never-takers.
\\

Target estimand
&
Targets interpretable CPCEs indexed by a scientifically chosen subset of baseline covariates.
&
Targets conditional PCEs given the full set of exogenous baseline covariates.
\\

Identification
&
Uses principal ignorability with a conditional odds-ratio sensitivity parameter.
&
Uses principal ignorability under monotonicity.
\\

Sensitivity analysis
&
Allows sensitivity analysis through a covariate-dependent odds-ratio function.
&
Does not primarily develop a sensitivity-analysis framework.
\\

Statistical challenges
&
Requires handling intermediate nuisance functions for which existing machine learners are not directly available.
&
Allows the nuisance functions to be learned using existing methods.
\\

Estimation
&
Uses efficient-influence-function-based orthogonal pseudo-outcomes and a doubly cross-fit procedure to address the resulting nuisance-estimation challenges.
&
Proposes several flexible causal machine learners.
\\

Inference
&
Provides computationally efficient pointwise inference and simultaneous inference.
&
Provides pointwise inference based on computationally intensive bootstrap procedures, but does not provide simultaneous inference.
\\

Theory
&
Derives $\mathcal{L}^2$ and uniform convergence rates, and establishes oracle properties and uniform inference under penalized linear sieves.
&
Establishes multiply robust large-sample theory based on high-level learner-agnostic rate conditions.
\\
\bottomrule
\end{tabularx}
\end{table}

\begin{table}[htbp]
\centering
\caption{ Simulation results presenting pointwise inference metrics, including bias, Monte Carlo standard deviation (MCSD), average estimated standard errors (AESE), and empirical coverage probabilities (CP), for the proposed PLSS-DCDR learner for the always-takers stratum ($11$) at four representative points $X=x_0\in\{-0.50,-0.25,0.25,0.50\}$ when non-monotonicity holds. The basis functions are constructed using P-splines with $K \in \{5, 10, 20\}$ knots. We employ three implementation strategies: a penalized estimator with the smoothing parameter selected via the generalized cross-validation (GCV) criterion \cite{eilers1996flexible}, a variant based on an undersmoothed GCV selection, and an unpenalized estimator corresponding to the least squares series approach.}
\label{tab:pointwise-11-compare}
\resizebox{\textwidth}{!}{%
\begin{tabular}{lrr cccc cccc cccc}
\toprule
 & & \multicolumn{4}{c}{GCV} & \multicolumn{4}{c}{GCV.under} & \multicolumn{4}{c}{Unpenalized} \\
\cmidrule(lr){3-6} \cmidrule(lr){7-10} \cmidrule(lr){11-14}
$K$ & $x_0$ & BIAS & MCSD & AESE & CP(\%) & BIAS & MCSD & AESE & CP(\%) & BIAS & MCSD & AESE & CP(\%) \\
\hline
5 & -0.50 & 0.02 & 0.11 & 0.14 & 97.8 & 0.01 & 0.12 & 0.14 & 96.6 & 0.00 & 0.14 & 0.14 & 95.1 \\ 
  & -0.25 & 0.02 & 0.11 & 0.12 & 95.5 & 0.01 & 0.12 & 0.12 & 95.1 & -0.00 & 0.13 & 0.12 & 93.4 \\ 
  & 0.25 & 0.01 & 0.11 & 0.12 & 96.0 & -0.00 & 0.12 & 0.12 & 95.6 & -0.01 & 0.12 & 0.12 & 94.2 \\ 
  & 0.50 & 0.01 & 0.13 & 0.14 & 98.9 & 0.01 & 0.14 & 0.14 & 96.5 & 0.00 & 0.16 & 0.14 & 94.0 \\ 
[2ex]
10 & -0.50 & 0.03 & 0.12 & 0.18 & 98.9 & 0.01 & 0.16 & 0.18 & 96.7 & 0.01 & 0.18 & 0.18 & 94.7 \\ 
  & -0.25 & 0.03 & 0.13 & 0.17 & 98.6 & 0.01 & 0.15 & 0.17 & 96.6 & -0.00 & 0.17 & 0.17 & 93.3 \\ 
  & 0.25 & 0.01 & 0.13 & 0.17 & 98.4 & -0.00 & 0.15 & 0.17 & 95.7 & -0.01 & 0.17 & 0.17 & 94.6 \\ 
  & 0.50 & 0.02 & 0.12 & 0.18 & 99.1 & 0.01 & 0.16 & 0.18 & 96.2 & 0.01 & 0.18 & 0.18 & 94.8 \\ 
[2ex]
20 & -0.50 & 0.03 & 0.13 & 0.28 & 99.7 & 0.01 & 0.23 & 0.28 & 97.7 & 0.00 & 0.28 & 0.28 & 95.7 \\ 
  & -0.25 & 0.03 & 0.14 & 0.26 & 99.8 & -0.00 & 0.23 & 0.26 & 96.7 & -0.01 & 0.26 & 0.26 & 95.5 \\ 
  & 0.25 & 0.01 & 0.14 & 0.25 & 99.5 & -0.00 & 0.23 & 0.25 & 97.0 & -0.00 & 0.26 & 0.25 & 94.6 \\ 
  & 0.50 & 0.02 & 0.13 & 0.27 & 99.8 & 0.01 & 0.24 & 0.27 & 96.7 & 0.01 & 0.28 & 0.27 & 93.3 \\ 
\bottomrule
\end{tabular}%
}
\end{table}

\begin{table}[htbp]
\centering
\caption{Simulation results presenting pointwise inference metrics, including bias, Monte Carlo standard deviation (MCSD), average estimated standard errors (AESE), and empirical coverage probabilities (CP), for the proposed PLSS-DCDR learner for the never-takers stratum ($00$) at four representative points $X=x_0\in\{-0.50,-0.25,0.25,0.50\}$ when non-monotonicity holds. The basis functions are constructed using P-splines with $K \in \{5, 10, 20\}$ knots. We employ three implementation strategies: a penalized estimator with the smoothing parameter selected via the generalized cross-validation (GCV) criterion \cite{eilers1996flexible}, a variant based on an undersmoothed GCV selection, and an unpenalized estimator corresponding to the least squares series approach.}
\label{tab:pointwise-00-compare}
\resizebox{\textwidth}{!}{%
\begin{tabular}{lrr cccc cccc cccc}
\toprule
 & & \multicolumn{4}{c}{GCV} & \multicolumn{4}{c}{GCV.under} & \multicolumn{4}{c}{Unpenalized} \\
\cmidrule(lr){3-6} \cmidrule(lr){7-10} \cmidrule(lr){11-14}
$K$ & $x_0$ & BIAS & MCSD & AESE & CP(\%) & BIAS & MCSD & AESE & CP(\%) & BIAS & MCSD & AESE & CP(\%) \\
\hline
5 & -0.50 & 0.03 & 0.12 & 0.16 & 98.7 & 0.01 & 0.14 & 0.16 & 98.0 & 0.00 & 0.15 & 0.16 & 95.9 \\ 
  & -0.25 & 0.03 & 0.13 & 0.14 & 96.2 & 0.01 & 0.13 & 0.14 & 95.9 & 0.00 & 0.14 & 0.14 & 94.5 \\ 
  & 0.25 & 0.01 & 0.15 & 0.15 & 96.8 & 0.00 & 0.14 & 0.14 & 96.0 & -0.00 & 0.15 & 0.14 & 94.4 \\ 
  & 0.50 & 0.01 & 0.21 & 0.20 & 97.9 & 0.01 & 0.18 & 0.19 & 96.3 & 0.00 & 0.21 & 0.19 & 93.9 \\ 
[2ex]
10 & -0.50 & 0.03 & 0.14 & 0.20 & 99.2 & 0.01 & 0.17 & 0.20 & 98.0 & -0.01 & 0.19 & 0.20 & 96.9 \\ 
  & -0.25 & 0.04 & 0.14 & 0.19 & 99.2 & 0.01 & 0.17 & 0.19 & 96.8 & 0.00 & 0.19 & 0.19 & 95.0 \\ 
  & 0.25 & 0.02 & 0.17 & 0.20 & 98.3 & 0.01 & 0.18 & 0.20 & 96.6 & 0.00 & 0.20 & 0.20 & 94.6 \\ 
  & 0.50 & 0.02 & 0.21 & 0.23 & 99.0 & 0.01 & 0.20 & 0.22 & 96.3 & 0.01 & 0.23 & 0.22 & 94.9 \\ 
[2ex]
20 & -0.50 & 0.02 & 0.15 & 0.31 & 99.5 & -0.00 & 0.26 & 0.31 & 97.2 & -0.01 & 0.31 & 0.31 & 94.8 \\ 
  & -0.25 & 0.04 & 0.15 & 0.29 & 99.6 & 0.01 & 0.25 & 0.29 & 96.5 & 0.00 & 0.29 & 0.29 & 94.7 \\ 
  & 0.25 & 0.03 & 0.18 & 0.31 & 99.5 & 0.01 & 0.27 & 0.31 & 97.0 & 0.01 & 0.32 & 0.31 & 95.2 \\ 
  & 0.50 & 0.01 & 0.24 & 0.35 & 99.2 & -0.00 & 0.33 & 0.35 & 94.6 & -0.00 & 0.38 & 0.35 & 92.0 \\ 
\bottomrule
\end{tabular}%
}
\end{table}

\begin{table}[htbp]
\centering
\caption{ Simulation results presenting pointwise inference metrics, including bias, Monte Carlo standard deviation (MCSD), average estimated standard errors (AESE), and empirical coverage probabilities (CP), for the proposed PLSS-DCDR learner for the defiers stratum ($10$) at four representative points $X=x_0\in\{-0.50,-0.25,0.25,0.50\}$ when non-monotonicity holds. The basis functions are constructed using P-splines with $K \in \{5, 10, 20\}$ knots. We employ three implementation strategies: a penalized estimator with the smoothing parameter selected via the generalized cross-validation (GCV) criterion \cite{eilers1996flexible}, a variant based on an undersmoothed GCV selection, and an unpenalized estimator corresponding to the least squares series approach.}
\label{tab:pointwise-10-compare}
\resizebox{\textwidth}{!}{%
\begin{tabular}{lrr cccc cccc cccc}
\toprule
 & & \multicolumn{4}{c}{GCV} & \multicolumn{4}{c}{GCV.under} & \multicolumn{4}{c}{Unpenalized} \\
\cmidrule(lr){3-6} \cmidrule(lr){7-10} \cmidrule(lr){11-14}
$K$ & $x_0$ & BIAS & MCSD & AESE & CP(\%) & BIAS & MCSD & AESE & CP(\%) & BIAS & MCSD & AESE & CP(\%) \\
\hline
5 & -0.50 & 0.01 & 0.12 & 0.15 & 98.1 & 0.00 & 0.14 & 0.15 & 97.3 & -0.01 & 0.15 & 0.15 & 95.8 \\ 
  & -0.25 & 0.01 & 0.12 & 0.13 & 97.1 & 0.00 & 0.13 & 0.13 & 96.6 & -0.00 & 0.14 & 0.13 & 95.2 \\ 
  & 0.25 & 0.02 & 0.13 & 0.12 & 96.5 & 0.01 & 0.13 & 0.12 & 96.5 & -0.00 & 0.13 & 0.12 & 95.3 \\ 
  & 0.50 & 0.02 & 0.17 & 0.15 & 98.3 & 0.00 & 0.19 & 0.15 & 96.4 & -0.01 & 0.21 & 0.15 & 94.4 \\ 
[2ex]
10 & -0.50 & 0.02 & 0.13 & 0.19 & 99.0 & 0.01 & 0.17 & 0.19 & 97.3 & -0.00 & 0.19 & 0.19 & 96.1 \\ 
  & -0.25 & 0.02 & 0.12 & 0.17 & 99.3 & 0.00 & 0.16 & 0.17 & 97.2 & -0.00 & 0.17 & 0.17 & 95.5 \\ 
  & 0.25 & 0.02 & 0.12 & 0.17 & 99.4 & 0.00 & 0.15 & 0.17 & 97.2 & -0.01 & 0.16 & 0.17 & 96.0 \\ 
  & 0.50 & 0.03 & 0.13 & 0.17 & 98.9 & 0.02 & 0.16 & 0.17 & 96.2 & 0.01 & 0.18 & 0.17 & 94.0 \\ 
[2ex]
20 & -0.50 & 0.02 & 0.15 & 0.30 & 99.7 & 0.00 & 0.25 & 0.30 & 97.3 & -0.00 & 0.30 & 0.30 & 95.3 \\ 
  & -0.25 & 0.02 & 0.14 & 0.27 & 99.3 & -0.00 & 0.23 & 0.27 & 97.6 & -0.01 & 0.26 & 0.27 & 95.4 \\ 
  & 0.25 & 0.03 & 0.13 & 0.25 & 99.8 & 0.00 & 0.21 & 0.25 & 97.2 & 0.01 & 0.25 & 0.25 & 94.5 \\ 
  & 0.50 & 0.02 & 0.14 & 0.26 & 99.6 & 0.00 & 0.22 & 0.26 & 97.1 & 0.00 & 0.26 & 0.26 & 94.9 \\ 
\bottomrule
\end{tabular}%
}
\end{table}

\begin{table}[htbp]
\centering
\caption{ Simulation results presenting pointwise inference metrics, including bias, Monte Carlo standard deviation (MCSD), average estimated standard errors (AESE), and empirical coverage probabilities (CP), for the proposed PLSS-DCDR learner for the always-takers stratum ($11$) at four representative points $X=x_0\in\{-0.50,-0.25,0.25,0.50\}$ when non-monotonicity holds but the estimator is fitted assuming monotonicity. The basis functions are constructed using P-splines with $K \in \{5, 10, 20\}$ knots. We employ three implementation strategies: a penalized estimator with the smoothing parameter selected via the generalized cross-validation (GCV) criterion \cite{eilers1996flexible}, a variant based on an undersmoothed GCV selection, and an unpenalized estimator corresponding to the least squares series approach.}
\label{tab:pointwise-11-nonmono-fitmono}
\resizebox{\textwidth}{!}{%
\begin{tabular}{lrr cccc cccc cccc}
\toprule
 & & \multicolumn{4}{c}{GCV} & \multicolumn{4}{c}{GCV.under} & \multicolumn{4}{c}{Unpenalized} \\
\cmidrule(lr){3-6} \cmidrule(lr){7-10} \cmidrule(lr){11-14}
$K$ & $x_0$ & BIAS & MCSD & AESE & CP(\%) & BIAS & MCSD & AESE & CP(\%) & BIAS & MCSD & AESE & CP(\%) \\
\hline
5 & -0.50 & 0.20 & 0.12 & 0.14 & 69.1 & 0.19 & 0.16 & 0.14 & 69.8 & 0.18 & 0.19 & 0.14 & 70.8 \\ 
  & -0.25 & 0.21 & 0.31 & 0.12 & 59.0 & 0.20 & 0.35 & 0.12 & 63.4 & 0.19 & 0.38 & 0.12 & 64.4 \\ 
  & 0.25 & 0.11 & 0.41 & 0.12 & 77.9 & 0.10 & 0.41 & 0.12 & 80.5 & 0.09 & 0.41 & 0.12 & 81.6 \\ 
  & 0.50 & 0.08 & 0.74 & 0.16 & 91.6 & 0.06 & 0.83 & 0.16 & 90.4 & 0.05 & 0.89 & 0.16 & 87.7 \\ 
[2ex]
10 & -0.50 & 0.21 & 0.14 & 0.17 & 85.8 & 0.20 & 0.17 & 0.17 & 82.5 & 0.20 & 0.19 & 0.17 & 79.3 \\ 
  & -0.25 & 0.19 & 0.17 & 0.16 & 82.9 & 0.18 & 0.19 & 0.16 & 80.9 & 0.17 & 0.21 & 0.16 & 79.0 \\ 
  & 0.25 & 0.12 & 0.26 & 0.16 & 94.3 & 0.11 & 0.28 & 0.16 & 90.6 & 0.10 & 0.30 & 0.16 & 88.0 \\ 
  & 0.50 & 0.12 & 0.21 & 0.16 & 96.8 & 0.10 & 0.24 & 0.16 & 93.2 & 0.10 & 0.25 & 0.16 & 90.5 \\ 
[2ex]
20 & -0.50 & 0.21 & 0.12 & 0.26 & 98.6 & 0.20 & 0.22 & 0.26 & 92.0 & 0.19 & 0.25 & 0.26 & 88.7 \\ 
  & -0.25 & 0.20 & 0.12 & 0.24 & 98.0 & 0.17 & 0.20 & 0.24 & 92.5 & 0.17 & 0.24 & 0.24 & 88.7 \\ 
  & 0.25 & 0.13 & 0.13 & 0.23 & 99.4 & 0.12 & 0.20 & 0.23 & 94.8 & 0.11 & 0.24 & 0.23 & 91.2 \\ 
  & 0.50 & 0.10 & 0.25 & 0.25 & 99.6 & 0.09 & 0.31 & 0.25 & 96.2 & 0.09 & 0.36 & 0.25 & 93.5 \\ 
\bottomrule
\end{tabular}%
}
\end{table}

\begin{table}[htbp]
\centering
\caption{ Simulation results presenting pointwise inference metrics, including bias, Monte Carlo standard deviation (MCSD), average estimated standard errors (AESE), and empirical coverage probabilities (CP), for the proposed PLSS-DCDR learner for the compliers stratum ($01$) at four representative points $X=x_0\in\{-0.50,-0.25,0.25,0.50\}$ when non-monotonicity holds but the estimator is fitted assuming monotonicity. The basis functions are constructed using P-splines with $K \in \{5, 10, 20\}$ knots. We employ three implementation strategies: a penalized estimator with the smoothing parameter selected via the generalized cross-validation (GCV) criterion \cite{eilers1996flexible}, a variant based on an undersmoothed GCV selection, and an unpenalized estimator corresponding to the least squares series approach.}
\label{tab:pointwise-01-nonmono-fitmono}
\resizebox{\textwidth}{!}{%
\begin{tabular}{lrr cccc cccc cccc}
\toprule
 & & \multicolumn{4}{c}{GCV} & \multicolumn{4}{c}{GCV.under} & \multicolumn{4}{c}{Unpenalized} \\
\cmidrule(lr){3-6} \cmidrule(lr){7-10} \cmidrule(lr){11-14}
$K$ & $x_0$ & BIAS & MCSD & AESE & CP(\%) & BIAS & MCSD & AESE & CP(\%) & BIAS & MCSD & AESE & CP(\%) \\
\hline
5 & -0.50 & $>10^{8}$ & $>10^{10}$ & $>10^{7}$ & 48.8 & $>10^{8}$ & $>10^{10}$ & $>10^{7}$ & 56.9 & $>10^{7}$ & $>10^{8}$ & $>10^{7}$ & 99.6 \\ 
  & -0.25 & $>10^{8}$ & $>10^{10}$ & $>10^{9}$ & 99.7 & $>10^{8}$ & $>10^{10}$ & $>10^{9}$ & 99.4 & $>10^{9}$ & $>10^{10}$ & $>10^{9}$ & 99.4 \\ 
  & 0.25 & $>10^{8}$ & $>10^{10}$ & $>10^{9}$ & 99.2 & $>10^{8}$ & $>10^{10}$ & $>10^{9}$ & 98.8 & $>10^{9}$ & $>10^{10}$ & $>10^{9}$ & 99.8 \\ 
  & 0.50 & $>10^{8}$ & $>10^{10}$ & $>10^{7}$ & 44.5 & $>10^{8}$ & $>10^{10}$ & $>10^{7}$ & 55.6 & $<-10^{7}$ & $>10^{8}$ & $>10^{7}$ & 99.6 \\ 
[2ex]
10 & -0.50 & $>10^{8}$ & $>10^{10}$ & $>10^{8}$ & 89.0 & $>10^{7}$ & $>10^{8}$ & $>10^{8}$ & 93.0 & $<-10^{8}$ & $>10^{10}$ & $>10^{8}$ & 99.4 \\ 
  & -0.25 & $>10^{8}$ & $>10^{10}$ & $>10^{8}$ & 90.6 & $>10^{9}$ & $>10^{10}$ & $>10^{8}$ & 89.7 & $>10^{8}$ & $>10^{10}$ & $>10^{8}$ & 99.3 \\ 
  & 0.25 & $>10^{8}$ & $>10^{10}$ & $>10^{8}$ & 88.9 & $>10^{9}$ & $>10^{10}$ & $>10^{8}$ & 89.7 & $>10^{8}$ & $>10^{10}$ & $>10^{8}$ & 99.7 \\ 
  & 0.50 & $>10^{8}$ & $>10^{10}$ & $>10^{8}$ & 88.2 & $<-10^{6}$ & $>10^{8}$ & $>10^{8}$ & 94.0 & $<-10^{8}$ & $>10^{10}$ & $>10^{8}$ & 99.7 \\ 
[2ex]
20 & -0.50 & $>10^{8}$ & $>10^{10}$ & $>10^{7}$ & 7.0 & $>10^{7}$ & $>10^{8}$ & $>10^{6}$ & 28.9 & $<-10^{6}$ & $>10^{8}$ & $>10^{7}$ & 100.0 \\ 
  & -0.25 & $>10^{8}$ & $>10^{10}$ & $>10^{8}$ & 80.0 & $<-10^{8}$ & $>10^{9}$ & $>10^{8}$ & 86.5 & $<-10^{8}$ & $>10^{10}$ & $>10^{8}$ & 99.5 \\ 
  & 0.25 & $>10^{8}$ & $>10^{10}$ & $>10^{8}$ & 73.4 & $<-10^{8}$ & $>10^{9}$ & $>10^{8}$ & 88.4 & $<-10^{8}$ & $>10^{10}$ & $>10^{8}$ & 99.6 \\ 
  & 0.50 & $>10^{8}$ & $>10^{10}$ & $>10^{7}$ & 7.4 & $>10^{7}$ & $>10^{8}$ & $>10^{6}$ & 34.3 & 8985.56 & $>10^{7}$ & $>10^{7}$ & 100.0 \\ 
\bottomrule
\end{tabular}%
}
\end{table}

\begin{table}[htbp]
\centering
\caption{ Simulation results presenting pointwise inference metrics, including bias, Monte Carlo standard deviation (MCSD), average estimated standard errors (AESE), and empirical coverage probabilities (CP), for the proposed PLSS-DCDR learner for the never-takers stratum ($00$) at four representative points $X=x_0\in\{-0.50,-0.25,0.25,0.50\}$ when non-monotonicity holds but the estimator is fitted assuming monotonicity. The basis functions are constructed using P-splines with $K \in \{5, 10, 20\}$ knots. We employ three implementation strategies: a penalized estimator with the smoothing parameter selected via the generalized cross-validation (GCV) criterion \cite{eilers1996flexible}, a variant based on an undersmoothed GCV selection, and an unpenalized estimator corresponding to the least squares series approach.}
\label{tab:pointwise-00-nonmono-fitmono}
\resizebox{\textwidth}{!}{%
\begin{tabular}{lrr cccc cccc cccc}
\toprule
 & & \multicolumn{4}{c}{GCV} & \multicolumn{4}{c}{GCV.under} & \multicolumn{4}{c}{Unpenalized} \\
\cmidrule(lr){3-6} \cmidrule(lr){7-10} \cmidrule(lr){11-14}
$K$ & $x_0$ & BIAS & MCSD & AESE & CP(\%) & BIAS & MCSD & AESE & CP(\%) & BIAS & MCSD & AESE & CP(\%) \\
\hline
5 & -0.50 & 0.13 & 0.09 & 0.05 & 35.2 & 0.13 & 0.09 & 0.05 & 36.3 & 0.13 & 0.10 & 0.05 & 36.8 \\ 
  & -0.25 & 0.22 & 0.09 & 0.05 & 7.6 & 0.21 & 0.09 & 0.05 & 9.4 & 0.21 & 0.09 & 0.05 & 11.1 \\ 
  & 0.25 & 0.15 & 0.09 & 0.05 & 24.5 & 0.15 & 0.09 & 0.05 & 26.3 & 0.14 & 0.09 & 0.05 & 28.0 \\ 
  & 0.50 & 0.03 & 0.10 & 0.06 & 75.3 & 0.02 & 0.10 & 0.06 & 73.1 & 0.02 & 0.11 & 0.06 & 70.9 \\ 
[2ex]
10 & -0.50 & 0.13 & 0.10 & 0.07 & 50.3 & 0.13 & 0.11 & 0.07 & 52.2 & 0.13 & 0.11 & 0.07 & 51.3 \\ 
  & -0.25 & 0.22 & 0.10 & 0.07 & 19.7 & 0.21 & 0.11 & 0.07 & 23.0 & 0.21 & 0.11 & 0.07 & 24.3 \\ 
  & 0.25 & 0.15 & 0.10 & 0.07 & 41.3 & 0.15 & 0.11 & 0.07 & 44.3 & 0.14 & 0.11 & 0.07 & 45.3 \\ 
  & 0.50 & 0.03 & 0.11 & 0.08 & 82.6 & 0.03 & 0.12 & 0.08 & 79.6 & 0.03 & 0.12 & 0.08 & 78.0 \\ 
[2ex]
20 & -0.50 & 0.14 & 0.10 & 0.11 & 77.4 & 0.13 & 0.13 & 0.11 & 73.2 & 0.13 & 0.14 & 0.11 & 71.5 \\ 
  & -0.25 & 0.22 & 0.11 & 0.10 & 43.4 & 0.21 & 0.13 & 0.10 & 46.3 & 0.21 & 0.14 & 0.10 & 45.3 \\ 
  & 0.25 & 0.15 & 0.11 & 0.11 & 71.5 & 0.14 & 0.13 & 0.11 & 69.7 & 0.14 & 0.14 & 0.11 & 67.1 \\ 
  & 0.50 & 0.03 & 0.12 & 0.12 & 94.7 & 0.02 & 0.15 & 0.12 & 87.3 & 0.03 & 0.16 & 0.12 & 83.6 \\ 
\bottomrule
\end{tabular}%
}
\end{table}

\begin{table}[htbp]
\centering
\caption{ Simulation results presenting pointwise inference metrics, including bias, Monte Carlo standard deviation (MCSD), average estimated standard errors (AESE), and empirical coverage probabilities (CP), for the proposed PLSS-DCDR learner for the always-takers stratum ($11$) at four representative points $X=x_0\in\{-0.50,-0.25,0.25,0.50\}$ when monotonicity holds. The basis functions are constructed using P-splines with $K \in \{5, 10, 20\}$ knots. We employ three implementation strategies: a penalized estimator with the smoothing parameter selected via the generalized cross-validation (GCV) criterion \cite{eilers1996flexible}, a variant based on an undersmoothed GCV selection, and an unpenalized estimator corresponding to the least squares series approach.}
\label{tab:pointwise-11-mono-fitmono}
\resizebox{\textwidth}{!}{%
\begin{tabular}{lrr cccc cccc cccc}
\toprule
 & & \multicolumn{4}{c}{GCV} & \multicolumn{4}{c}{GCV.under} & \multicolumn{4}{c}{Unpenalized} \\
\cmidrule(lr){3-6} \cmidrule(lr){7-10} \cmidrule(lr){11-14}
$K$ & $x_0$ & BIAS & MCSD & AESE & CP(\%) & BIAS & MCSD & AESE & CP(\%) & BIAS & MCSD & AESE & CP(\%) \\
\hline
5 & -0.50 & 0.00 & 0.11 & 0.15 & 99.7 & 0.00 & 0.12 & 0.15 & 98.2 & 0.00 & 0.15 & 0.15 & 95.8 \\ 
  & -0.25 & 0.00 & 0.11 & 0.14 & 98.6 & 0.00 & 0.12 & 0.13 & 96.4 & 0.00 & 0.13 & 0.13 & 94.3 \\ 
  & 0.25 & 0.00 & 0.12 & 0.14 & 98.3 & 0.00 & 0.13 & 0.14 & 96.6 & 0.00 & 0.14 & 0.14 & 94.4 \\ 
  & 0.50 & 0.00 & 0.17 & 0.18 & 98.9 & 0.00 & 0.16 & 0.17 & 98.3 & 0.01 & 0.18 & 0.17 & 95.8 \\ 
[2ex]
10 & -0.50 & 0.00 & 0.12 & 0.20 & 99.6 & 0.00 & 0.17 & 0.20 & 96.7 & 0.00 & 0.20 & 0.20 & 94.6 \\ 
  & -0.25 & 0.00 & 0.11 & 0.19 & 99.6 & 0.00 & 0.17 & 0.19 & 95.6 & -0.00 & 0.19 & 0.19 & 93.8 \\ 
  & 0.25 & 0.00 & 0.12 & 0.19 & 99.8 & 0.00 & 0.17 & 0.19 & 97.7 & -0.00 & 0.19 & 0.19 & 95.1 \\ 
  & 0.50 & 0.01 & 0.13 & 0.21 & 99.3 & 0.02 & 0.18 & 0.21 & 96.3 & 0.01 & 0.21 & 0.21 & 94.9 \\ 
[2ex]
20 & -0.50 & 0.01 & 0.14 & 0.30 & 99.5 & 0.00 & 0.26 & 0.30 & 96.4 & 0.00 & 0.31 & 0.30 & 93.8 \\ 
  & -0.25 & -0.00 & 0.13 & 0.28 & 99.6 & -0.00 & 0.25 & 0.28 & 96.7 & -0.00 & 0.29 & 0.28 & 94.7 \\ 
  & 0.25 & -0.00 & 0.14 & 0.30 & 99.5 & -0.00 & 0.26 & 0.30 & 95.9 & 0.00 & 0.30 & 0.30 & 94.5 \\ 
  & 0.50 & 0.01 & 0.15 & 0.33 & 99.5 & 0.01 & 0.28 & 0.33 & 96.5 & 0.01 & 0.34 & 0.33 & 94.4 \\ 
\bottomrule
\end{tabular}%
}
\end{table}

\begin{table}[htbp]
\centering
\caption{ Simulation results presenting pointwise inference metrics, including bias, Monte Carlo standard deviation (MCSD), average estimated standard errors (AESE), and empirical coverage probabilities (CP), for the proposed PLSS-DCDR learner for the compliers stratum ($01$) at four representative points $X=x_0\in\{-0.50,-0.25,0.25,0.50\}$ when monotonicity holds. The basis functions are constructed using P-splines with $K \in \{5, 10, 20\}$ knots. We employ three implementation strategies: a penalized estimator with the smoothing parameter selected via the generalized cross-validation (GCV) criterion \cite{eilers1996flexible}, a variant based on an undersmoothed GCV selection, and an unpenalized estimator corresponding to the least squares series approach.}
\label{tab:pointwise-01-mono-fitmono}
\resizebox{\textwidth}{!}{%
\begin{tabular}{lrr cccc cccc cccc}
\toprule
 & & \multicolumn{4}{c}{GCV} & \multicolumn{4}{c}{GCV.under} & \multicolumn{4}{c}{Unpenalized} \\
\cmidrule(lr){3-6} \cmidrule(lr){7-10} \cmidrule(lr){11-14}
$K$ & $x_0$ & BIAS & MCSD & AESE & CP(\%) & BIAS & MCSD & AESE & CP(\%) & BIAS & MCSD & AESE & CP(\%) \\
\hline
5 & -0.50 & 0.00 & 0.10 & 0.15 & 98.9 & 0.00 & 0.12 & 0.15 & 97.6 & 0.01 & 0.15 & 0.15 & 95.2 \\ 
  & -0.25 & 0.00 & 0.11 & 0.14 & 98.6 & 0.00 & 0.13 & 0.14 & 97.3 & 0.00 & 0.15 & 0.14 & 95.2 \\ 
  & 0.25 & 0.00 & 0.15 & 0.15 & 97.7 & -0.00 & 0.15 & 0.15 & 96.5 & -0.01 & 0.16 & 0.15 & 96.1 \\ 
  & 0.50 & 0.00 & 0.22 & 0.20 & 98.5 & 0.00 & 0.24 & 0.20 & 97.4 & 0.00 & 0.27 & 0.20 & 94.7 \\ 
[2ex]
10 & -0.50 & 0.01 & 0.11 & 0.19 & 99.5 & 0.01 & 0.16 & 0.19 & 96.8 & 0.01 & 0.19 & 0.19 & 94.7 \\ 
  & -0.25 & -0.00 & 0.11 & 0.18 & 99.5 & -0.00 & 0.16 & 0.18 & 97.4 & 0.00 & 0.18 & 0.18 & 95.8 \\ 
  & 0.25 & 0.00 & 0.14 & 0.20 & 98.9 & -0.00 & 0.17 & 0.20 & 96.8 & -0.01 & 0.19 & 0.20 & 96.0 \\ 
  & 0.50 & 0.01 & 0.18 & 0.23 & 98.5 & 0.01 & 0.21 & 0.23 & 96.4 & 0.01 & 0.23 & 0.23 & 94.5 \\ 
[2ex]
20 & -0.50 & 0.01 & 0.13 & 0.29 & 99.6 & 0.01 & 0.24 & 0.29 & 97.2 & 0.01 & 0.29 & 0.29 & 94.6 \\ 
  & -0.25 & 0.00 & 0.13 & 0.28 & 99.7 & -0.00 & 0.24 & 0.28 & 97.4 & -0.01 & 0.28 & 0.28 & 95.6 \\ 
  & 0.25 & 0.00 & 0.15 & 0.31 & 99.6 & -0.01 & 0.27 & 0.31 & 96.6 & -0.01 & 0.31 & 0.31 & 94.6 \\ 
  & 0.50 & 0.01 & 0.20 & 0.36 & 99.4 & 0.03 & 0.31 & 0.35 & 96.7 & 0.03 & 0.37 & 0.35 & 94.3 \\ 
\bottomrule
\end{tabular}%
}
\end{table}

\begin{table}[htbp]
\centering
\caption{ Simulation results presenting pointwise inference metrics, including bias, Monte Carlo standard deviation (MCSD), average estimated standard errors (AESE), and empirical coverage probabilities (CP), for the proposed PLSS-DCDR learner for the never-takers stratum ($00$) at four representative points $X=x_0\in\{-0.50,-0.25,0.25,0.50\}$ when monotonicity holds. The basis functions are constructed using P-splines with $K \in \{5, 10, 20\}$ knots. We employ three implementation strategies: a penalized estimator with the smoothing parameter selected via the generalized cross-validation (GCV) criterion \cite{eilers1996flexible}, a variant based on an undersmoothed GCV selection, and an unpenalized estimator corresponding to the least squares series approach.}
\label{tab:pointwise-00-mono-fitmono}
\resizebox{\textwidth}{!}{%
\begin{tabular}{lrr cccc cccc cccc}
\toprule
 & & \multicolumn{4}{c}{GCV} & \multicolumn{4}{c}{GCV.under} & \multicolumn{4}{c}{Unpenalized} \\
\cmidrule(lr){3-6} \cmidrule(lr){7-10} \cmidrule(lr){11-14}
$K$ & $x_0$ & BIAS & MCSD & AESE & CP(\%) & BIAS & MCSD & AESE & CP(\%) & BIAS & MCSD & AESE & CP(\%) \\
\hline
5 & -0.50 & -0.01 & 0.11 & 0.15 & 99.1 & -0.01 & 0.12 & 0.15 & 97.8 & -0.01 & 0.15 & 0.15 & 95.0 \\ 
  & -0.25 & -0.00 & 0.10 & 0.12 & 97.8 & -0.00 & 0.11 & 0.12 & 96.9 & -0.00 & 0.13 & 0.12 & 94.5 \\ 
  & 0.25 & -0.00 & 0.09 & 0.12 & 98.6 & -0.00 & 0.11 & 0.12 & 97.0 & -0.00 & 0.12 & 0.12 & 95.4 \\ 
  & 0.50 & 0.00 & 0.11 & 0.14 & 98.5 & -0.00 & 0.13 & 0.14 & 96.8 & -0.01 & 0.15 & 0.14 & 94.6 \\ 
[2ex]
10 & -0.50 & -0.01 & 0.12 & 0.19 & 99.5 & -0.01 & 0.16 & 0.19 & 97.2 & -0.01 & 0.19 & 0.19 & 95.9 \\ 
  & -0.25 & -0.00 & 0.11 & 0.18 & 98.9 & 0.00 & 0.15 & 0.18 & 97.0 & 0.00 & 0.18 & 0.18 & 95.2 \\ 
  & 0.25 & -0.00 & 0.11 & 0.17 & 99.0 & -0.01 & 0.15 & 0.17 & 97.0 & -0.01 & 0.17 & 0.17 & 94.2 \\ 
  & 0.50 & 0.01 & 0.11 & 0.18 & 99.3 & 0.00 & 0.15 & 0.18 & 96.2 & 0.00 & 0.17 & 0.18 & 94.0 \\ 
[2ex]
20 & -0.50 & -0.01 & 0.13 & 0.30 & 99.9 & -0.01 & 0.24 & 0.30 & 97.9 & -0.01 & 0.30 & 0.30 & 96.0 \\ 
  & -0.25 & -0.00 & 0.12 & 0.27 & 99.9 & 0.01 & 0.23 & 0.27 & 97.1 & 0.00 & 0.27 & 0.27 & 94.5 \\ 
  & 0.25 & -0.00 & 0.12 & 0.25 & 99.8 & -0.00 & 0.23 & 0.25 & 96.8 & -0.00 & 0.27 & 0.25 & 94.0 \\ 
  & 0.50 & 0.00 & 0.12 & 0.27 & 99.7 & -0.00 & 0.23 & 0.27 & 96.8 & -0.01 & 0.27 & 0.27 & 94.5 \\ 
\bottomrule
\end{tabular}%
}
\end{table}


\begin{table}[htbp]
\centering
\caption{ Simulation results presenting pointwise inference metrics, including bias, Monte Carlo standard deviation (MCSD), average estimated standard errors (AESE), and empirical coverage probabilities (CP), for the proposed PLSS-DCDR learner for the always-takers stratum ($11$) at four representative points $X=x_0\in\{-0.50,-0.25,0.25,0.50\}$, when the true data-generating process follows monotonicity but the model is fitted under non-monotonicity with $\theta=5$. The basis functions are constructed using P-splines with $K \in \{5, 10, 20\}$ knots. We employ three implementation strategies: a penalized estimator with the smoothing parameter selected via the generalized cross-validation (GCV) criterion \cite{eilers1996flexible}, a variant based on an undersmoothed GCV selection, and an unpenalized estimator corresponding to the least squares series approach.}
\label{tab:pointwise-11-mono-fit5}
\resizebox{\textwidth}{!}{%
\begin{tabular}{lrr cccc cccc cccc}
\toprule
 & & \multicolumn{4}{c}{GCV} & \multicolumn{4}{c}{GCV.under} & \multicolumn{4}{c}{Unpenalized} \\
\cmidrule(lr){3-6} \cmidrule(lr){7-10} \cmidrule(lr){11-14}
$K$ & $x_0$ & BIAS & MCSD & AESE & CP(\%) & BIAS & MCSD & AESE & CP(\%) & BIAS & MCSD & AESE & CP(\%) \\
\hline
5 & -0.50 & -0.07 & 0.10 & 0.15 & 98.0 & -0.07 & 0.12 & 0.15 & 96.0 & -0.07 & 0.14 & 0.15 & 92.4 \\ 
  & -0.25 & -0.07 & 0.11 & 0.13 & 96.0 & -0.07 & 0.12 & 0.13 & 94.4 & -0.07 & 0.14 & 0.13 & 91.0 \\ 
  & 0.25 & -0.08 & 0.12 & 0.13 & 95.0 & -0.08 & 0.13 & 0.13 & 92.9 & -0.08 & 0.14 & 0.13 & 89.0 \\ 
  & 0.50 & -0.09 & 0.16 & 0.16 & 97.5 & -0.09 & 0.18 & 0.16 & 95.6 & -0.09 & 0.20 & 0.16 & 92.1 \\ 
[2ex]
10 & -0.50 & -0.07 & 0.12 & 0.19 & 98.9 & -0.07 & 0.16 & 0.19 & 95.8 & -0.07 & 0.19 & 0.19 & 93.8 \\ 
  & -0.25 & -0.07 & 0.11 & 0.18 & 99.1 & -0.08 & 0.16 & 0.18 & 94.3 & -0.08 & 0.19 & 0.18 & 92.3 \\ 
  & 0.25 & -0.08 & 0.12 & 0.18 & 98.0 & -0.08 & 0.16 & 0.18 & 94.3 & -0.09 & 0.18 & 0.18 & 91.5 \\ 
  & 0.50 & -0.08 & 0.13 & 0.20 & 98.5 & -0.08 & 0.17 & 0.20 & 95.1 & -0.08 & 0.20 & 0.20 & 92.2 \\ 
[2ex]
20 & -0.50 & -0.06 & 0.14 & 0.30 & 99.6 & -0.07 & 0.25 & 0.30 & 95.8 & -0.07 & 0.31 & 0.30 & 92.7 \\ 
  & -0.25 & -0.08 & 0.13 & 0.28 & 99.5 & -0.08 & 0.24 & 0.28 & 96.3 & -0.08 & 0.28 & 0.28 & 94.9 \\ 
  & 0.25 & -0.09 & 0.14 & 0.28 & 99.3 & -0.09 & 0.25 & 0.28 & 95.0 & -0.08 & 0.29 & 0.28 & 92.7 \\ 
  & 0.50 & -0.08 & 0.15 & 0.31 & 99.6 & -0.08 & 0.27 & 0.31 & 96.7 & -0.09 & 0.32 & 0.31 & 93.2 \\ 
\bottomrule
\end{tabular}%
}
\end{table}

\begin{table}[htbp]
\centering
\caption{ Simulation results presenting pointwise inference metrics, including bias, Monte Carlo standard deviation (MCSD), average estimated standard errors (AESE), and empirical coverage probabilities (CP), for the proposed PLSS-DCDR learner for the compliers stratum ($01$) at four representative points $X=x_0\in\{-0.50,-0.25,0.25,0.50\}$, when the true data-generating process follows monotonicity but the model is fitted under non-monotonicity with $\theta=5$. The basis functions are constructed using P-splines with $K \in \{5, 10, 20\}$ knots. We employ three implementation strategies: a penalized estimator with the smoothing parameter selected via the generalized cross-validation (GCV) criterion \cite{eilers1996flexible}, a variant based on an undersmoothed GCV selection, and an unpenalized estimator corresponding to the least squares series approach.}
\label{tab:pointwise-01-mono-fit5}
\resizebox{\textwidth}{!}{%
\begin{tabular}{lrr cccc cccc cccc}
\toprule
 & & \multicolumn{4}{c}{GCV} & \multicolumn{4}{c}{GCV.under} & \multicolumn{4}{c}{Unpenalized} \\
\cmidrule(lr){3-6} \cmidrule(lr){7-10} \cmidrule(lr){11-14}
$K$ & $x_0$ & BIAS & MCSD & AESE & CP(\%) & BIAS & MCSD & AESE & CP(\%) & BIAS & MCSD & AESE & CP(\%) \\
\hline
5 & -0.50 & 0.07 & 0.08 & 0.12 & 96.8 & 0.08 & 0.10 & 0.12 & 94.4 & 0.08 & 0.12 & 0.12 & 90.4 \\ 
  & -0.25 & 0.08 & 0.09 & 0.11 & 95.4 & 0.08 & 0.10 & 0.11 & 91.7 & 0.08 & 0.12 & 0.11 & 87.9 \\ 
  & 0.25 & 0.10 & 0.10 & 0.11 & 91.5 & 0.10 & 0.11 & 0.11 & 89.1 & 0.10 & 0.12 & 0.11 & 85.4 \\ 
  & 0.50 & 0.12 & 0.14 & 0.14 & 93.1 & 0.12 & 0.15 & 0.14 & 90.4 & 0.12 & 0.17 & 0.14 & 84.3 \\ 
[2ex]
10 & -0.50 & 0.08 & 0.10 & 0.16 & 98.8 & 0.08 & 0.14 & 0.16 & 94.2 & 0.08 & 0.16 & 0.16 & 92.0 \\ 
  & -0.25 & 0.08 & 0.09 & 0.15 & 98.1 & 0.08 & 0.13 & 0.15 & 95.2 & 0.08 & 0.15 & 0.15 & 91.4 \\ 
  & 0.25 & 0.10 & 0.10 & 0.16 & 97.3 & 0.10 & 0.13 & 0.16 & 93.2 & 0.10 & 0.15 & 0.16 & 90.6 \\ 
  & 0.50 & 0.12 & 0.11 & 0.17 & 97.3 & 0.12 & 0.15 & 0.17 & 91.5 & 0.12 & 0.17 & 0.17 & 87.1 \\ 
[2ex]
20 & -0.50 & 0.08 & 0.11 & 0.24 & 99.6 & 0.08 & 0.21 & 0.24 & 96.8 & 0.08 & 0.24 & 0.24 & 94.3 \\ 
  & -0.25 & 0.08 & 0.11 & 0.23 & 99.4 & 0.08 & 0.20 & 0.23 & 96.6 & 0.07 & 0.23 & 0.23 & 93.9 \\ 
  & 0.25 & 0.11 & 0.12 & 0.24 & 99.2 & 0.10 & 0.20 & 0.24 & 95.8 & 0.10 & 0.24 & 0.24 & 93.2 \\ 
  & 0.50 & 0.12 & 0.13 & 0.26 & 98.9 & 0.12 & 0.23 & 0.26 & 94.9 & 0.13 & 0.27 & 0.26 & 91.4 \\ 
\bottomrule
\end{tabular}%
}
\end{table}

\begin{table}[htbp]
\centering
\caption{ Simulation results presenting pointwise inference metrics, including bias, Monte Carlo standard deviation (MCSD), average estimated standard errors (AESE), and empirical coverage probabilities (CP), for the proposed PLSS-DCDR learner for the never-takers stratum ($00$) at four representative points $X=x_0\in\{-0.50,-0.25,0.25,0.50\}$, when the true data-generating process follows monotonicity but the model is fitted under non-monotonicity with $\theta=5$. The basis functions are constructed using P-splines with $K \in \{5, 10, 20\}$ knots. We employ three implementation strategies: a penalized estimator with the smoothing parameter selected via the generalized cross-validation (GCV) criterion \cite{eilers1996flexible}, a variant based on an undersmoothed GCV selection, and an unpenalized estimator corresponding to the least squares series approach.}
\label{tab:pointwise-00-mono-fit5}
\resizebox{\textwidth}{!}{%
\begin{tabular}{lrr cccc cccc cccc}
\toprule
 & & \multicolumn{4}{c}{GCV} & \multicolumn{4}{c}{GCV.under} & \multicolumn{4}{c}{Unpenalized} \\
\cmidrule(lr){3-6} \cmidrule(lr){7-10} \cmidrule(lr){11-14}
$K$ & $x_0$ & BIAS & MCSD & AESE & CP(\%) & BIAS & MCSD & AESE & CP(\%) & BIAS & MCSD & AESE & CP(\%) \\
\hline
5 & -0.50 & -0.02 & 0.11 & 0.15 & 98.8 & -0.03 & 0.12 & 0.15 & 97.6 & -0.03 & 0.15 & 0.15 & 94.5 \\ 
  & -0.25 & -0.02 & 0.10 & 0.12 & 97.7 & -0.02 & 0.11 & 0.12 & 95.9 & -0.02 & 0.13 & 0.12 & 94.3 \\ 
  & 0.25 & -0.01 & 0.09 & 0.12 & 98.4 & -0.01 & 0.10 & 0.12 & 97.4 & -0.02 & 0.12 & 0.12 & 95.0 \\ 
  & 0.50 & -0.01 & 0.10 & 0.14 & 98.7 & -0.01 & 0.11 & 0.14 & 96.9 & -0.01 & 0.14 & 0.14 & 94.5 \\ 
[2ex]
10 & -0.50 & -0.02 & 0.12 & 0.19 & 99.2 & -0.03 & 0.17 & 0.19 & 96.6 & -0.03 & 0.19 & 0.19 & 95.4 \\ 
  & -0.25 & -0.02 & 0.11 & 0.18 & 98.3 & -0.01 & 0.15 & 0.18 & 97.3 & -0.01 & 0.18 & 0.18 & 95.3 \\ 
  & 0.25 & -0.02 & 0.11 & 0.17 & 99.0 & -0.02 & 0.15 & 0.17 & 96.9 & -0.02 & 0.17 & 0.17 & 95.0 \\ 
  & 0.50 & -0.00 & 0.11 & 0.18 & 99.3 & -0.01 & 0.15 & 0.18 & 96.1 & -0.01 & 0.18 & 0.18 & 93.9 \\ 
[2ex]
20 & -0.50 & -0.03 & 0.13 & 0.30 & 99.9 & -0.03 & 0.25 & 0.30 & 97.8 & -0.03 & 0.30 & 0.30 & 95.0 \\ 
  & -0.25 & -0.02 & 0.13 & 0.27 & 99.7 & -0.01 & 0.23 & 0.27 & 97.3 & -0.02 & 0.27 & 0.27 & 94.8 \\ 
  & 0.25 & -0.01 & 0.12 & 0.25 & 99.6 & -0.02 & 0.23 & 0.25 & 96.7 & -0.01 & 0.27 & 0.25 & 94.5 \\ 
  & 0.50 & -0.01 & 0.12 & 0.27 & 99.8 & -0.01 & 0.23 & 0.27 & 96.2 & -0.01 & 0.27 & 0.27 & 93.3 \\ 
\bottomrule
\end{tabular}%
}
\end{table}


\clearpage

\begin{figure}[ht!]
    \centering
    \includegraphics[width=\linewidth]{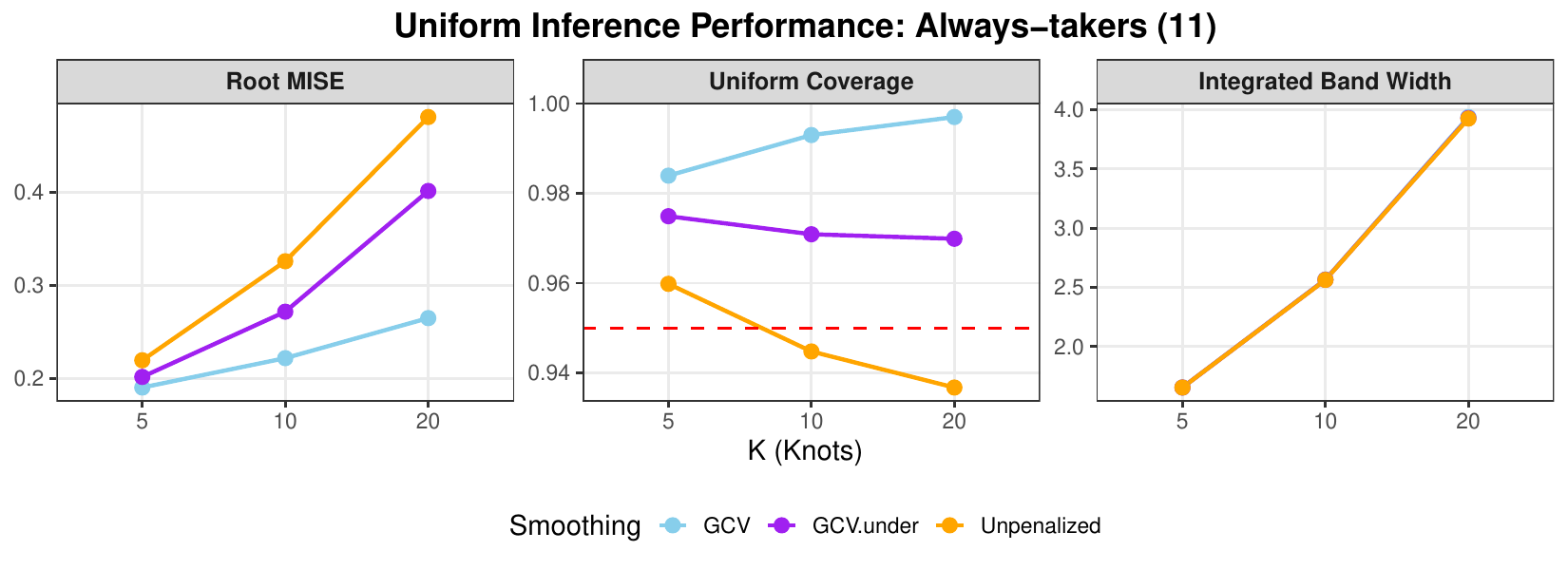}
    \caption{ Simulation results presenting the Root mean integrated squared error (MISE), uniform empirical coverage probability, and average integrated width of the uniform confidence bands for the proposed PLSS-DCDR learner of the CPCE among always-takers when non-monotonicity holds. The basis functions are constructed using P-splines with $K \in \{5, 10, 20\}$ knots. We employ three implementation strategies: a penalized estimator with the smoothing parameter selected via the generalized cross-validation (GCV) criterion \cite{eilers1996flexible}, a variant based on an undersmoothed GCV selection, and an unpenalized estimator corresponding to the least squares series approach.}
    \label{fig:sim-results-11}
\end{figure}

\begin{figure}[ht!]
    \centering
    \includegraphics[width=\linewidth]{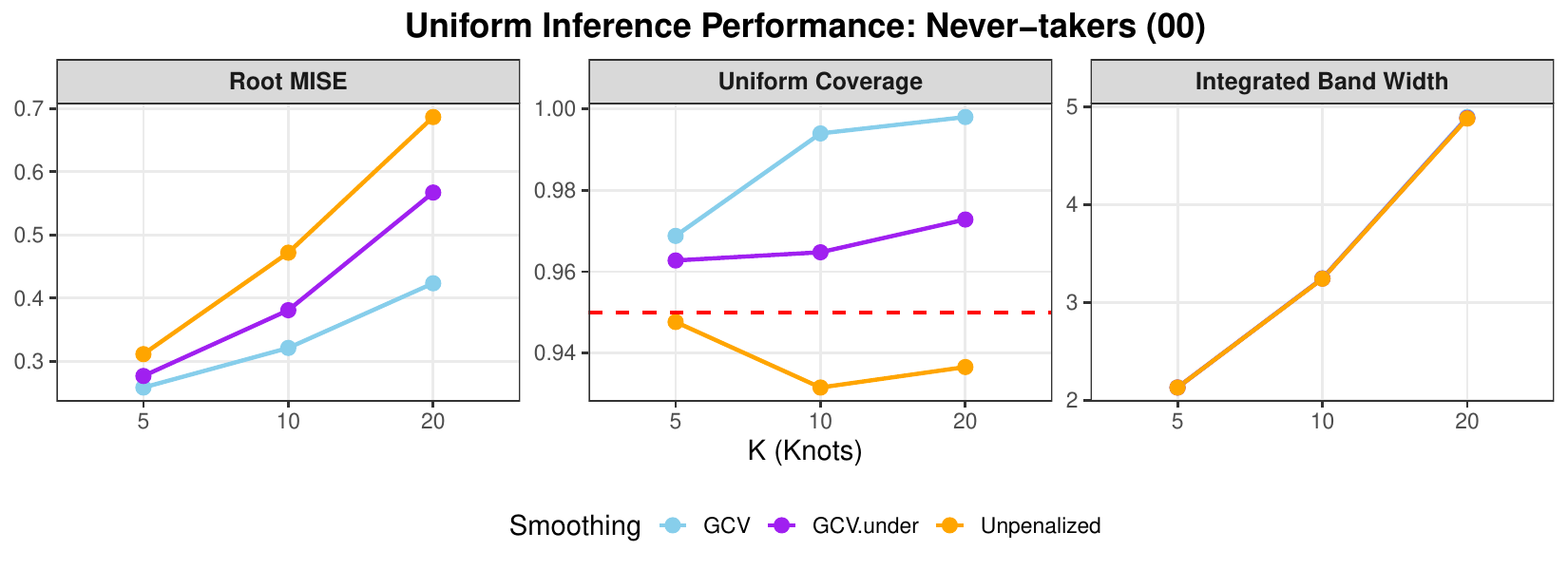}
    \caption{ Simulation results presenting the Root mean integrated squared error (MISE), uniform empirical coverage probability, and average integrated width of the uniform confidence bands for the proposed PLSS-DCDR learner of the CPCE among never-takers when non-monotonicity holds. The basis functions are constructed using P-splines with $K \in \{5, 10, 20\}$ knots. We employ three implementation strategies: a penalized estimator with the smoothing parameter selected via the generalized cross-validation (GCV) criterion \cite{eilers1996flexible}, a variant based on an undersmoothed GCV selection, and an unpenalized estimator corresponding to the least squares series approach.}
    \label{fig:sim-results-00}
\end{figure}

\begin{figure}[ht!]
    \centering
    \includegraphics[width=\linewidth]{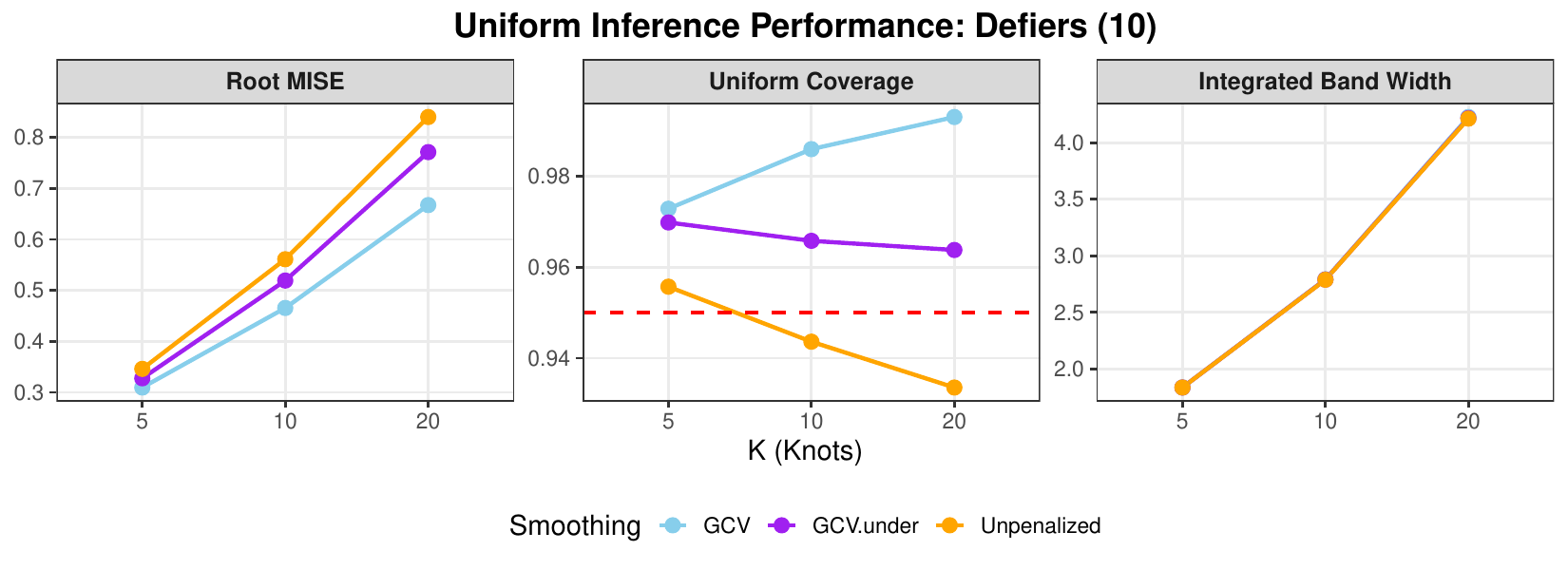}
    \caption{ Simulation results presenting the Root mean integrated squared error (MISE), uniform empirical coverage probability, and average integrated width of the uniform confidence bands for the proposed PLSS-DCDR learner of the CPCE among defiers when non-monotonicity holds. The basis functions are constructed using P-splines with $K \in \{5, 10, 20\}$ knots. We employ three implementation strategies: a penalized estimator with the smoothing parameter selected via the generalized cross-validation (GCV) criterion \cite{eilers1996flexible}, a variant based on an undersmoothed GCV selection, and an unpenalized estimator corresponding to the least squares series approach.}
    \label{fig:sim-results-10}
\end{figure}

\begin{figure}[ht!]
    \centering
    \includegraphics[width=0.5\linewidth]{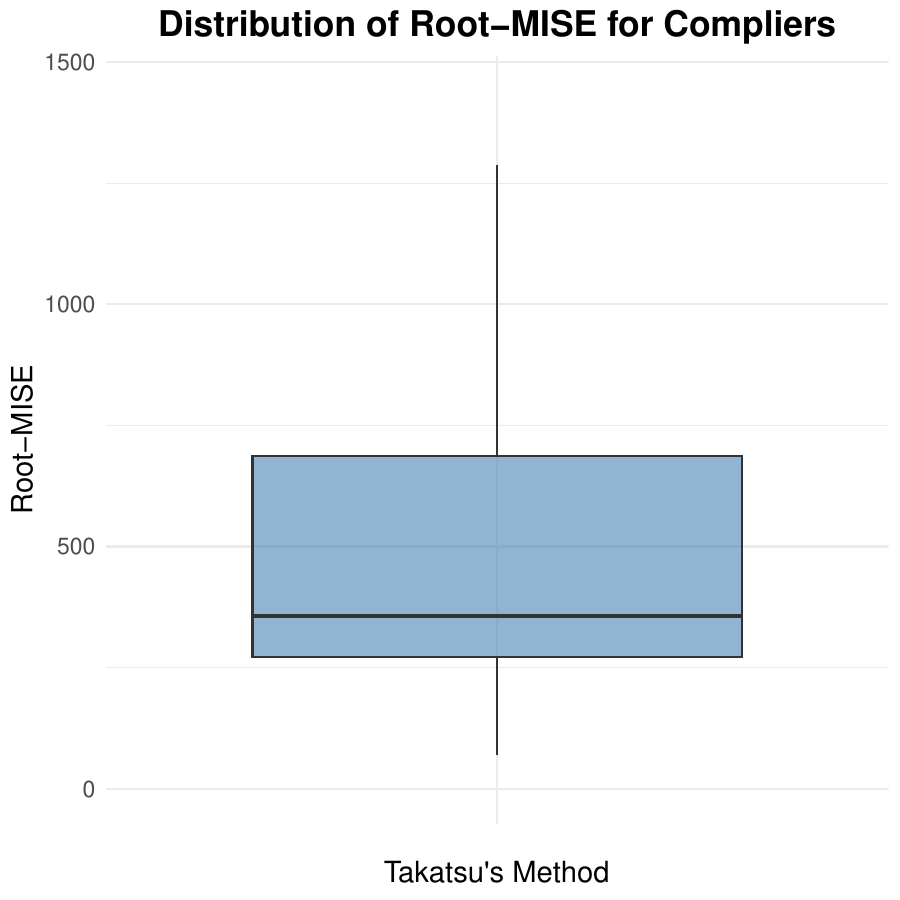}
    \caption{Simulation results presenting box plots of the Root mean integrated squared error (MISE) for the estimators of the CLATE proposed by \cite{takatsu2025doubly} when non-monotonicity holds. Both the numerator and denominator DR-learners were implemented using the Super Learner \citep{van2007super}.}
    \label{fig:sim-results-JRSSA-uniform}
\end{figure}

\begin{figure}[ht!]
    \centering
    \includegraphics[width=0.8\linewidth]{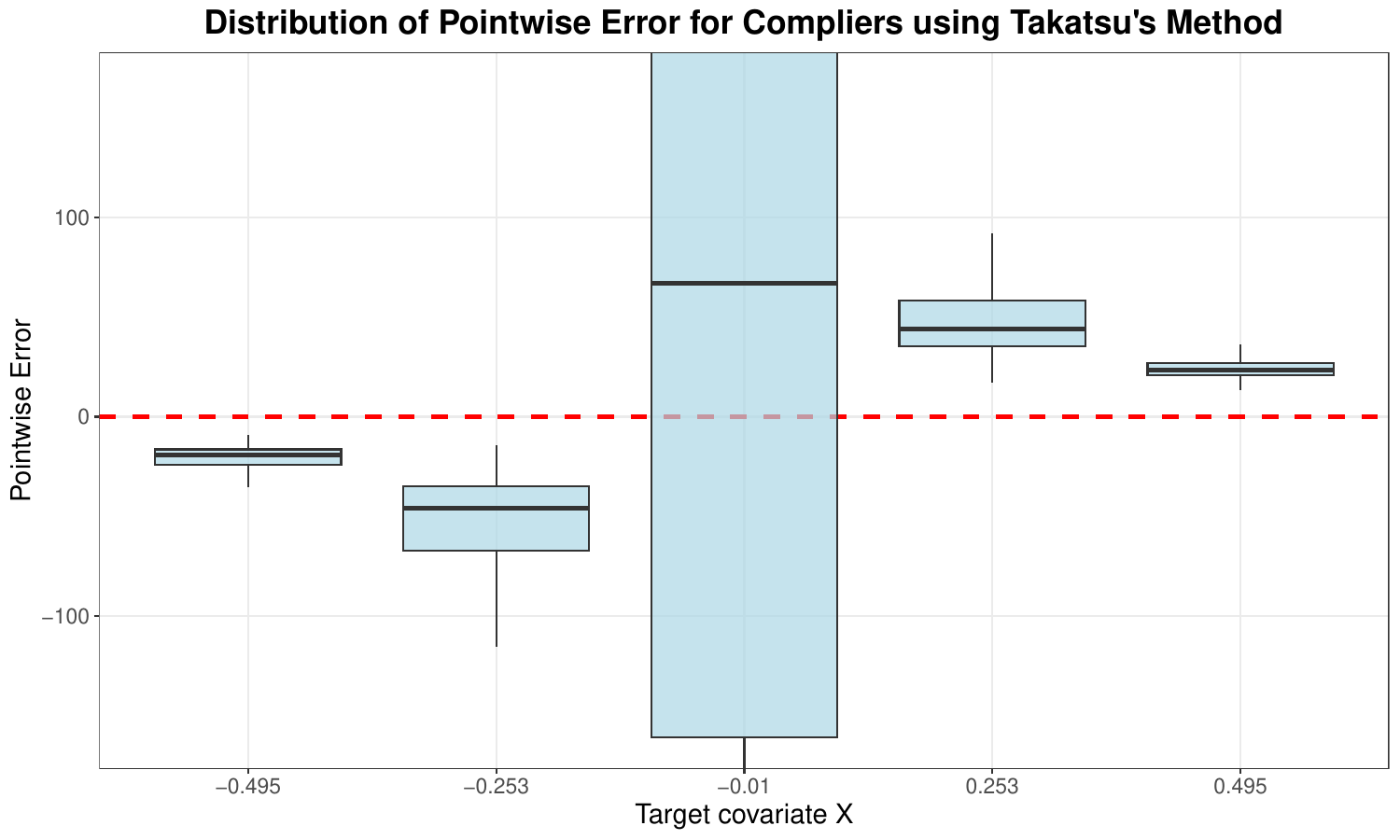}
    \caption{ Simulation results displaying box plots of the pointwise estimation errors for the CLATE proposed by \cite{takatsu2025doubly}, evaluated at five representative points: $X\in\{-0.495, -0.253, -0.01, 0.253, 0.495\}$ when non-monotonicity holds. Both the numerator and denominator DR-learners were implemented using the Super Learner \citep{van2007super}.}
    \label{fig:sim-results-JRSSA-pointwise}
\end{figure}

\begin{figure}[ht!]
    \centering
    \includegraphics[width=0.5\linewidth]{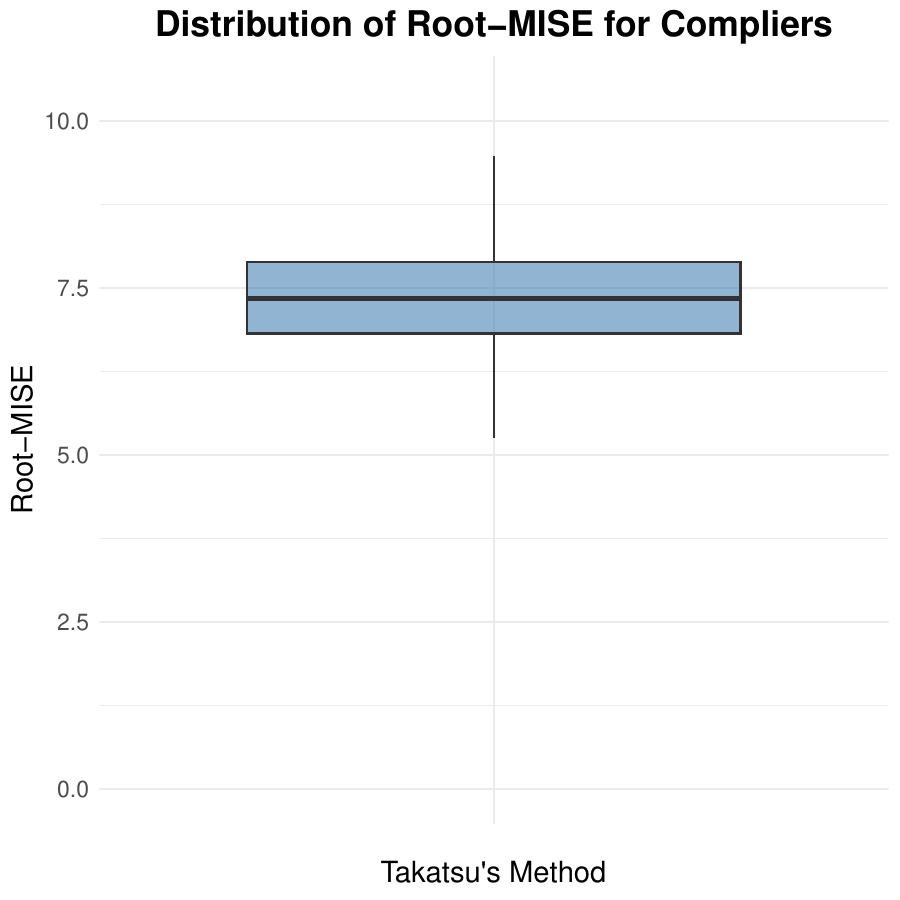}
    \caption{Simulation results presenting box plots of the Root mean integrated squared error (MISE) for the estimators of the CLATE proposed by \cite{takatsu2025doubly} when monotonicity holds. Both the numerator and denominator DR-learners were implemented using the Super Learner \citep{van2007super}.}
    \label{fig:sim-results-JRSSA-uniform-mono}
\end{figure}

\begin{figure}[ht!]
    \centering
    \includegraphics[width=0.8\linewidth]{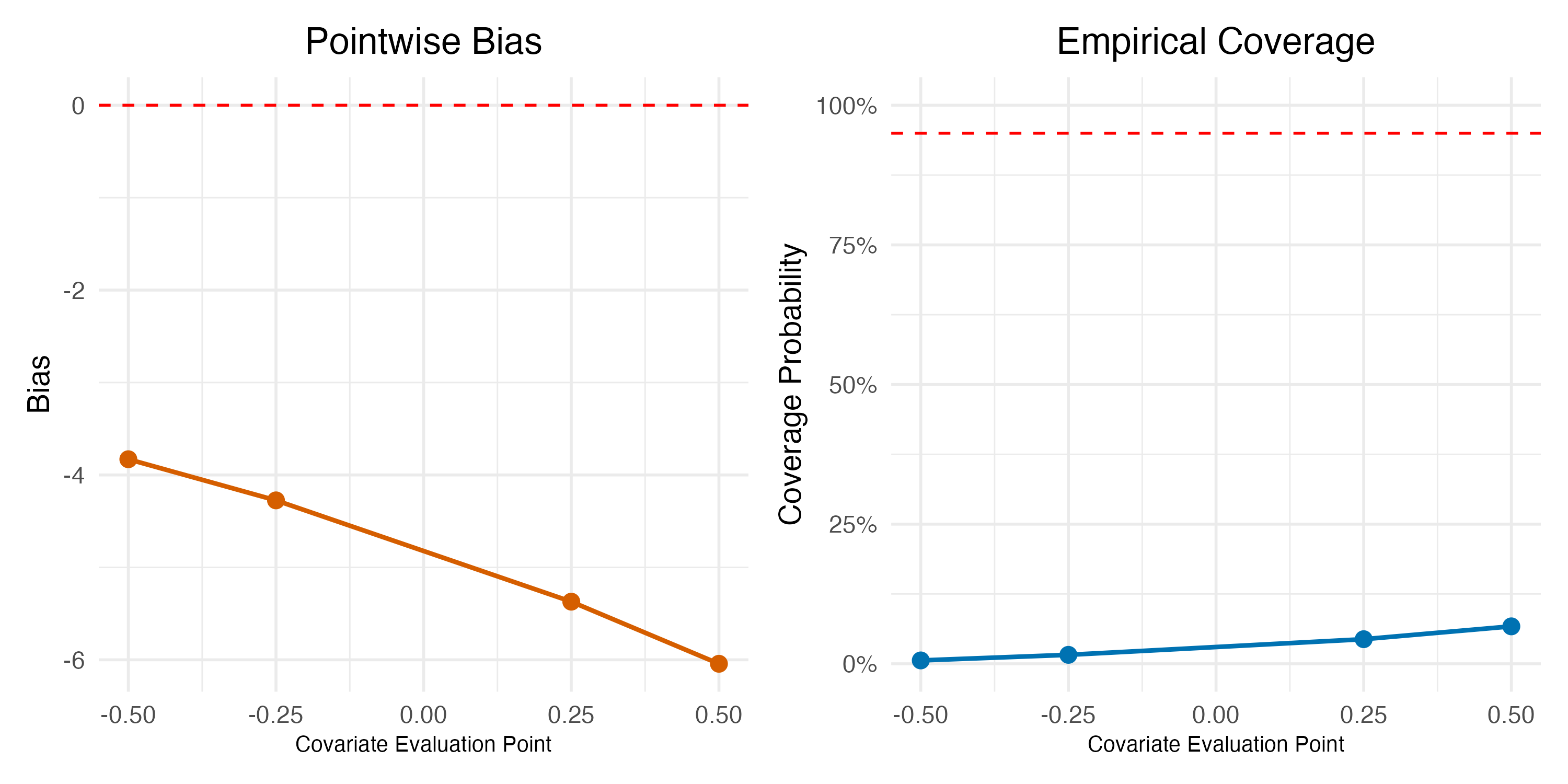}
    \caption{ Simulation results displaying the pointwise bias and empirical pointwise coverage probabilities of the CLATE estimator proposed by \cite{takatsu2025doubly}, evaluated at four representative points $X \in \{-0.5, -0.25, 0.25, 0.50\}$ under the assumption that monotonicity holds. Both the numerator and denominator DR-learners were implemented using the Super Learner \citep{van2007super}.}
    \label{fig:sim-results-JRSSA-pointwise-mono}
\end{figure}

\begin{figure}[ht!]
    \centering
    \includegraphics[width=\linewidth]{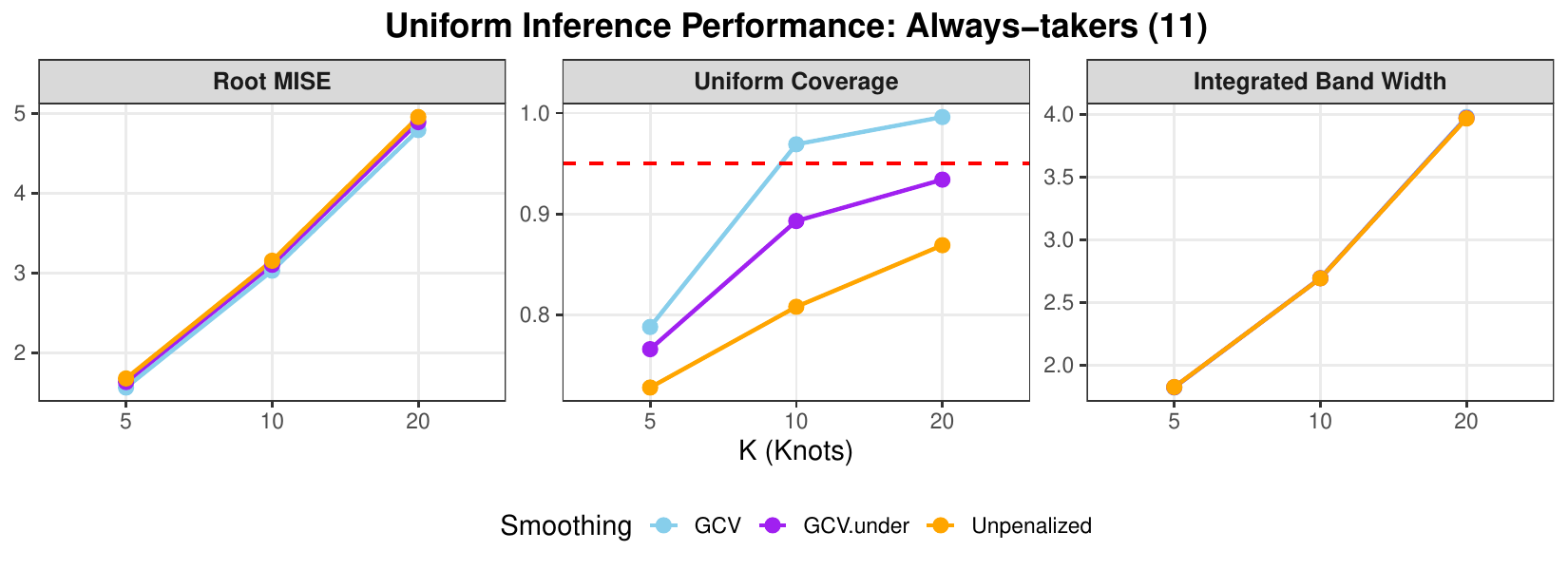}
    \caption{ Simulation results presenting the Root mean integrated squared error (MISE), uniform empirical coverage probability, and average integrated width of the uniform confidence bands for the proposed PLSS-DCDR learner of the CPCE among always-takers, when the true data-generating process follows non-monotonicity but the model is fitted under monotonicity. The basis functions are constructed using P-splines with $K \in \{5, 10, 20\}$ knots. We employ three implementation strategies: a penalized estimator with the smoothing parameter selected via the generalized cross-validation (GCV) criterion \cite{eilers1996flexible}, a variant based on an undersmoothed GCV selection, and an unpenalized estimator corresponding to the least squares series approach.}
    \label{fig:sim-results-nonmono-fitmono-11}
\end{figure}

\begin{figure}[ht!]
    \centering
    \includegraphics[width=\linewidth]{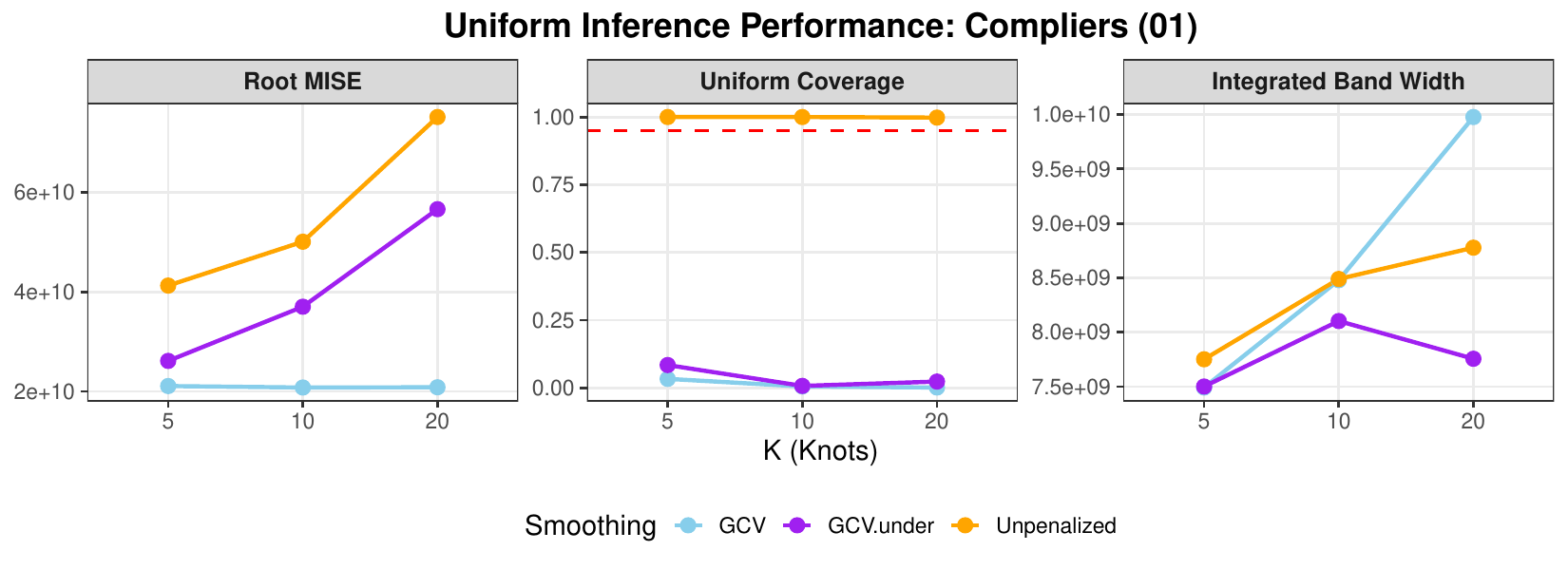}
    \caption{ Simulation results presenting the Root mean integrated squared error (MISE), uniform empirical coverage probability, and average integrated width of the uniform confidence bands for the proposed PLSS-DCDR learner of the CPCE among compliers, when the true data-generating process follows non-monotonicity but the model is fitted under monotonicity. The basis functions are constructed using P-splines with $K \in \{5, 10, 20\}$ knots. We employ three implementation strategies: a penalized estimator with the smoothing parameter selected via the generalized cross-validation (GCV) criterion \cite{eilers1996flexible}, a variant based on an undersmoothed GCV selection, and an unpenalized estimator corresponding to the least squares series approach.}
    \label{fig:sim-results-nonmono-fitmono-01}
\end{figure}

\begin{figure}[ht!]
    \centering
    \includegraphics[width=\linewidth]{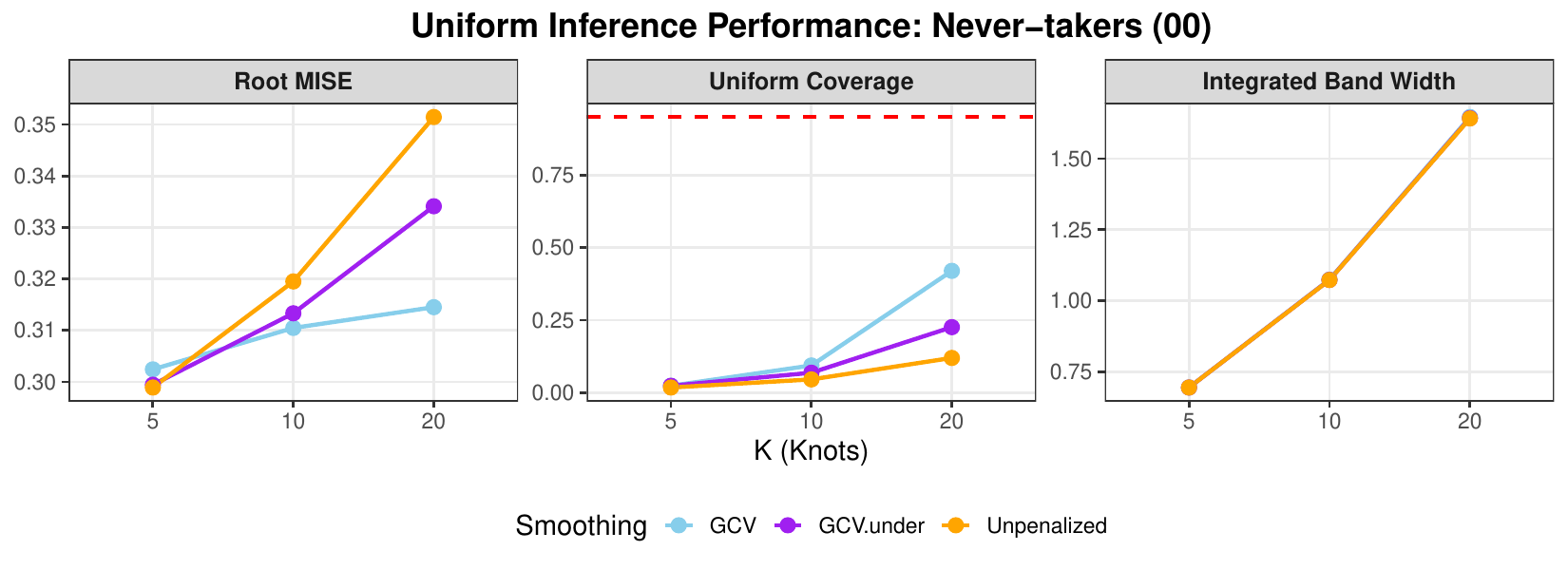}
    \caption{ Simulation results presenting the Root mean integrated squared error (MISE), uniform empirical coverage probability, and average integrated width of the uniform confidence bands for the proposed PLSS-DCDR learner of the CPCE among never-takers, when the true data-generating process follows non-monotonicity but the model is fitted under monotonicity. The basis functions are constructed using P-splines with $K \in \{5, 10, 20\}$ knots. We employ three implementation strategies: a penalized estimator with the smoothing parameter selected via the generalized cross-validation (GCV) criterion \cite{eilers1996flexible}, a variant based on an undersmoothed GCV selection, and an unpenalized estimator corresponding to the least squares series approach.}
    \label{fig:sim-results-nonmono-fitmono-00}
\end{figure}

\begin{figure}[ht!]
    \centering
    \includegraphics[width=\linewidth]{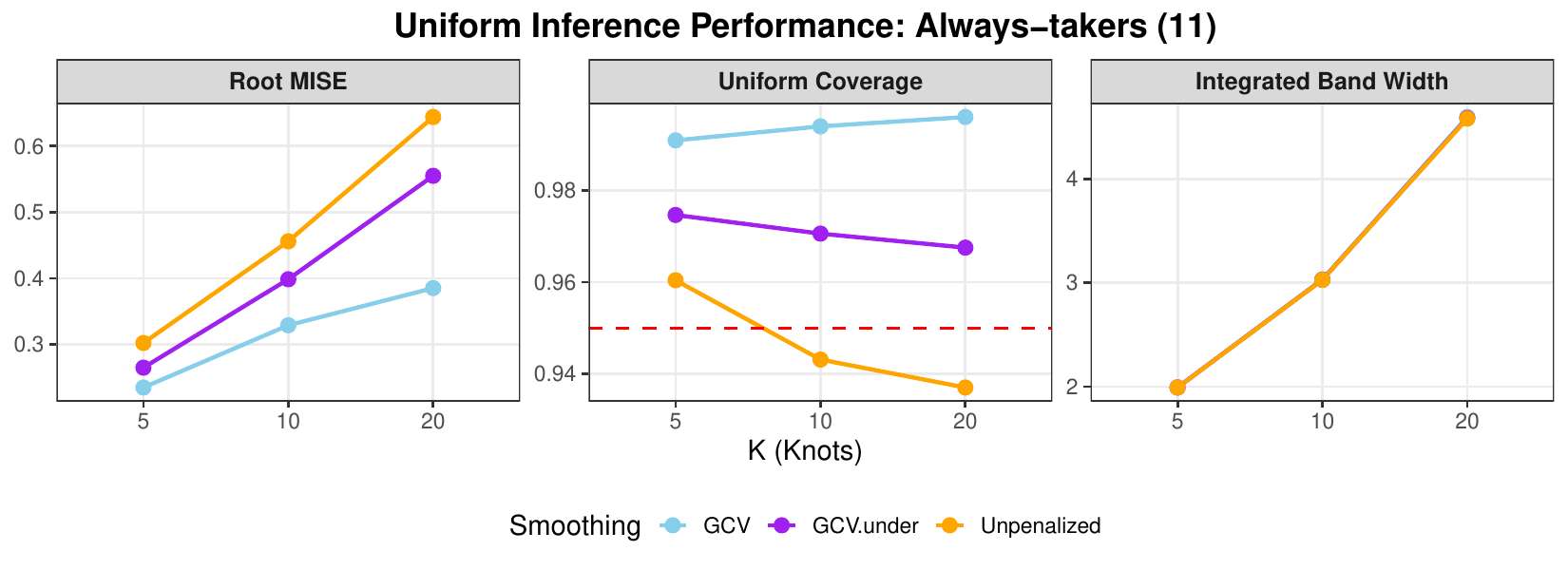}
    \caption{ Simulation results presenting the Root mean integrated squared error (MISE), uniform empirical coverage probability, and average integrated width of the uniform confidence bands for the proposed PLSS-DCDR learner of the CPCE among always-takers when monotonicity holds. The basis functions are constructed using P-splines with $K \in \{5, 10, 20\}$ knots. We employ three implementation strategies: a penalized estimator with the smoothing parameter selected via the generalized cross-validation (GCV) criterion \cite{eilers1996flexible}, a variant based on an undersmoothed GCV selection, and an unpenalized estimator corresponding to the least squares series approach.}
    \label{fig:sim-mono-fitmono-11}
\end{figure}

\begin{figure}[ht!]
    \centering
    \includegraphics[width=\linewidth]{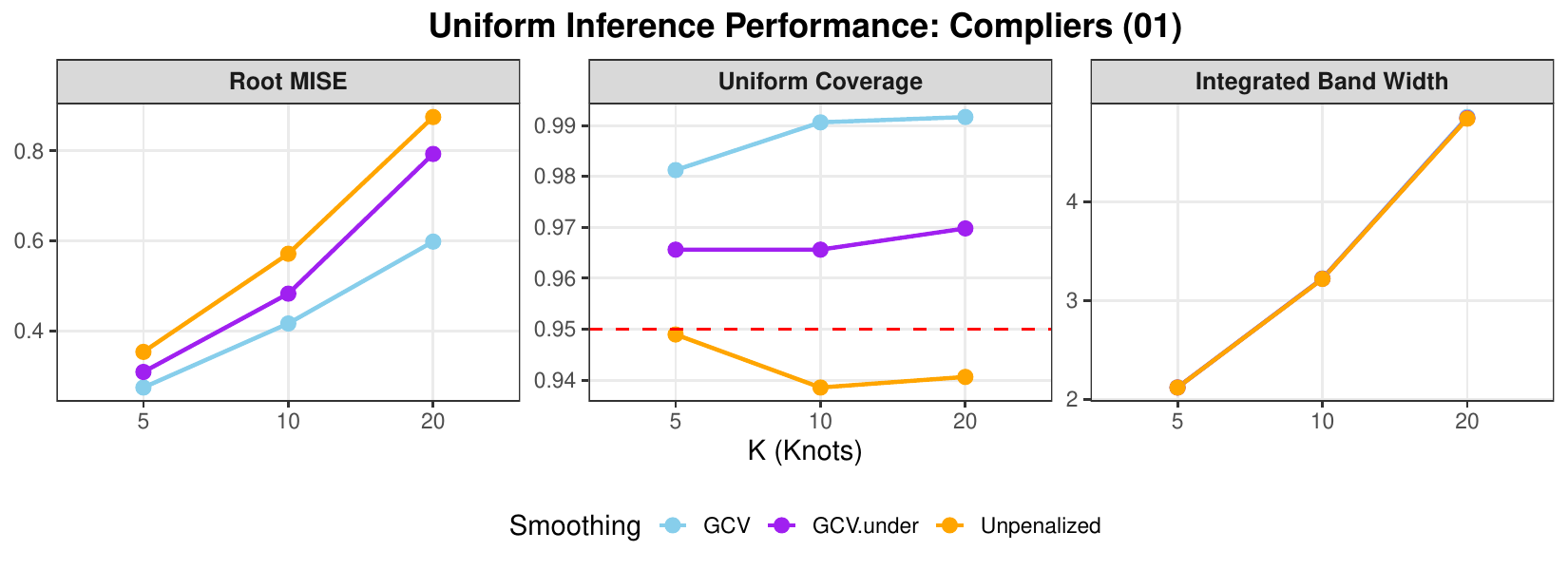}
    \caption{ Simulation results presenting the Root mean integrated squared error (MISE), uniform empirical coverage probability, and average integrated width of the uniform confidence bands for the proposed PLSS-DCDR learner of the CPCE among compliers when monotonicity holds. The basis functions are constructed using P-splines with $K \in \{5, 10, 20\}$ knots. We employ three implementation strategies: a penalized estimator with the smoothing parameter selected via the generalized cross-validation (GCV) criterion \cite{eilers1996flexible}, a variant based on an undersmoothed GCV selection, and an unpenalized estimator corresponding to the least squares series approach.}
    \label{fig:sim-mono-fitmono-01}
\end{figure}

\begin{figure}[ht!]
    \centering
    \includegraphics[width=\linewidth]{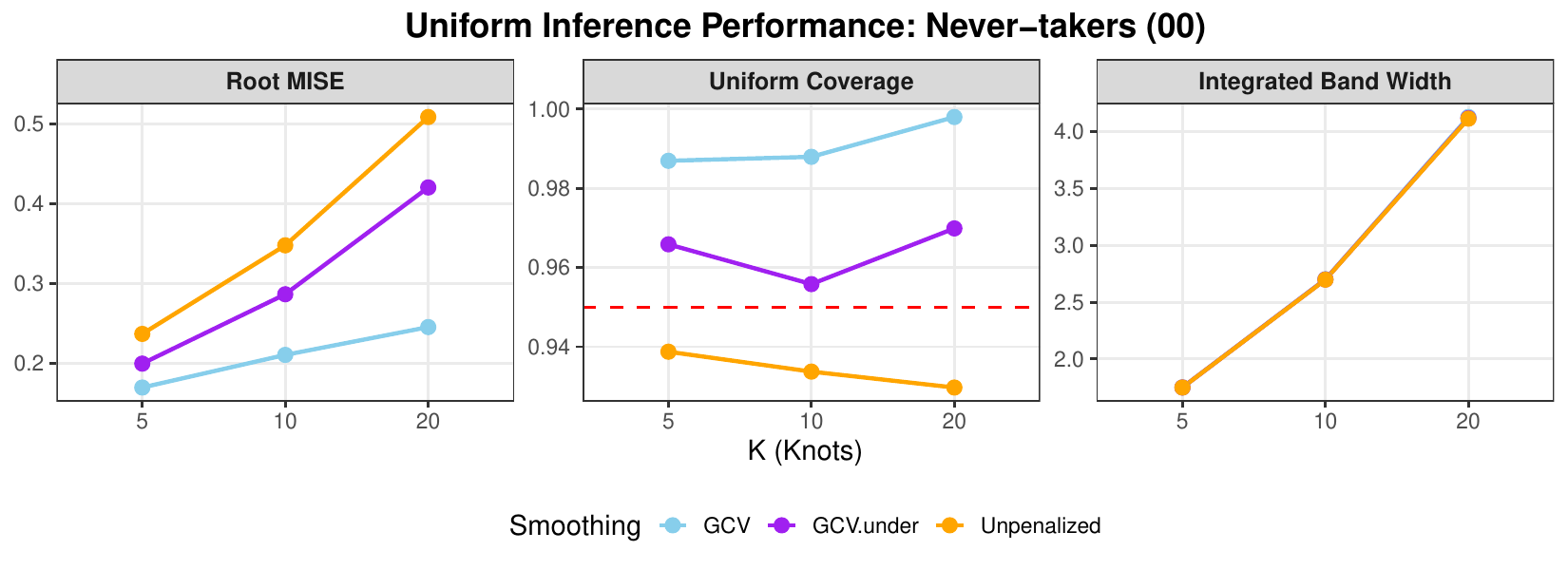}
    \caption{ Simulation results presenting the Root mean integrated squared error (MISE), uniform empirical coverage probability, and average integrated width of the uniform confidence bands for the proposed PLSS-DCDR learner of the CPCE among never-takers when monotonicity holds. The basis functions are constructed using P-splines with $K \in \{5, 10, 20\}$ knots. We employ three implementation strategies: a penalized estimator with the smoothing parameter selected via the generalized cross-validation (GCV) criterion \cite{eilers1996flexible}, a variant based on an undersmoothed GCV selection, and an unpenalized estimator corresponding to the least squares series approach.}
    \label{fig:sim-mono-fitmono-00}
\end{figure}


\begin{figure}[ht!]
    \centering
    \includegraphics[width=\linewidth]{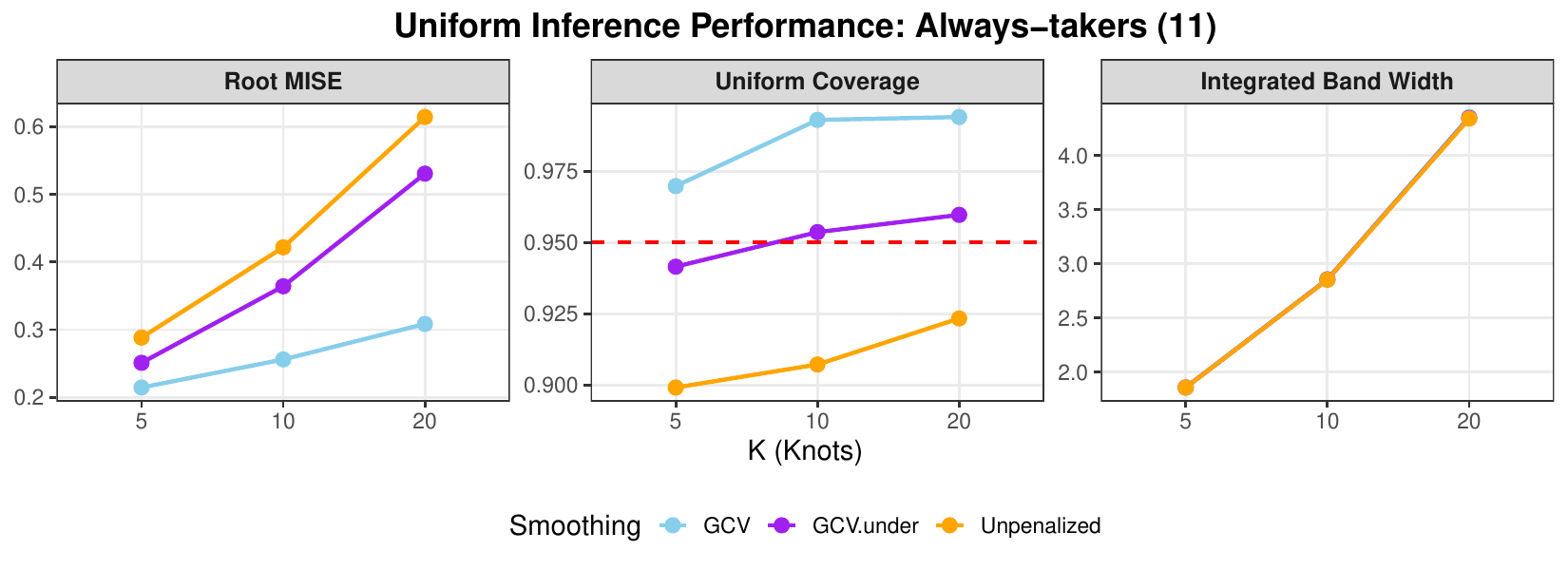}
    \caption{ Simulation results presenting the Root mean integrated squared error (MISE), uniform empirical coverage probability, and average integrated width of the uniform confidence bands for the proposed PLSS-DCDR learner of the CPCE among always-takers, when the true data-generating process follows monotonicity but the model is fitted under non-monotonicity with $\theta=5$. The basis functions are constructed using P-splines with $K \in \{5, 10, 20\}$ knots. We employ three implementation strategies: a penalized estimator with the smoothing parameter selected via the generalized cross-validation (GCV) criterion \cite{eilers1996flexible}, a variant based on an undersmoothed GCV selection, and an unpenalized estimator corresponding to the least squares series approach.}
    \label{fig:sim-mono-fit5-11}
\end{figure}

\begin{figure}[ht!]
    \centering
    \includegraphics[width=\linewidth]{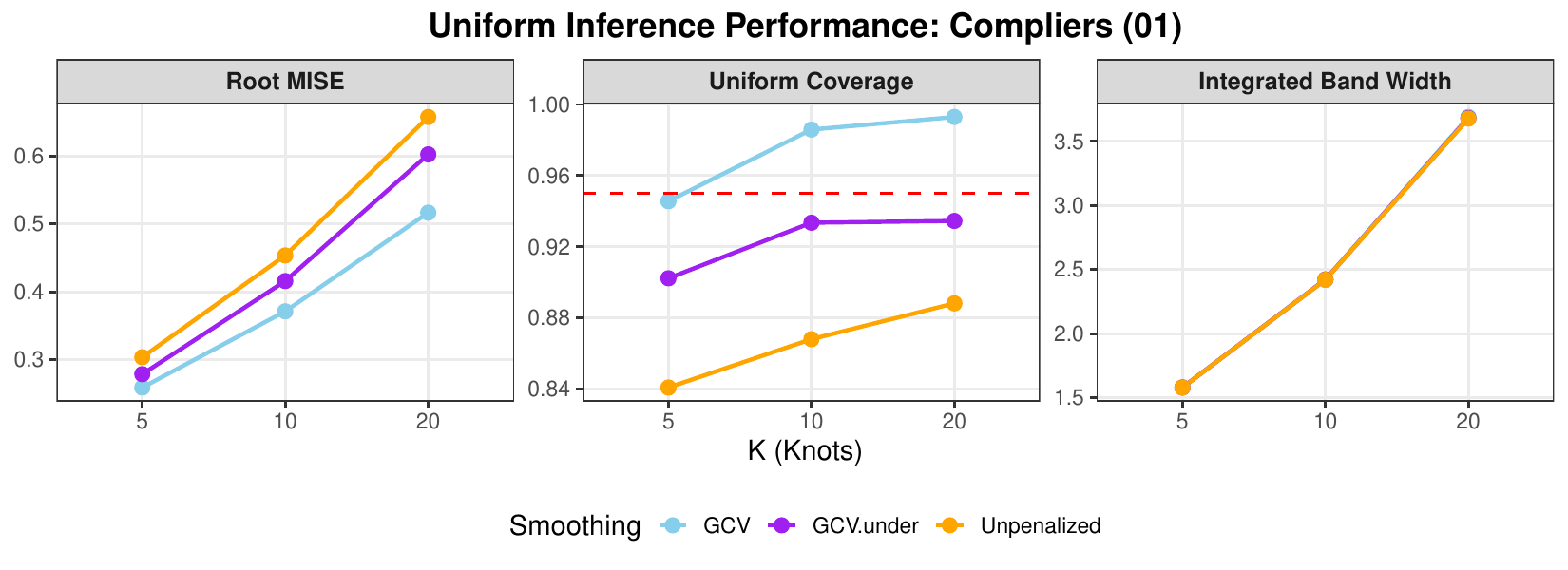}
    \caption{ Simulation results presenting the Root mean integrated squared error (MISE), uniform empirical coverage probability, and average integrated width of the uniform confidence bands for the proposed PLSS-DCDR learner of the CPCE among compliers, when the true data-generating process follows monotonicity but the model is fitted under non-monotonicity with $\theta=5$. The basis functions are constructed using P-splines with $K \in \{5, 10, 20\}$ knots. We employ three implementation strategies: a penalized estimator with the smoothing parameter selected via the generalized cross-validation (GCV) criterion \cite{eilers1996flexible}, a variant based on an undersmoothed GCV selection, and an unpenalized estimator corresponding to the least squares series approach.}
    \label{fig:sim-mono-fit5-01}
\end{figure}

\begin{figure}[ht!]
    \centering
    \includegraphics[width=\linewidth]{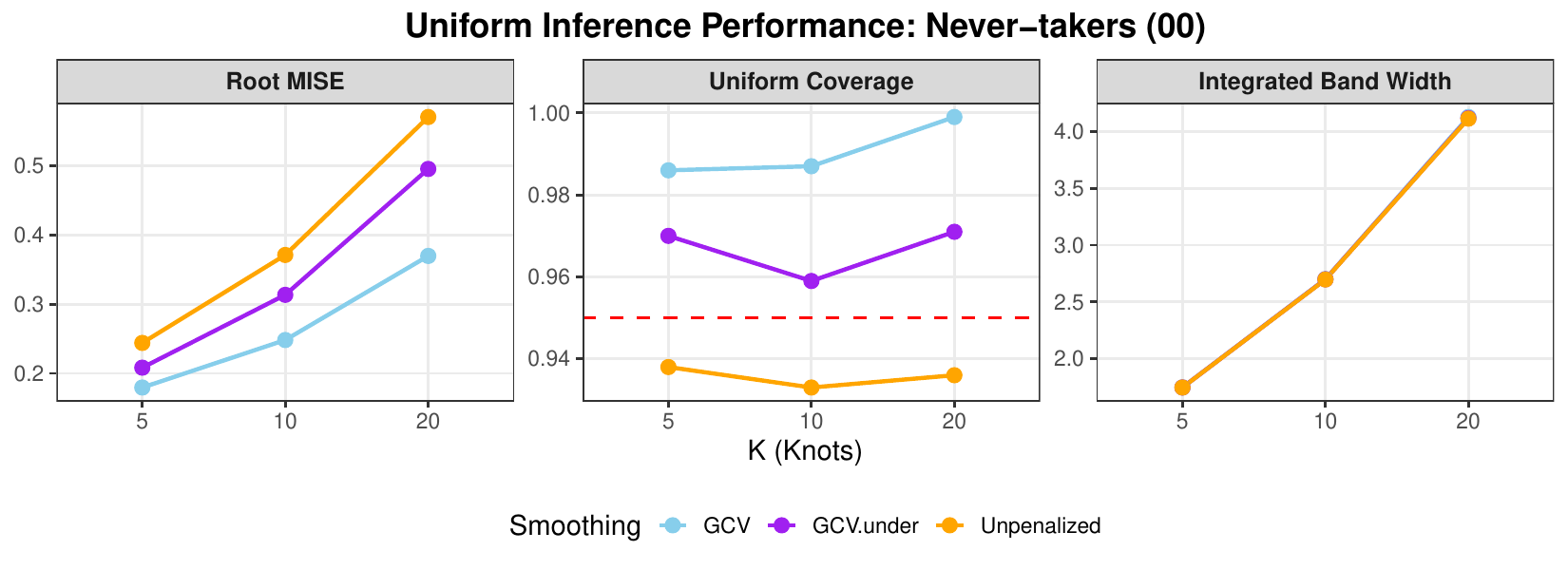}
    \caption{ Simulation results presenting the Root mean integrated squared error (MISE), uniform empirical coverage probability, and average integrated width of the uniform confidence bands for the proposed PLSS-DCDR learner of the CPCE among never-takers, when the true data-generating process follows monotonicity but the model is fitted under non-monotonicity with $\theta=5$. The basis functions are constructed using P-splines with $K \in \{5, 10, 20\}$ knots. We employ three implementation strategies: a penalized estimator with the smoothing parameter selected via the generalized cross-validation (GCV) criterion \cite{eilers1996flexible}, a variant based on an undersmoothed GCV selection, and an unpenalized estimator corresponding to the least squares series approach.}
    \label{fig:sim-mono-fit5-00}
\end{figure}

\begin{figure}[htbp]
    \centering
    \includegraphics[width=\linewidth]{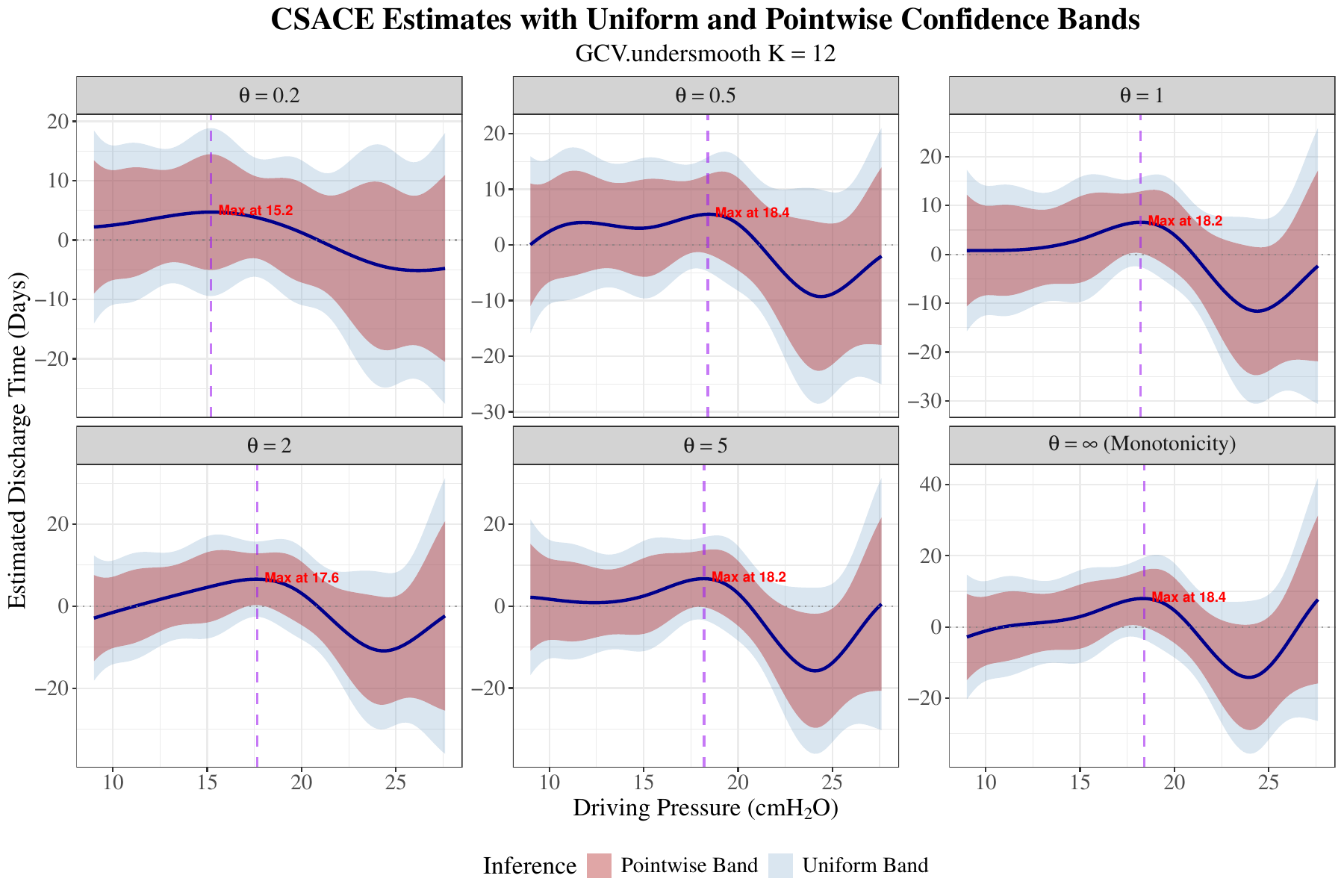}
    \caption{The estimated curves and associated 95\% pointwise and uniform confidence bands for the PLSS TR-learner, utilizing P-spline smoothing with $K=12$ knots, are presented for the CSACE in the ARDS study across six values of the conditional odds ratio $\theta\in\{0.2,0.5,1,2,5,\infty\}$ (where $\theta=\infty$ corresponds to monotonicity). The vertical purple dashed lines indicate the maximizer or turning point of the estimated CSACE curves.}
    \label{fig:ARDS-CSACE-12}
\end{figure}
\end{spacing} 
\clearpage

\bibliography{bibliography.bib}

\bibliographystyle{chicago}

\end{document}